\documentclass[11pt]{article}

\usepackage[a4paper,left=20mm,top=20mm,right=20mm,bottom=25mm]{geometry}

\usepackage{amsmath, amsfonts, amssymb, amsthm}
\usepackage[noend]{algorithmic}
\usepackage[linesnumbered,ruled,vlined]{algorithm2e}

\usepackage{mathtools}
\usepackage[mathscr]{euscript}
\usepackage{wasysym}

\usepackage{graphicx}
\usepackage{xcolor}
\usepackage{enumitem}
\usepackage{authblk}
\usepackage{float}

\usepackage{natbib}
\setcitestyle{authoryear,open={(},close={)}}
\let\cite\citep

\usepackage{color}

\usepackage[colorlinks]{hyperref}
\hypersetup{
  citecolor=blue,
  linkcolor=blue,
}

\usepackage[font=footnotesize,width=.85\textwidth,labelfont=bf]{caption}

\SetAlCapNameFnt{\small}
\SetAlCapFnt{\small}
\usepackage{float}

\makeatletter
\def\moverlay{\mathpalette\mov@rlay}
\def\mov@rlay#1#2{\leavevmode\vtop{%
    \baselineskip\z@skip \lineskiplimit-\maxdimen
    \ialign{\hfil$\m@th#1##$\hfil\cr#2\crcr}}}
\newcommand{\charfusion}[3][\mathord]{
  #1{\ifx#1\mathop\vphantom{#2}\fi
    \mathpalette\mov@rlay{#2\cr#3}
  }
  \ifx#1\mathop\expandafter\displaylimits\fi}
\DeclareRobustCommand\bigop[1]{%
  \mathop{\vphantom{\sum}\mathpalette\bigop@{#1}}\slimits@
}
\newcommand{\bigop@}[2]{%
  \vcenter{%
    \sbox\z@{$#1\sum$}%
    \hbox{\resizebox{\ifx#1\displaystyle.9\fi\dimexpr\ht\z@+\dp\z@}{!}{$\m@th#2$}}%
  }%
}
\makeatother

\usepackage{hyperref}
\hypersetup{
    colorlinks=true,       
    linkcolor=black,          
    citecolor=blue,        
}

\newcommand{\PROBLEM}[1]{{\sc #1}}
\newtheorem{xmpl}{Example}{\bfseries\itshape}{\slshape}
\newtheorem{Xclaim}{Claim}{\itshape}{\rmfamily}
\newtheorem{fact}{Observation}{\bfseries}{\itshape}
\newenvironment{claim-proof}%
{\begin{description}[leftmargin = 0.2cm, labelsep = 0.2cm]
\item[\emph{Proof:}]}{\hfill$\diamond$\end{description}}

\DeclareMathOperator*{\bigjoin}{\bigop{\triangledown}}
\DeclareMathOperator{\join}{\triangledown}
\newcommand{\cupdot}{\charfusion[\mathbin]{\cup}{\cdot}}
\DeclareMathOperator*{\bigcupdot}{\charfusion[\mathop]{\bigcup}{\cdot}}
\DeclareMathOperator{\child}{child}
\DeclareMathOperator{\parent}{par}
\DeclareMathOperator{\lca}{lca}
\DeclareMathOperator{\Aho}{Aho}
\newcommand{\Tri}{\ensuremath{\mathfrak{T}}}
\newcommand{\Sri}{\ensuremath{\mathfrak{S}}}
\newcommand{\mfscen}{\ensuremath{\mathscr{T}}}
\newcommand{\scen}{\ensuremath{\mathscr{S}}}
\newcommand{\rel}{\ensuremath{\mathfrak{R}}}
\newcommand{\partS}{\mathcal{C}_{S}}
\newcommand{\tT}{\ensuremath{\tau_{T}}}
\newcommand{\tS}{\ensuremath{\tau_{S}}}
\newcommand{\Gu}{\ensuremath{G_{_{<}}}} 
\newcommand{\gfitch}{\digamma}
\newcommand{\auxfitch}{\mathscr{A}_{\gfitch}} 

\newtheorem{innercustomthm}{Theorem}
\newenvironment{ctheorem}[1]
  {\renewcommand\theinnercustomthm{#1}\innercustomthm}
  {\endinnercustomthm}

\newtheorem{innercustomlem}{Lemma}
\newenvironment{clemma}[1]
  {\renewcommand\theinnercustomlem{#1}\innercustomlem}
  {\endinnercustomlem}

\newtheorem{innercustomdef}{Definition}
\newenvironment{cdefinition}[1]
  {\renewcommand\theinnercustomdef{#1}\innercustomdef}
  {\endinnercustomdef}

\newtheorem{innercustomprop}{Proposition}
\newenvironment{cproposition}[1]
  {\renewcommand\theinnercustomprop{#1}\innercustomprop}
  {\endinnercustomprop}

\newtheorem{innercustomcor}{Corollary}
\newenvironment{ccorollary}[1]
  {\renewcommand\theinnercustomcor{#1}\innercustomcor}
  {\endinnercustomcor}

\newtheorem{innercustomfact}{Observation}
\newenvironment{cfact}[1]
  {\renewcommand\theinnercustomfact{#1}\innercustomfact}
  {\endinnercustomfact}

\newtheorem{theorem}{Theorem}
\newtheorem{lemma}{Lemma}
\newtheorem{corollary}{Corollary}
\newtheorem{definition}{Definition}
\newtheorem{proposition}{Proposition}
\newtheorem{problem}{Problem}
\newtheorem{remark}{Remark}

\providecommand{\keywords}[1]{\textbf{\textit{Keywords: }} #1}

\title{Indirect Identification of Horizontal Gene Transfer}

\author[1,2]{David Schaller}
\author[3]{Manuel Lafond} 
\author[1,2,4-7]{Peter F. Stadler}    
\author[8]{Nicolas Wieseke} 
\author[9,*]{Marc Hellmuth}

\affil[1]{Max Planck Institute for Mathematics in the Sciences,
  Inselstra{\ss}e 22, D-04103 Leipzig, Germany}

\affil[2]{Bioinformatics Group, Department of Computer Science \&
  Interdisciplinary Center for Bioinformatics, Universit{\"a}t Leipzig,
  H{\"a}rtelstra{\ss}e~16--18, D-04107 Leipzig, Germany.}

\affil[3]{  Department of Computer Science, Universit{\'e} de Sherbrooke,
  2500 boul. de l'Universit{\'e} Sherbrooke, Qc, Canada, J1K 2R1}

\affil[4]{German Centre for Integrative Biodiversity Research
  (iDiv) Halle-Jena-Leipzig, Competence Center for Scalable Data Services
  and Solutions Dresden-Leipzig, Leipzig Research Center for Civilization
  Diseases, and Centre for Biotechnology and Biomedicine at Leipzig
  University at Universit{\"a}t Leipzig}

\affil[5]{Institute for Theoretical Chemistry, University of Vienna,
  W{\"a}hringerstrasse 17, A-1090 Wien, Austria}

\affil[6]{Facultad de Ciencias, Universidad National de Colombia, Sede
  Bogot{\'a}, Colombia}

\affil[7]{Santa Fe Insitute, 1399 Hyde Park Rd., Santa Fe NM 87501,
  USA}

\affil[8]{Swarm Intelligence and Complex Systems Group, 
  Department of Computer Science, Leipzig University,
  Augustusplatz 10, D-04109 Leipzig, Germany}

\affil[9]{Department of Mathematics, Faculty of Science, Stockholm University,
  SE - 106 91 Stockholm,   Sweden \newline \texttt{marc.hellmuth@math.su.se}}

\affil[*]{corresponding author}

\date{\ }

\begin{document}

\maketitle 

\abstract{  
Several implicit methods to infer Horizontal Gene Transfer (HGT) focus on
pairs of genes that have diverged only after the divergence of the two
species in which the genes reside. This situation defines the edge set of
a graph, the later-divergence-time (LDT) graph, whose vertices correspond
to genes colored by their species. We investigate these graphs in the
setting of relaxed scenarios, i.e., evolutionary scenarios that
encompass all commonly used variants of duplication-transfer-loss
scenarios in the literature. We characterize LDT graphs as a subclass of
properly vertex-colored cographs, and provide a polynomial-time
recognition algorithm as well as an algorithm to construct a relaxed
scenario that explains a given LDT. An edge in an LDT graph implies that
the two corresponding genes are separated by at least one HGT event. The
converse is not true, however. We show that the complete xenology
relation is described by an rs-Fitch graph, i.e., a complete multipartite
graph satisfying constraints on the vertex coloring. This class of
vertex-colored graphs is also recognizable in polynomial time. We finally
address the question ``how much information about all HGT events is
contained in LDT graphs'' with the help of simulations of evolutionary
scenarios with a wide range of duplication, loss, and HGT events. In
particular, we show that a simple greedy graph editing scheme can be used
to efficiently detect HGT events that are implicitly contained in LDT
graphs.
}

\bigskip
\noindent

\sloppy

\section{Introduction}

Horizontal gene transfer (HGT) laterally introduces foreign genetic
material into a genome. The phenomenon is particularly frequent in
prokaryotes \cite{Soucy:15,NelsonSathi:15} but also contributed to shaping
eukaryotic genomes
\cite{Keeling:08,Husnik:18,Acuna:12,Li:14,Moran:10,Schonknecht:13}. HGT may
be additive, in which case its effect is similar to gene duplications, or
lead to the replacement of a vertically inherited homolog. From a
phylogenetic perspective, HGT leads to an incongruence of gene trees and
species trees, thus complicating the analysis of gene family histories.

A broad spectrum of computational methods have been developed to identify
horizontally transferred genes and/or HGT events, recently reviewed by
\citet{Ravenhall:15}.  Parametric methods use genomic signatures, i.e.,
sequence features specific to a (group of) species identify horizontally
inserted material. Genomic signatures include e.g.\ GC content, $k$-mer
distributions, sequence autocorrelation, or DNA deformability
\cite{Dufraigne:05,Becq:10}.  Direct (or ``explicit'') phylogenetic methods
start from a given gene tree $T$ and species tree $S$ and compute a
reconciliation, i.e., a mapping of the gene tree into the species
tree. This problem first arose in the context of host/parasite assemblages
\cite{Page:94,Charleston:98} considering the equivalent problem of mapping
a parasite tree $T$ to a host phylogeny $S$ such that the number of
events such as host-switches, i.e., horizontal transfers, is minimized. For
a review of the early literature we refer to \cite{Charleston:06}. A major
difficulty is to enforce time consistency in the presence of multiple
horizontal transfer events, which renders the problem of finding optimal
reconciliations NP-hard \cite{Hallett:01,Ovadia:11,Tofigh:11,Hasic:19}.
Nevertheless several practical approaches have become available, see e.g.\
\cite{Tofigh:11,Chen:12,Ma:18}.

Indirect (or ``implicit'') phylogenetic methods forego the reconstruction
of trees and start from sequence similarity or evolutionary distances and
use unexpectedly small or large distances between genes as indicators of
HGT. While indirect methods have been used successfully in the past,
reviewed by \citet{Ravenhall:15}, they have received very little attention
from a more formal point of view. In this contribution, we focus on a
particular type of implicit phylogenetic information, following the ideas
of \citet{Novichkov:04}. The basic idea is that the evolutionary distance
between orthologous genes is approximately proportional to the distances
between their species. Xenologous gene pairs as well as duplicate genes
thus appear as outliers
\cite{Lawrence:92,Clarke:02,Novichkov:04,Dessimoz:08}. More precisely,
consider a family of homologous genes in a set of species and plot the
phylogenetic distance of pairs of most similar homologs as a function of
the phylogenetic distances between the species in which they
reside. Since distances between orthologous genes can be expected to be
approximately proportional to the distances between the species,
orthologous pairs fall onto a regression line that defines equal
divergence time for the last common ancestor of corresponding gene and
species pairs. The gene pairs with ``later divergence times'', i.e.,
those that are more closely related than expected from their species,
fall below the regression line \cite{Novichkov:04}. \citet{Kanhere:09}
complemented this idea with a statistical test based on the Cook distance
to identify xenologous pairs in a statistically sound manner. For the
mathematical analysis we assume that we can perfectly identify all pairs
of genes $a$ and $b$ that are more closely related than expected from the
phylogenetic distance of their respective genomes. Naturally, this defines
a graph $(G,\sigma)$, whose vertices $x$ (the genes) are colored by the
species $\sigma(x)$ in which they appear. Here, we are interested in two
questions:
\begin{enumerate}
  \item[(1)]  What are the
  mathematical properties that characterize these
  ``\emph{later-divergence-time}'' (\emph{LDT}) graphs? 
  \item[(2)] What kind of
  information about HGT events, the gene and species tree, and the
  reconciliation map between them is contained implicitly in an LDT graph?
\end{enumerate}
In Sec.~\ref{sect:edit} we will briefly consider the situation that
later-divergence-time information is fraught with experimental errors.

These questions are motivated by a series of recent publications that
characterized the mathematical structure of orthology
\cite{Hellmuth:13a,Lafond:14}, the xenology relation \textit{sensu} Fitch
\cite{Geiss:18a,Hellmuth:18a,Hellmuth:2019a}, and the (reciprocal) best
match relation \cite{Geiss:19a,Geiss:20a,Schaller:20x,Schaller:21g}. Each
of these relations satisfies stringent mathematical conditions that -- at
least in principle -- can be used to correct empirical estimates and thus
serve as a potential means of noise reduction
\cite{Hellmuth:15a,Stadler:20a}. This approach has also lead to efficient
algorithms to extract gene trees, species trees, and reconciliations from
the relation data. Although the resulting representations of gene family
histories are usually not fully resolved, they can provide important
constraints for subsequent refinements. The advantage of the relation-based
approach is primarily robustness. While the inference of phylogenetic trees
relies on detailed probability models or the additivity of distance
metrics, our approach starts from yes/no answers to simple, pairwise
comparisons.  These data can therefore be represented as edges in a graph,
possibly augmented by a measure of confidence. Noise and inaccuracies in
the initial estimates then translate into violations of the required
mathematical properties of the graphs in question. Graph editing approaches
can therefore be harnessed as a means of noise reduction
\cite{Hellmuth:15a,dondi2017approximating,Lafond:14,Lafond:16,
  Hellmuth:20a,Hellmuth:20b,Schaller:20y}.

Previous work following this paradigm has largely been confined to
duplication-loss (DL) scenarios, excluding horizontal transfer. As shown in
\cite{Hellmuth:2017}, it is possible to partition a gene set into HGT-free
classes separated by HGTs. Within each class, the reconstruction problems
then simplify to the much easier DL scenarios. It is of utmost interest,
therefore, to find robust methods to infer this partition directly from
(dis)similarity data. Here, we explore the usefulness and limitations of
LDT graphs for this purpose.

This contribution is organized as follows. After introducing the necessary
notation, we introduce \emph{relaxed scenarios}, a very general
framework to describe evolutionary scenarios that emphasizes time
consistency of reconciliation rather than particular types of evolutionary
events. In Sec.~\ref{sect:LDT}, LDT graphs are defined formally and
characterized as those properly colored cographs for which a set of
accompanying rooted triples is consistent
(Thm.~\ref{thm:characterization}). The proof is constructive and provides a
method (Algorithm~\ref{alg:Ru-recognition}) to compute a  relaxed 
scenario for a given LDT graph. Sec.~\ref{sect:HGT} defines HGT events,
shows that every edge in a LDT graph corresponds to an HGT event, and
characterizes those LDT graphs that already capture all HGT events. In
addition, we provide a characterization of ``rs-Fitch graphs'' (general
vertex-colored graphs that capture all HGT events) in terms of their
coloring.  These properties can be verified in polynomial time.  Since LDT
graphs do not usually capture all HGT events, we discuss in
Sec.~\ref{app:edit} several ways to obtain a plausible set of HGT
candidates from LDT graphs. In Sec.~\ref{sect:simul}, we address the
question ``how much information about all HGT events is contained in LDT
graphs'' with the help of simulations of evolutionary scenarios with a wide
range of duplication, loss, and HGT events. We find that LDT graphs cover
roughly a third of xenologous pairs, while a simple greedy graph editing
scheme can more than double the recall at moderate false positive
rates. This greedy approach already yields a median accuracy of $89 \%$,
and in $99.8 \%$ of the cases produces biologically feasible solutions in
the sense that the inferred graphs are rs-Fitch graphs. We close with a
discussion of several open problems and directions for future research in
Sec.~\ref{sect:concl}.

The material of this contribution is extensive and contains several
lengthy, very technical proofs. We therefore divided the presentation into
a Narrative Part that contains only those mathematical results that
contribute to our main conclusions, and a Technical Part providing
additional results and all proofs. To facilitate cross-referencing between
the two parts, the same numbering of Definitions, Lemmas, Theorems, etc.,
is used. Sections \ref{TP:sect:LDT}, \ref{TP:sect:HGT}, and \ref{app:edit}
contain the technical material corresponding to Sections \ref{sect:LDT},
\ref{sect:HGT}, and \ref{sect:edit}, respectively.

\section{Notation}

\paragraph{Graphs.} 
We consider undirected graphs $G=(V,E)$ with vertex set
$V(G)\coloneqq V$ and edge set $E(G)\coloneqq E$, and denote edges
connecting vertices $x,y\in V$ by $xy$. The graphs $K_1$ and $K_2$
denote the complete graphs on one and two vertices, respectively. The
graph $K_2+K_1$ is the disjoint union of a $K_2$ and a $K_1$.

The join $G\join H$ of two graphs $G=(V,E)$ and $H=(W,F)$ is the graph
with vertex set $V\cupdot W$ and edge set
$E\cupdot F\cupdot \{xy\mid x\in V,y\in W\}$.  We write $H\subseteq G$
if $V(H)\subseteq V(G)$ and $E(H)\subseteq E(G)$, in which case $H$ is
called a \emph{subgraph of $G$}.  Given a graph $G=(V,E)$, we write $G[W]$
for the graph induced by $W\subseteq V$.  A \emph{connected component} $C$
of $G$ is an inclusion-maximal vertex set such that $G[C]$ is connected. A
\emph{(maximal) clique} $C$ in an undirected graph $G$ is an
(inclusion-maximal) vertex set such that, for all vertices $x,y\in C$, it
holds that $xy\in E(G)$, i.e., $G[C]$ is \emph{complete}.  A subset
$W\subseteq V$ is a \emph{(maximal) independent set} if $G[W]$ is edgeless
(and $W$ is maximal w.r.t.\ inclusion).  A graph $G = (V,E)$ is
\emph{complete multipartite} if $V$ consists of $k\ge 1$ pairwise disjoint
independent sets $I_1,\dots, I_k$ and $xy\in E$ if and only if $x\in I_i$
and $y\in I_j$ with $i\neq j$.

A graph $G$ together with a vertex coloring $\sigma$, denoted by
$(G,\sigma)$, is \emph{properly colored} if $uv \in E(G)$ implies
$\sigma(u)\neq \sigma(v)$.  For a coloring $\sigma\colon V\to M$ and a
subset $W\subseteq V$, we write
$\sigma(W) \coloneqq \{\sigma(w)\mid w\in W\}$ for the set of colors that
appear on the vertices in $W$. Throughout, we will need
restrictions of the coloring map $\sigma$.
\begin{definition}
  Let $\sigma\colon L\to M$ be a map, $L'\subseteq L$ and
  $\sigma(L') \subseteq M' \subseteq M$.  Then, the map
  $\sigma_{|L',M'}\colon L'\to M'$ is defined by putting
  $\sigma_{|L',M'}(v) = \sigma(v)$ for all $v\in L'$.  If we only restrict
  the domain of $\sigma$, we just write $\sigma_{|L'}$ instead of
  $\sigma_{|L',M}$.
  \label{def:sigma-restrictions}
\end{definition}
We do neither assume that $\sigma$ nor that its restriction $\sigma_{|L',M'}$ is
surjective.

\paragraph{Rooted Trees.} All trees appearing in this contribution are
rooted in one of their vertices. We write $x \preceq_{T} y$ if $y$ lies on
the unique path from the root to $x$, in which case $y$ is called an
ancestor of $x$, and $x$ is called a descendant of $y$. We may also write
$y \succeq_{T} x$ instead of $x \preceq_{T} y$. We use $x \prec_T y$ for
$x \preceq_{T} y$ and $x \neq y$.  In the latter case, $y$ is a
\emph{strict ancestor} of $x$. If $x \preceq_{T} y$ or $y \preceq_{T} x$,
the vertices $x$ and $y$ are \emph{comparable} and, otherwise,
\emph{incomparable}. We write $L(T)$ for the set of leaves of the tree $T$,
i.e., the $\preceq_T$-minimal vertices and say that $T$ is a tree \emph{on
  $L(T)$}. We write $T(u)$ for the subtree of $T$ rooted in $u$. The
\emph{last common ancestor} of a vertex set $W\subseteq V(T)$ is the
$\preceq_T$-minimal vertex $u\coloneqq \lca_T(W)$ for which $w\preceq_T u$
for all $w\in W$. For brevity we write $\lca_T(x,y)=\lca_T(\{x,y\})$.

We employ the convention that edges $(x,y)$ in a tree are always written
such that $y \preceq_{T} x$ is satisfied.  If $(x,y)$ is an edge in $T$,
then $\parent(y)\coloneqq x$ is the \emph{parent} of $y$, and $y$ the
\emph{child} of $x$. We denote with $\child_T(x)$ the set of all children
of $x$ in $T$. It will be convenient for the discussion below to extend the
ancestor relation $\preceq_T$ on $V$ to the union of the edge and vertex
sets of $T$. More precisely, for a vertex $x\in V(T)$ and an edge
$e=(u,v)\in E(T)$ we put $x \prec_T e$ if and only if $x\preceq_T v$; and
$e \prec_T x$ if and only if $u\preceq_T x$.  In addition, for edges
$e=(u,v)$ and $f=(a,b)$ in $T$ we put $e\preceq_T f$ if and only if
$v \preceq_T b$.

A rooted tree is \emph{phylogenetic} if all vertices that are adjacent to
at least two vertices have at least two children.  A rooted tree $T$ is
planted if its root has degree $1$. In this case, we denote the ``planted
root'' by $0_T$.  In planted phylogenetic trees there is a unique ``planted
edge'' $(0_T,\rho_T)$ where $\rho_T\coloneqq \lca_T(L(T))$.  Note that by
definition $0_T\notin L(T)$.

\emph{Throughout, we will assume that all trees are rooted and
  phylogenetic unless explicitly stated otherwise. Whenever there is no
  danger of confusion, we will refer also to planted phylogenetic trees
  simply as trees.}

The set of \emph{inner vertices} is given by
$V^0(T)\coloneqq V(T)\setminus (L(T)\cup \{0_T\})$.  An edge $(u,v)$ is an
\emph{inner} edge if both vertices $u$ and $v$ are inner vertices and,
otherwise, an \emph{outer} edge. The restriction of $T$ to a subset
$L'\subseteq L(T)$ of leaves, denoted by $T_{|L'}$ is obtained by
identifying the (unique) minimal subtree of $T$ that connects all leaves in
$L'$, and suppressing all vertices with degree two except possibly the root
$\rho_{T_{L'}}=\lca_T(L')$. $T$ \emph{displays} a tree $T'$, in symbols
$T'\le T$, if $T'$ can be obtained from a restriction $T_{|L'}$ of $T$ by a
series of inner edge contractions \cite{Bryant:95}.  If, in addition,
$L(T)=L(T')$, then $T$ is a \emph{refinement} of $T'$.  Throughout this
contribution, we will consider leaf-colored trees $(T,\sigma)$ with
$\sigma$ being defined for $L(T)$ only.

\paragraph{Rooted Triples.} A rooted triple is a tree $T$ on three leaves
and two internal vertices.  We write $ab|c$ for the triple with
$\lca_T(a,b)\prec \lca_T(a,c)=\lca_T(b,c)$. For a set $\mathscr{R}$ of
triples we write
$L(\mathscr{R})\coloneqq
\bigcup_{\mathsf{t}\in\mathscr{R}}L(\mathsf{t})$.  The set $\mathscr{R}$
is \emph{compatible} if there is a tree $T$ with
$L(\mathscr{R}) \subseteq L(T)$ that displays every triple
$\mathsf{t}\in\mathscr{R}$. The construction of such a tree $T$ from a
triple set $\mathscr{R}$ on $L$ makes use of an auxiliary graph that will
play a prominent role in this contribution.
\begin{definition}{\cite{Aho:81}}
  Let $\mathscr{R}$ be a set of rooted triples on the vertex set $L$.  The
  \emph{Aho graph} $[\mathscr{R},L]$ has vertex set $L$ and edge set
  $\{ xy \mid \exists z\in L:\, xy|z \in\mathscr{R}\}$.
\end{definition}
The algorithm \texttt{BUILD} \cite{Aho:81} uses Aho graphs in a top-down
recursion starting from a given set of triples $\mathscr{R}$ and returns
for compatible triple sets $\mathscr{R}$ on $L$ an unambiguously defined
tree $\Aho(\mathscr{R}, L)$ on $L$, which is known as the \emph{Aho
  tree}. \texttt{BUILD} runs in polynomial time. The key property of the
Aho graph that ensures the correctness of \texttt{BUILD} can be stated as
follows:
\begin{proposition}{\cite{Aho:81,Bryant:95}}
  A set of triples $\mathscr{R}$ is compatible if and only if for each
  subset $L\subseteq L(\mathscr{R})$ with $|L|>1$ the graph
  $[\mathscr{R},L]$ is disconnected.
  \label{prop:ahograph}
\end{proposition}

\paragraph{Cographs} are recursively defined as undirected graphs that can
be generated as joins or disjoint unions of cographs, starting from
single-vertex graphs $K_1$. The recursive construction defines a rooted
tree $(T,t)$, called \emph{cotree}, whose leaves are the vertices of the
cograph $G$, i.e., the $K_1$s, while each of its inner vertices $u$ of $T$
represent the join or disjoint union operations, labeled as $t(u)=1$ and
$t(u)=0$, respectively. Hence, for a given cograph $G$ and its cotree
$(T,t)$, we have $xy\in E(G)$ if and only if $t(\lca_T(x,y))=1$.
Contraction of all tree edges $(u,v)\in E(T)$ with $t(u)=t(v)$ results in
the \emph{discriminating cotree} $(T_G,\hat t)$ of $G$ with cotree-labeling
$\hat t$ such that $\hat t(u)\ne \hat t(v)$ for any two adjacent interior
vertices of $T_G$. The discriminating cotree $(T_G,\hat t)$ is uniquely
determined by $G$ \cite{Corneil:81}. Cographs have a large number of
equivalent characterizations. In this contribution, we will need the
following classical results:
\begin{proposition}{\cite{Corneil:81}}
  Given an undirected graph $G$, the following statements are equivalent:
  \begin{enumerate}
    \item $G$ is a cograph.
    \item $G$ does not contain a $P_4$, i.e., a path on four vertices, as an
    induced subgraph.
    \item $\mathrm{diam}(H) \leq 2$ for all connected induced subgraphs $H$
    of $G$.
    \item Every induced subgraph $H$ of $G$ is a cograph.
  \end{enumerate}
  \label{prop:cograph}
\end{proposition}

\section{Relaxed Reconciliation Maps and Relaxed Scenarios}
\label{sec:reconciliation}

\citet{Tofigh:11} and \citet{Bansal:12} define
``Duplication-Transfer-Loss'' (DTL) scenarios in terms of a vertex-only map
$\gamma:V(T)\to V(S)$. The H-trees introduced by
\citet{Gorecki:2010,Gorecki:2012} formalize the same concept in a very
different manner.  A definition of a DTL-like class of scenarios in
terms of a reconciliation map $\mu: V(T)\to V(S)\cup E(S)$ was analyzed by
\citet{Nojgaard:18a}.  For binary trees, the two definitions are
equivalent; for non-binary trees, however, the DTL-scenarios are a proper
subset, see \cite[Fig.~1]{Nojgaard:18a} for an example. Several other
mathematical frameworks have been used in the literature to specify
evolutionary scenarios. Examples include the DLS-trees of
\citet{Gorecki:06}, which can be seen as event-labeled gene trees with
leaves denoting both surviving genes and loss-events, maps
$g:V(S')\to 2^{V(T)}$ from a suitable subdivision $S'$ of the species tree
$S$ to the gene tree as used by \citet{Hallett:01}, and associations of
edges, i.e., subsets of $E(T)\times E(S)$ \cite{Wieseke:13}.

In the presence of HGT, the relationships of gene trees and species are not
only constrained by local conditions corresponding to the admissible local
evolutionary events (duplication, speciation, gene loss, and HGT) but also
by the global condition that the HGT events within each lineage admit a
temporal order \cite{Merkle:05,Gorbunov:09,Tofigh:11}. In order to capture
time consistency from the outset and to establish the mathematical
framework, we consider here trees with explicit timing information
\cite{Merkle:05}. 
\begin{definition}[Time Map]
  \label{def:time-map}
  The map $\tT : V(T) \rightarrow \mathbb{R}$ is a time map for a 
  tree $T$ if $x\prec_T y$ implies $\tT(x)<\tT(y)$ for all $x,y\in V(T)$.
\end{definition}
It is important to note that only \emph{qualitative}, relative timing
information will be used in practice, i.e., we will never need the actual
value of time maps but only information on whether an event pre-dates,
post-dates, or is concurrent with another.  Def.~\ref{def:time-map}
ensures that the ancestor relation $\preceq_T$ and the timing of the
vertices are not in conflict.  For later reference, we provide the
following simple result.
\begin{lemma}
  Given a tree $T$, a time map $\tT$ for $T$ satisfying
  $\tT(x)=\tau_0(x)$ with arbitrary choices of $\tau_0(x)$ for all
  $x\in L(T)$ can be constructed in linear time.
  \label{lem:arbitrary-tT}
\end{lemma}
\begin{proof}
  We traverse $T$ in postorder. If $x$ is a leaf, we set
  $\tT(x)=\tau_0(x)$, and otherwise compute
  $t\coloneqq\max_{u\in\child(x)} \tT(u)$ and set $\tT(x)=t'$ with an
  arbitrary value $t'> t$. Clearly the total effort is
  $O(|V(T)|+|E(T)|)$, and thus also linear in the number of leaves $L(T)$.
\end{proof}
Lemma~\ref{lem:arbitrary-tT} will be useful for the construction of time
maps as it, in particular, allows us to put $\tT(x)=\tT(y)$ for all
$x,y\in L(T)$.

\begin{definition}[Time Consistency]
  \label{def:tc-map}
  Let $T$ and $S$ be two  trees.  A map
  $\mu \colon V(T) \to V(S) \cup E(S)$ is called \emph{time-consistent} if
  there are time maps $\tT$ for $T$ and $\tS$ for $S$ satisfying the
  following conditions for all $u \in V(T)$:
  \begin{description}
    \item[(C1)] If $\mu(u) \in V(S)$, then $\tT(u) = \tS(\mu(u))$.
    \item[(C2)] Else, if $\mu(u) = (x,y) \in E(S)$, then
    $\tS(y)<\tT(u)<\tS(x)$.
  \end{description}
\end{definition}
Conditions (C1) and (C2) ensure that the  reconciliation map $\mu$
preserves time in the following sense: If vertex $u$ of the gene tree is
mapped to a vertex $\mu(u)=v$ in the species tree, then $u$ and $v$
receive the same time stamp by Condition (C1). If $u$ is mapped to an edge
$\mu(u) = (x,y)$, then the time stamp of $u$ falls within the time range
$[\tS(x),\tS(y)]$ of the edge $xy$ in the species tree.
The following definition of reconciliation is designed (1) to be general
enough to encompass the notions of reconciliation that have been studied in
the literature, and (2) to separate the mapping between gene tree and
species tree from specific types of events.  Event types such as
duplication or horizontal transfer therefore are considered here as a
matter of \emph{interpreting} scenarios, not as part of their definition.

\begin{definition}[Relaxed Reconciliation Map] 
  Let $T$ and $S$ be two planted trees with leaf 
  sets $L(T)$ and $L(S)$, respectively and let $\sigma:L(T)\to L(S)$ be a map. A
  map $\mu \colon V(T)\to V(S)\cup E(S)$ is a 
  \emph{relaxed  reconciliation map} for $(T,S,\sigma)$ if the 
  following conditions are satisfied:
  \begin{description}
    \item[(G0)] \emph{Root Constraint.}
    $\mu(x) = 0_{S}$ if and only if $x = 0_{T}$
    \item[(G1)] \emph{Leaf Constraint.}
    $\mu(x)=\sigma(x)$ if and only if $x\in L(T)$.
    \item[(G2)] \emph{Time Consistency Constraint.}  The map $\mu$ is
    time-consistent for some time maps $\tT$ for $T$ and $\tS$ 
    for $S$.
  \end{description}
  \label{def:relaxed-reconc}
\end{definition}

Condition (G0) is used to map the respective planted roots. (G1)
ensures that genes are mapped to the species in which they reside.  (G2)
enforces time consistency. The reconciliation maps most commonly used in
the literature, see e.g.\ \cite{Tofigh:11,Bansal:12}, usually not only
satisfy (G0)--(G2) but also impose additional conditions. We therefore
call the map $\mu$ defined here ``relaxed''.

\begin{definition}[ relaxed  Scenario] 
  The 6-tuple $\scen = (T,S,\sigma,\mu,\tT,\tS)$ is a
  \emph{relaxed scenario}
  if $\mu$ is a relaxed reconciliation
  map for $(T,S,\sigma)$ that satisfies (G2) w.r.t.\ the time maps $\tT$ 
  and $\tS$.
  \label{def:relaxed-scenario}
\end{definition}
By definition, relaxed reconciliation maps are time-consistent.
Moreover, $\tT(x)=\tS(\sigma(x))$ for all $x \in L(T)$ by
Def.~\ref{def:tc-map}(C1) and Def.~\ref{def:relaxed-reconc}(G1,G2).  In
the following we will refer to the map $\sigma:L(T)\to L(S)$ as the
\emph{coloring of $\scen$}.

\section{Later-Divergence-Time Graphs}
\label{sect:LDT}

\subsection{LDT Graphs and $\mu$-free Scenarios}

In the absence of horizontal gene transfer, the last common ancestor of two
species $A$ and $B$ should mark the latest possible time point at which two
genes $a$ and $b$ residing in $\sigma(a)=A$ and $\sigma(b)=B$,
respectively, may have diverged. Situations in which this constraint is
violated are therefore indicative of HGT. To address this issue in some
more detail, we next define ``$\mu$-free scenarios'' that
eventually will lead us to the class of ``LDT graphs'' that contain all
information about genes that diverged after the species in which they
reside.

\begin{cdefinition}{\ref{def:mu-free}}[{$\mathbf{\mu}$}-free scenario]
  Let $T$ and $S$ be planted trees, $\sigma\colon L(T)\to L(S)$ be a map,
  and $\tT$ and $\tS$ be time maps of $T$ and $S$, respectively, such that
  $\tT(x) = \tS(\sigma(x))$ for all $x\in L(T)$.  Then,
  $\mfscen=(T,S,\sigma,\tT,\tS)$ is called a \emph{$\mu$-free scenario}.
\end{cdefinition}

This definition of a scenario without a reconciliation map $\mu$ is
mainly a technical convenience that simplifies the arguments in various
proofs by avoiding the construction of a reconciliation map.  It is
motivated by the observation that the ``later-divergence-time'' of two
genes in comparison with their species is independent from any such
$\mu$. Every relaxed scenario $\scen=(T,S,\sigma,\mu,\tT,\tS)$
implies an underlying $\mu$-free scenario
$\mfscen=(T,S,\sigma,\tT,\tS)$. Statements proved for $\mu$-free
scenarios therefore also hold for relaxed scenarios. Note that, by
Lemma~\ref{lem:arbitrary-tT}, given the time map $\tS$, one can easily
construct a time map $\tT$ such that $\tT(x) = \tS(\sigma(x))$ for all
$x\in L(T)$.  In particular, when constructing relaxed scenarios
explicitly, we may simply choose $\tT(u)=0$ and $\tS(x)=0$ as common time
for all leaves $u\in L(T)$ and $x\in L(S)$. Although not all
$\mu$-free scenarios admit a reconciliation map and thus can be turned into
relaxed scenarios, Lemma~\ref{lem:mfscen} below implies that for every
$\mu$-free scenario $\mfscen$ there is a relaxed scenario with possibly slightly
distorted time maps that encodes the same LDT graph as $\mfscen$.

\begin{cdefinition}{\ref{def:LDTgraph}}[LDT  graph]
  For a $\mu$-free scenario $\mfscen=(T,S,\sigma,\tT,\tS)$, we define
  $\Gu(\mfscen) = \Gu(T,S,\sigma,\tT,\tS) = (V,E)$ as the graph with vertex
  set $V\coloneqq L(T)$ and edge set
  \begin{equation*}
    E \coloneqq \{ab\mid a,b\in L(T),
    \tT(\lca_T(a,b))<\tS(\lca_S(\sigma(a),\sigma(b))). \}
  \end{equation*}
  A vertex-colored graph $(G,\sigma)$ is a \emph{later-divergence-time graph
    (LDT graph)}, if there is a $\mu$-free scenario
  $\mfscen=(T,S,\sigma,\tT,\tS)$  such that
  $G=\Gu(\mfscen)$.  In this case, we say that $\mfscen$ \emph{explains}
  $(G,\sigma)$.
\end{cdefinition}

It is easy to see that the edge set of $\Gu(\mfscen)$ defines an
\emph{undirected} graph and that two genes $a$ and $b$ form an edge if
the divergence time of $a$ and $b$ is strictly less than the divergence
time of the underlying species $\sigma(a)$ and $\sigma(b)$. Moreover,
there are no edges of the form $aa$, since
$\tT(\lca_T(a,a)) = \tT(a) = \tS(\sigma(a))
=\tS(\lca_S(\sigma(a),\sigma(a)))$. Hence $\Gu(\mfscen)$ is a simple graph.

By definition, every relaxed scenario $\scen =(T,S,\sigma,\mu,\tT,\tS)$
satisfies $\tT(x)=\tS(\sigma(x))$ all $x \in L(T)$. Therefore, removing
$\mu$ from $\scen$ yields a $\mu$-free scenario
$\mfscen=(T,S,\sigma,\tT,\tS)$. Thus, we will use the following simplified
notation.
\begin{cdefinition}{\ref{def:Gu-scen}}
  We put $\Gu(\scen) \coloneqq \Gu(T,S,\sigma,\tT,\tS)$ for a given relaxed
  scenario $\scen =(T,S,\sigma,\mu,\tT,\tS)$ and the underlying $\mu$-free
  scenario $(T,S,\sigma,\tT,\tS)$ and say, by slight abuse of notation,
  that $\scen$ \emph{explains} $(\Gu(\scen),\sigma)$.
\end{cdefinition}
The next two results show that the existence of a reconciliation map $\mu$
does not impose additional constraints on LDT graphs.

\begin{clemma}{\ref{lem:mfscen}}
  For every $\mu$-free scenario $\mfscen=(T,S,\sigma,\tT,\tS)$, there is a
  relaxed scenario $\scen=(T,S,\sigma,\mu,\widetilde\tT,\widetilde\tS)$ for
  $T, S$ and $\sigma$ such that
  $(\Gu(\mfscen),\sigma) = (\Gu(\scen), \sigma)$.
\end{clemma}

\begin{ctheorem}{\ref{thm:LDT-scen}}
  $(G,\sigma)$ is an LDT graph if and only if there is a relaxed scenario
  $\scen = (T,S,\sigma,\mu,\tT,\tS)$  such that
  $(G,\sigma) = (\Gu(\scen),\sigma)$.
\end{ctheorem}

\begin{remark}
  From here on, we omit the explicit reference to Lemma~\ref{lem:mfscen}
  and Thm.~\ref{thm:LDT-scen} and assume that the reader is aware of the
  fact that every LDT graph is explained by some relaxed scenario $\scen$
  and that for every $\mu$-free scenario $\mfscen=(T,S,\sigma,\tT,\tS)$,
  there is a relaxed scenario $\scen$ for $T, S$ and $\sigma$ such that
  $(\Gu(\mfscen),\sigma) = (\Gu(\scen), \sigma)$.
\end{remark}

\begin{figure}[t]
  \begin{center}
    \includegraphics[width=0.7\textwidth]{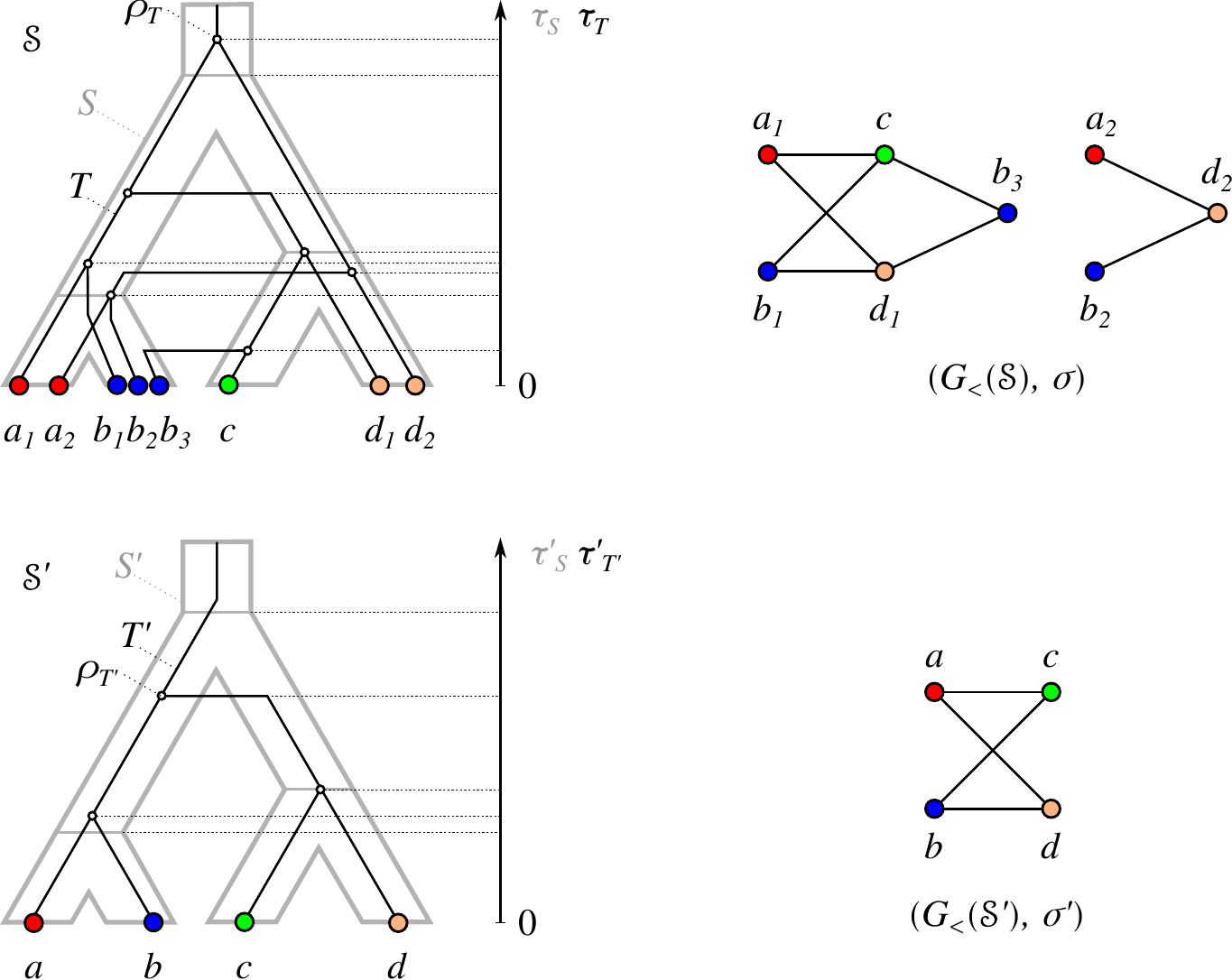}
  \end{center}
  \caption{Top row: A relaxed scenario $\scen =(T,S,\sigma,\mu,\tT,\tS)$
    (left) with its LDT graph $(\Gu(\scen),\sigma)$ (right). The
    reconciliation map $\mu$ is shown implicitly by the embedding of the
    gene tree $T$ into the species tree $S$. The times $\tT$ and $\tS$ are
    indicated by the position on the vertical axis, i.e., if a vertex $x$
    is drawn higher than a vertex $y$, this implies $\tT(y)<\tT(x)$. In
    subsequent figures we will not show the time maps explicitly.  Bottom
    row: Another relaxed scenario $\scen' =(T',S',\sigma',\mu',\tT',\tS')$
    with a connected LDT graph $(\Gu(\scen'),\sigma')$. As we shall see,
    connectedness of an LDT graph depends on the relative timing of the
    roots of the gene and species tree (cf.\ Lemma~\ref{lem:Gu-connected}).
  }
  \label{fig:Gu-example}
\end{figure}

\subsection{Properties of LDT Graphs} 

We continue by deriving several interesting characteristics LDT graphs.
\begin{cproposition}{\ref{prop:properCol}}
  Every LDT graph $(G,\sigma)$ is properly colored.
\end{cproposition}

As we shall see below, LDT graphs $(G,\sigma)$ contain detailed information
about both the underlying gene trees $T$ and species trees $S$ for
\emph{all} $\mu$-scenarios that explain $(G,\sigma)$, and thus by 
Lemma~\ref{lem:mfscen} and Thm.~\ref{thm:LDT-scen} also about every relaxed
scenario $\scen$ satisfying $G=\Gu(\scen)$. This information is encoded in
the form of certain rooted triples that can be retrieved directly from
local features in the colored graphs $(G,\sigma)$.

\begin{cdefinition}{\ref{def:infoTriples}}
  For a graph $G=(L,E)$, we define the set of triples on $L$ as
  \begin{equation*}
    \Tri(G) \coloneqq \{xy|z \; \colon
    x,y,z\in L \text{ are pairwise distinct, }
    xy\in E,\; xz,yz\notin E\} \,.
  \end{equation*}
  If $G$ is endowed with a coloring $\sigma\colon L\to M$ we also define a
  set of color triples
  \begin{align*}
    \Sri(G,\sigma) \coloneqq \{\sigma(x)\sigma(y)|\sigma(z)\; \colon 
    & x,y,z\in L,\, \sigma(x),\sigma(y),\sigma(z) \text{ are pairwise 
      distinct},\\
    &xz, yz\in E,\; xy\notin E\}.
  \end{align*}
\end{cdefinition}

\begin{clemma}{\ref{lem:Ru-SpeciesTriple}}
  If a graph $(G,\sigma)$ is an LDT graph, then $\Sri(G,\sigma)$ is
  compatible and $S$ displays $\Sri(G,\sigma)$ for every $\mu$-free
  scenario $\mfscen=(T,S,\sigma,\tT,\tS)$ that explains $(G,\sigma)$.
\end{clemma}  

The next lemma shows that induced $K_2+K_1$ subgraphs
in LDT graphs imply triples that must be displayed by the gene tree $T$.

\begin{clemma}{\ref{lem:Ru-GeneTriple}}
  If $(G,\sigma)$ is an LDT graph, then $\Tri(G)$ is compatible and $T$
  displays $\Tri(G)$ for every $\mu$-free scenario
  $\mfscen=(T,S,\sigma,\tT,\tS)$ that explains $(G,\sigma)$.
\end{clemma}

The next results shows that LDT graphs cannot contain induced $P_4$s.
\begin{clemma}{\ref{lem:propcolcograph}}
  Every LDT graph $(G,\sigma)$ is a properly colored cograph.
\end{clemma}

The converse of Lemma~\ref{lem:propcolcograph} is not true is in general.
To see this, consider the properly-colored cograph $(G,\sigma)$ with vertex
$V(G)=\{a,a',b,b',c,c'\}$, edges $ab,bc, a'b',a'c' $ and coloring
$\sigma(a)=\sigma(a')=A$, $\sigma(b)=\sigma(b')=B$, and
$\sigma(c)=\sigma(c')=C$ with $A,B,C$ being pairwise distinct. In this
case, $\Sri(G,\sigma)$ contains the triples $AC|B$ and $BC|A$.  By Lemma
\ref{lem:Ru-SpeciesTriple}, the tree $S$ in every $\mu$-free scenario
$\mfscen=(T,S,\sigma,\tT,\tS)$ or relaxed scenario
$\scen=(T,S,\sigma,\mu,\tT,\tS)$ explaining $(G,\sigma)$ displays $AC|B$
and $BC|A$. Since no such scenario can exist, $(G,\sigma)$ is not an LDT
graph.

\subsection{Recognition and Characterization of LDT Graphs}

In order to design an algorithm for the recognition of LDT graphs, we
will consider partitions of the vertex set of a given input graph
$(G=(L,E),\sigma)$. To construct suitable partitions, we start with
the connected components of $G$. The coloring $\sigma\colon L\to M$
imposes additional constraints. We capture these with the help of
binary relations that are defined in terms of partitions $\mathcal{C}$
of the color set $M$ and employ them to further refine the partition of $G$.

\begin{cdefinition}{\ref{def:rel}}
  Let $(G=(L,E),\sigma)$ be a graph with coloring $\sigma\colon L\to M$.
  Let $\mathcal{C}$ be a partition of $M$, and $\mathcal{C}'$ be the set of
  connected components of $G$.  We define the following binary relation
  $\rel(G, \sigma, \mathcal{C})$ by setting
  \begin{align*}
    (x,y)\in \rel(G, \sigma, \mathcal{C}) \iff  x,y\in L,\;
    \sigma(x), \sigma(y) & \in  C 
    \text{ for some } C\in\mathcal{C},
    \text{ and } \\
    x,y & \in C' \text{ for some } C'\in\mathcal{C}'.
  \end{align*}
\end{cdefinition}
By construction, two vertices $x,y\in L$ are in relation
$\rel(G, \sigma, \mathcal{C})$ whenever they are in the same connected
component of $G$ and their colors $\sigma(x), \sigma(y)$ are contained in
the same set of the partition of $M$. As shown in Lemma~\ref{lem:KinCC} in
the Technical Part, the relation
$\rel\coloneqq\rel(G, \sigma, \mathcal{C})$ is an equivalence relation and
every equivalence class of $\rel$ is contained in some connected component
of $G$. In particular, each connected component of $G$ is the disjoint
union of $\rel$-classes.

The following partition of the leaf sets of subtrees of a tree $S$ rooted
at some vertex $u\in V(S)$ will be useful:
\begin{align*}
  &\text{If } u 
  \textrm{ is not a leaf, then }
  &\partS(u)&  \coloneqq \{L(S(v))
  \mid v\in\child_S(u)\} \\
  & \textrm{and, otherwise, } &\partS(u)&\coloneqq \{\{u\}\}.
\end{align*}
One easily verifies that, in both cases, $\partS(u)$ yields a valid
partition of the leaf set $L(S(u))$.  Recall that
$\sigma_{|L',M'}\colon L'\to M'$ was defined as the ``submap'' of $\sigma$
with $L'\subseteq L$ and $\sigma(L') \subseteq M' \subseteq M$.

\begin{clemma}{\ref{lem:xy-iff-Ks-in-same-CC}}
  Let $(G=(L,E),\sigma)$ be a properly colored cograph. Suppose that the
  triple set $\Sri(G,\sigma)$ is compatible and let $S$ be a tree on $M$
  that displays $\Sri(G,\sigma)$. Moreover, let $L'\subseteq L$ and
  $u\in V(S)$ such that $\sigma(L') \subseteq L(S(u))$. \
  Finally, set $\rel\coloneqq \rel(G[L'],\sigma_{|L',L(S(u))},\partS(u))$.\\
  Then, for all distinct $\rel$-classes $K$ and $K'$, either $xy\in E$ for
  all $x\in K$ and $y\in K'$, or $xy\notin E$ for all $x\in K$ and
  $y\in K'$.  In particular, for $x\in K$ and $y\in K'$, it holds that
  \begin{equation*}
    xy\in E \iff K, K' \text{ are contained in the same connected 
      component of } G[L'].
  \end{equation*}
\end{clemma}

\begin{figure}[t]
  \begin{center}
    \includegraphics[width=0.85\textwidth]{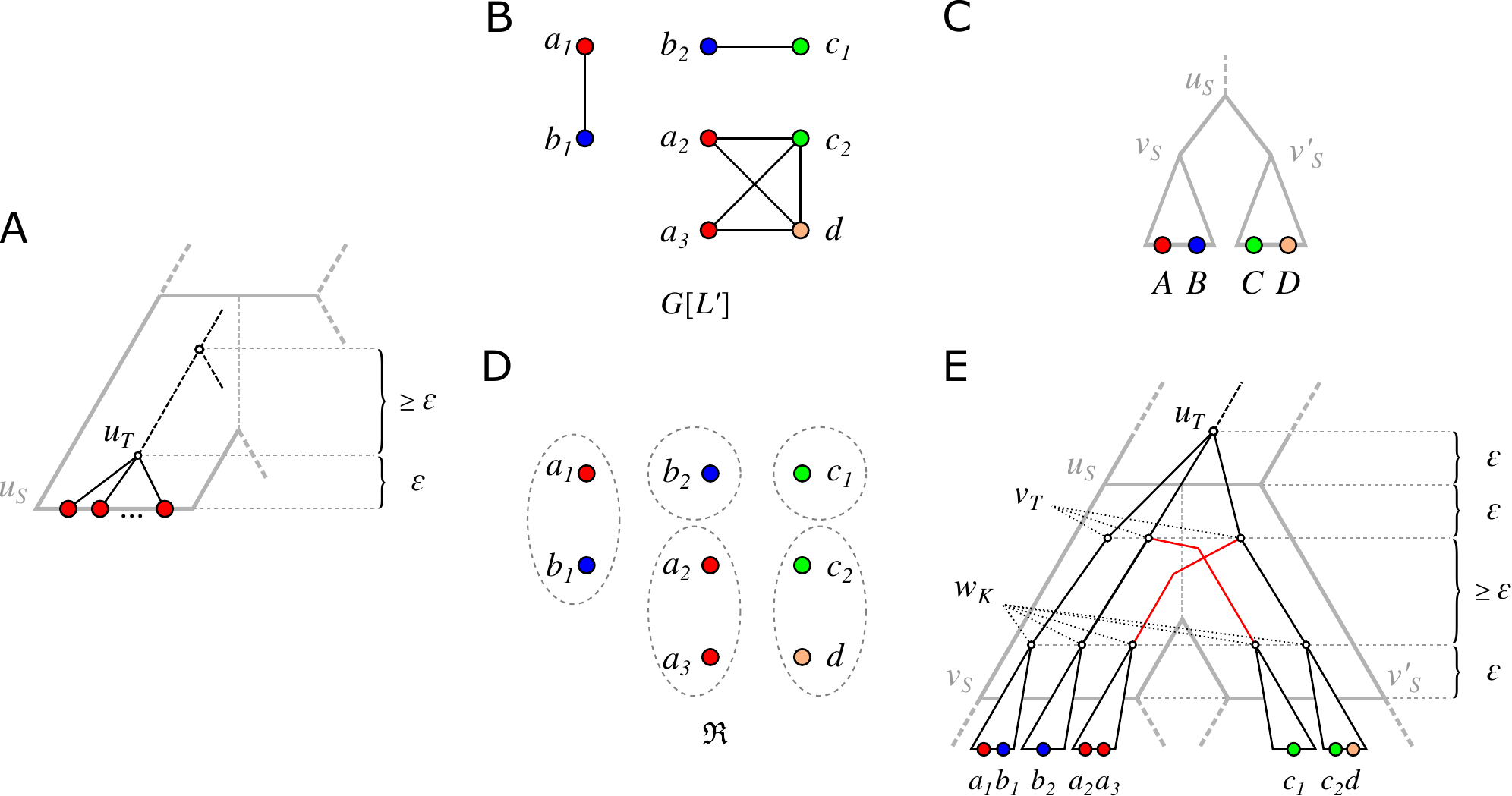}
  \end{center}
  \caption{Visualization of Algorithm~\ref{alg:Ru-recognition}. (A) The
    case $u_S$ is a leaf (cf.\ Line~\ref{line:species-leaf}). (B)-(E) The
    case $u_S$ is an inner vertex (cf.\ Line~\ref{line:else}). (B) The
    subgraph of $(G,\sigma)$ induced by $L'$. (C) The local topology of the
    species tree $S$ yields
    $\partS(u_S)=\{\{A,B,\dots\},\{C,D,\dots\}\}$. Note that $L(S(u_S))$
    may contain colors that are not present in $\sigma(L')$ (not
    shown). (D) The equivalence classes of
    $\rel\coloneqq \rel(G[L'], \sigma_{|L',L(S(u))}, \partS(u_S))$. (E)
    The vertex $u_T$ and the vertices $v_T$ are created in this recursion
    step. The vertices $w_K$ corresponding to the $\rel$-classes $K$ are
    created in the next-deeper steps. Note that some vertices have only a
    single child, and thus get suppressed in Line~\ref{line:Tphylo}.}
  \label{fig:algo-visu}
\end{figure}

\begin{algorithm}[t]
  \small
  
  \caption{Construction of a relaxed scenario $\scen$ for a properly
    colored cograph $(G,\sigma)$ with compatible triple set
    $\Sri(G,\sigma)$.}
  \label{alg:Ru-recognition}
  
  \DontPrintSemicolon
  \SetKwFunction{FRecurs}{void FnRecursive}%
  \SetKwFunction{FRecurs}{BuildGeneTree}
  \SetKwProg{Fn}{Function}{}{}
  
  \KwIn{A cograph $(G=(L,E),\sigma)$ with proper coloring
    $\sigma\colon L\to M$ and compatible triple set $\Sri(G,\sigma)$.
  } \label{line:if-false}
  \KwOut{A relaxed scenario $\scen=(T,S,\sigma,\mu,\tT,\tS)$ explaining 
    $(G,\sigma)$.}
  \BlankLine
  
  $S\leftarrow$ tree on $M$ displaying $\Sri(G,\sigma)$ with planted root 
  $0_S$ \label{line:S}\;
  $\tS\leftarrow$ time map for $S$ satisfying $\tS(x)=0$ for all $x\in L(S)$ 
  \label{line:tS}\;
  $\epsilon \leftarrow \frac{1}{3} \min\{\tS(y)-\tS(x) \mid (y,x)\in E(S) 
  \}$ \label{line:epsilon}\;
  initialize empty maps $\mu, \tT$\;
  
  \BlankLine
  \Fn{\FRecurs{$L',u_{S}$}}{
    create a vertex $u_T$ \label{line:create-uT}\;
    $\tT(u_T)\leftarrow \tS(u_{S}) + \epsilon$ and
    $\mu(u_T)\leftarrow (\parent_S(u_S), u_S)$ \label{line:mu-tT-inner1}\;
    \uIf{$u_S$ is a leaf\label{line:species-leaf}}{
      \ForEach{$x\in L'$}{
        connect $x$ as a child of $u_T$ \label{line:attach-leaf}\;
        $\tT(x)\leftarrow 0$ and $\mu(x)\leftarrow \sigma(x)$ 
        \label{line:mu-tT-leaves}\;
      }
    }
    \Else{\label{line:else}
      $\rel\leftarrow \rel(G[L'], \sigma_{|L',L(S(u_S))}, \partS(u_S))$
      \label{line:rel}\;
      \ForEach{connected component $C$ of $G[L']$}{
        create a vertex $v_T$ \label{line:create-vT}\;
        connect $v_T$ as a child of $u_T$\;
        choose $v^*_S\in \child_S(u_{S})$ such that $\sigma(C)\cap 
        L(S(v^*_S))\ne\emptyset$ \label{line:choose-v-S}\;
        $\tT(v_T)\leftarrow \tS(u_{S}) - \epsilon$ and
        $\mu(v_T)\leftarrow (u_S, v^*_S)$ \label{line:mu-tT-inner2}\;
        \ForEach{$\rel$-class $K$
          such that $K\subseteq C$}{
          identify $v_S\in \child_S(u_{S})$ such that $\sigma(K)\subseteq 
          L(S(v_S))$ \label{line:choose-v-S-for-class}\;
          $w_K\leftarrow$ \FRecurs{$K, v_S$} \label{line:recursive-call}\;
          connect $w_K$ as a child of $v_T$\;
        }
      }
    }
    \Return $u_T$\;
  }
  \BlankLine
  
  $T' \leftarrow$ tree with root \FRecurs{$L,\rho_S$}\;
  $T\leftarrow T'$ with (i) a planted root $0_T$ added, and (ii) all 
  inner degree-2 vertices (except $0_T$) suppressed \label{line:Tphylo}\;
  $\tT(0_T)\leftarrow \tS(0_S)$ and $\mu(0_T)\leftarrow 0_S$ 
  \label{line:mu-tT-planted-root}\;
  \textbf{return} $(T,S,\sigma,\mu_{|V(T)},\tau_{T|V(T)},\tS)$\;
\end{algorithm}

Lemma~\ref{lem:xy-iff-Ks-in-same-CC} suggests a recursive strategy to
construct a  relaxed  scenario $\scen=(T,S,\sigma,\mu,\tT,\tS)$ for a
given properly-colored cograph $(G,\sigma)$, which is illustrated in
Fig.~\ref{fig:algo-visu}. The starting point is a species tree $S$
displaying all the triples in $\Sri(G,\sigma)$ that are required by
Lemma~\ref{lem:Ru-SpeciesTriple}. We show below that there are no further
constraints on $S$ and thus we may choose $S=\Aho(\Sri(G,\sigma),L)$ and
endow it with an arbitrary time map $\tS$.  Given $(S,\tS)$, we construct
$(T,\tT)$ in top-down order.  In order to reduce the complexity of the
presentation and to make the algorithm more compact and readable, we will
not distinguish the cases in which $(G,\sigma)$ is connected or
disconnected, nor whether a connected component is a superset of one or
more $\rel$-classes. The tree $T$ therefore will not be phylogenetic in
general. We shall see, however, that this issue can be alleviated by simply
suppressing all inner vertices with a single child.

The root $u_T$ is placed above $\rho_S$ to ensure that no two vertices
from distinct connected components of $G$ will be connected by an edge in
$\Gu(\scen)$.  The vertices $v_T$ representing the connected components $C$
of $G$ are each placed within an edge of $S$ below $\rho_S$. W.l.o.g., the
edges $(\rho_S,v_S)$ are chosen such that the colors of the corresponding
connected component $C$ and the colors in $L(S(v_S))$ overlap. Next we
compute the relation $\rel\coloneqq\rel(G,\sigma,\partS(\rho_S))$ and
determine, for each connected component $C$, the $\rel$-classes $K$ that
are a subset of $C$. For each of them, a child $w_K$ is appended to the
tree vertex $v_T$. The subtree $T(w_K)$ will have leaf set $L(T(w_K))=K$.
Since $\rel$ is defined on $\partS(\rho_S)$ in this first step, $G(\scen)$
will have all edges between vertices that are in the same connected
component $C$ but in distinct $\rel$-classes (cf.\
Lemma~\ref{lem:xy-iff-Ks-in-same-CC}).  The definition of $\rel$ also
implies that we always find a vertex $v_S\in\child_S(\rho_S)$ such that
$\sigma(K)\subseteq L(S(v_S))$ (more detailed arguments for this are given
in the proof of Claim~\ref{clm:color-subset} in the proof of
Thm.~\ref{thm:algo-works} below).  Thus we can place $w_K$ into this edge
$(\rho_S,v_S)$, and proceed recursively on the $\rel$-classes
$L'\coloneqq K$, the induced subgraphs $G[L']$ and their corresponding
vertices $v_S\in V(S)$, which then serve as the root of the species
trees. More precisely, we identify $w_K$ with the root $u'_T$ created in
the ``next-deeper'' recursion step.  Since we alternate between vertices
$u_T$ for which no edges between vertices of distinct subtrees exist, and
vertices $v_T$ for which all such edges exist, we can label the vertices
$u_T$ with ``0'' and the vertices $v_T$ with ``1'' and obtain a  cotree for
the cograph $G$.

This recursive procedure is described more formally in
Algorithm~\ref{alg:Ru-recognition} which also describes the constructions
of an appropriate time map $\tT$ for $T$ and a reconciliation map $\mu$.
We note that we find it convenient to use as trivial case in the
recursion the situation in which the current root $u_S$ of the species
tree is a leaf rather than the condition $|L'|=1$. In this manner we avoid
the distinction between the cases $u_S\in L(S)$ and $u_S\notin L(S)$ in
the \textbf{else}-condition starting in Line~\ref{line:else}. This
results in a shorter presentation at the expense of more inner vertices
that need to be suppressed at the end in order to obtain the final tree
$T$.  We proceed by proving the correctness of
Algorithm~\ref{alg:Ru-recognition}.

\begin{ctheorem}{\ref{thm:algo-works}}
  Let $(G,\sigma)$ be a properly colored cograph, and assume that the
  triple set $\Sri(M,G)$ is compatible.  Then
  Algorithm~\ref{alg:Ru-recognition} returns a relaxed scenario
  $\scen=(T,S,\sigma,\mu,\tT,\tS)$  such that
  $\Gu(\scen)=G$ in polynomial time.
\end{ctheorem}

As a consequence of Lemma~\ref{lem:Ru-SpeciesTriple}
and~\ref{lem:propcolcograph}, and the fact that
Algorithm~\ref{alg:Ru-recognition} returns a relaxed scenario $\scen$ for a
given properly colored cograph with compatible triple set
$\Sri(G,\sigma)$, we obtain
\begin{ctheorem}{\ref{thm:characterization}}
  A graph $(G,\sigma)$ is an LDT graph if and only if it is a properly
  colored cograph and $\Sri(G,\sigma)$ is compatible.
\end{ctheorem}

Thm.~\ref{thm:characterization} has two consequences that are of
immediate interest:
\begin{ccorollary}{\ref{cor:LDTpoly}}
  LDT graphs can be recognized in polynomial time.
\end{ccorollary}
\begin{ccorollary}{\ref{cor:LDT-here}}
  The property of being an LDT graph is hereditary,
  that is, if $(G,\sigma)$ is an LDT graph then each of its vertex induced
  subgraphs is an LDT graph.
\end{ccorollary}

The relaxed scenarios $\scen$ explaining an LDT graph $(G,\sigma)$ are far
from being unique.  In fact, we can choose from a large set of trees
$(S,\tS)$ that is determined only by the triple set $\Sri(G,\sigma)$:
\begin{ccorollary}{\ref{cor:manyT}}
  If $(G=(L,E),\sigma)$ is an LDT graph with coloring $\sigma\colon L\to M$, 
  then for
  all planted trees $S$ on $M$ that display $\Sri(G,\sigma)$ there is a
  relaxed scenario $\scen=(T,S,\sigma,\mu,\tT,\tS)$ that contains $\sigma$
  and $S$ and that explains $(G,\sigma)$.
\end{ccorollary}

As shown in the Technical Part, for every LDT graph $(G,\sigma)$ there is a
relaxed scenario $\scen=(T,S,\sigma,\mu,\tT,\tS)$ explaining $(G,\sigma)$
such that $T$ displays the discriminating cotree $T_{G}$ of $G$ (cf.\ Cor.\
\ref{cor:displayed-cotree} in the Technical Part).  However, this property
is not satisfied by all relaxed scenarios that explain an $(G,\sigma)$.
Nevertheless, the latter results enable us to relate connectedness of LDT
graphs to properties of the relaxed scenarios by which it can be explained
(cf.\ Lemma~\ref{lem:Gu-connected} in Technical Part).

\subsection{Least Resolved Trees for LDT graphs}

As we have seen e.g.\ in Cor.~\ref{cor:manyT}, there are in general many
trees $S$ and $T$ forming  relaxed  scenarios $\scen$ that explain a
given LDT graph $(G,\sigma)$.  This begs the question to what extent these
trees are determined by ``representatives''. For $S$, we have seen that $S$
always displays $\Sri(G,\sigma)$, suggesting to consider the role of
$S=\Aho(\Sri(G,\sigma),M)$, where $M$ is the codomain of $\sigma$.  This
tree is least resolved in the sense that there is no  relaxed  scenario
explaining the LDT graph $(G,\sigma)$ with a tree $S'$ that is obtained
from $S$ by edge-contractions. The latter is due to the fact that any edge
contraction in $\Aho(\Sri(G,\sigma),M)$ yields a tree $S'$ that does not
display $\Sri(G,\sigma)$ any more \cite{Jansson:12}. By
Prop.~\ref{lem:Ru-SpeciesTriple}, none of the  relaxed  scenarios
containing $S'$ explain the LDT graph $(G,\sigma)$.

\begin{cdefinition}{\ref{def:LRT-LDT}}
  Let $\scen=(T,S,\sigma,\mu,\tT,\tS)$ be a relaxed scenario explaining the
  LDT graph $(G,\sigma)$. The planted tree $T$ is \emph{least resolved} for
  $(G,\sigma)$ if no relaxed scenario $(T',S',\sigma',\mu',\tT',\tS')$ with
  $T'<T$ explains $(G,\sigma)$.
\end{cdefinition}
In other words, $T$ is least resolved for $(G,\sigma)$ if no  relaxed 
scenario with a gene tree $T'$ obtained from $T$ by a series of edge
contractions explains $(G,\sigma)$.

\begin{figure}[t]
  \begin{center}
    \includegraphics[width=0.85\textwidth]{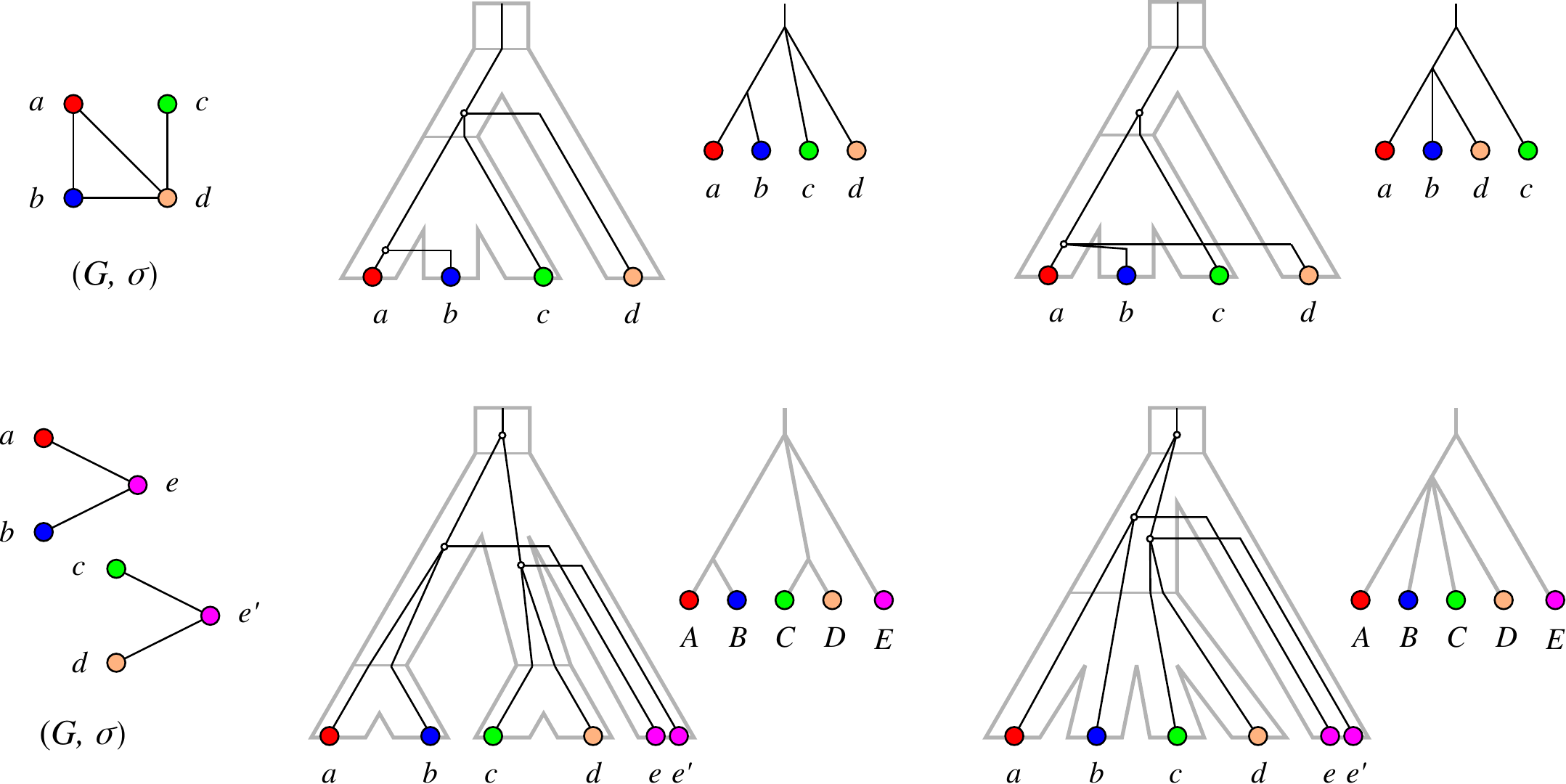}
  \end{center}
  \caption{Examples of LDT graphs $(G,\sigma)$ with multiple least resolved
    trees.  Top row: No unique least resolved gene tree. For both trees,
    contraction of the single inner edge leads to a loss of the gene triple
    $ab|c\in \Tri(G)$ (cf.\ Lemma~\ref{lem:Ru-GeneTriple}).  The species
    tree is also least resolved since contraction of its single inner edge
    leads to loss of the species triples
    $\sigma(a)\sigma(c)|\sigma(d), \sigma(b)\sigma(c)|\sigma(d)\in
    \Sri(G,\sigma)$ (cf.\ Lemma~\ref{lem:Ru-SpeciesTriple}).  Bottom row:
    No unique least resolved species tree. Both trees display the two
    necessary triples $AB|E,CD|E\in\Sri(G,\sigma)$, and are again least
    resolved w.r.t.\ these triples. The gene trees are also least resolved
    since contraction of either of its two inner edges leads e.g.\ to loss
    of one of the triples $ae|c, ce'|a\in \Tri(G)$. }
  \label{fig:LRT-not-unique}
\end{figure}

\begin{figure}[t]
  \begin{center}
    \includegraphics[width=0.85\textwidth]{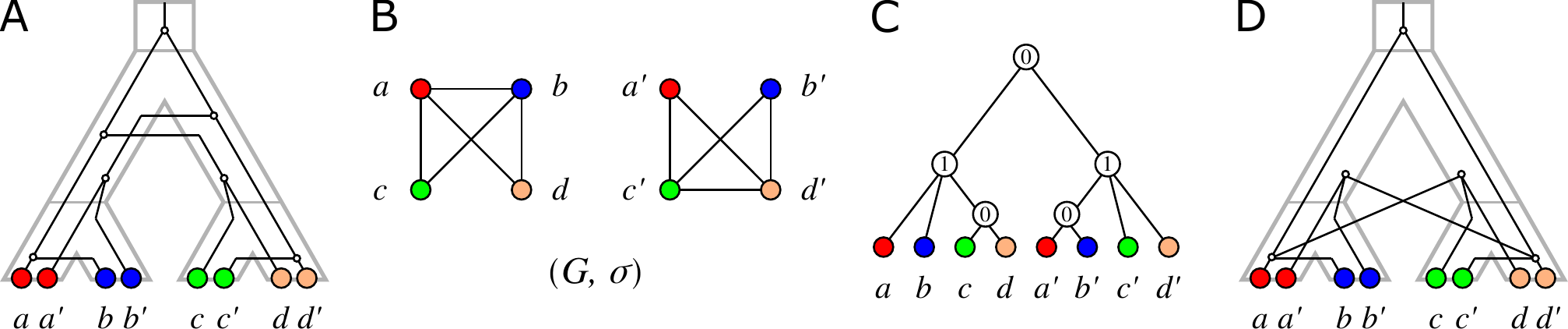}
  \end{center}
  \caption{Example of an LDT graph $(G,\sigma)$ in Panel B that is
    explained by the relaxed scenario shown in Panel~A. Here,
    $(G,\sigma)$ cannot be explained by a relaxed scenario
    $\scen=(T,S,\sigma,\mu, \tT,\tS)$ such that $T$ is the unique
    discriminating cotree (shown in panel C) for the cograph $G$, see
    Panel D and the text for further explanations.}
  \label{fig:cotree-not-resolved-enough}
\end{figure}

The examples in Fig.~\ref{fig:LRT-not-unique} show that LDT graphs are in
general not accompanied by unique least resolved trees. In the top row,
relaxed  scenarios with different least resolved gene trees $T$ and the 
same least resolved species tree $S$ explain the LDT graph $(G,\sigma)$. In the
example below, two distinct least resolved species trees exist for a given
least-resolved gene tree.

The example in Fig.~\ref{fig:cotree-not-resolved-enough} shows,
furthermore, that the unique discriminating cotree $T_G$ of an LDT graph
$(G,\sigma)$ is not always ``sufficiently resolved''.  To see this, assume
that the graph $(G,\sigma)$ in the example can be explained by a relaxed
scenario $\scen=(T,S,\sigma,\mu, \tT,\tS)$ such that $T=T_G$.  First
consider the connected component consisting of $a,b,c,d$. Since
$\lca_T(a,b)\succ_T \lca_T(c,d)$, $ab\in E(G)$ and $cd\notin E(G)$, we have
$\tS(\lca_S(\sigma(a),\sigma(b))) > \tT(\lca_T(a,b))> \tT(\lca_T(c,d))\ge
\tS(\lca_S(\sigma(c),\sigma(d)))$.  By similar arguments, the second
connected component implies
$\tS(\lca_S(\sigma(c),\sigma(d))) > \tS(\lca_S(\sigma(a),\sigma(b)))$; a
contradiction. These examples emphasize that LDT graphs constrain the
relaxed  scenarios, but are far from determining them.

\section{Horizontal Gene Transfer and Fitch Graphs}
\label{sect:HGT}

\subsection{HGT-Labeled Trees and rs-Fitch Graphs}

As alluded to in the introduction, the LDT graphs are intimately related
with horizontal gene transfer. To formalize this connection we first define
transfer edges. These will then be used to encode Walter Fitch's concept of
xenologous gene pairs \cite{Fitch:00,Darby:17} as a binary relation, and
thus, the edge set of a graph.
\begin{cdefinition}{\ref{def:HGT-label}}
  Let $\scen = (T,S,\sigma,\mu,\tT,\tS)$ be a relaxed scenario.  An edge
  $(u,v)$ in $T$ is a \emph{transfer edge} if $\mu(u)$ and $\mu(v)$ are
  incomparable in $S$. The \emph{HGT-labeling} of $T$ in $\scen$ is the
  edge labeling $\lambda_{\scen}: E(T)\to\{0,1\}$ with $\lambda(e)=1$ if and 
  only if $e$ is a transfer edge.
\end{cdefinition}
The vertex $u$ in $T$ thus corresponds to an HGT event, with $v$ denoting
the subsequent event, which now takes place in the ``recipient'' branch of
the species tree. Note that $\lambda_{\scen}$ is completely determined by
$\scen$.  In general, for a given a gene tree $T$, HGT events correspond to
a labeling or coloring of the edges of $T$. 

\begin{cdefinition}{\ref{def:FitchG}}[Fitch graph] 
  Let $(T,\lambda)$ be a tree $T$ together with a map
  $\lambda\colon E(T)\to \{0,1\}$.  The \emph{Fitch graph}
  $\gfitch(T,\lambda) = (V,E)$ has vertex set $V\coloneqq L(T)$ and edge set
  \begin{align*}
    E \coloneqq \{xy \mid  x,y\in L,
    &\text{ the unique path connecting  }
    x \text{ and } y \text{ in } T \\ 
    &\text{ contains an edge }
    e \text{ with } \lambda(e)=1. \}
  \end{align*}
\end{cdefinition}
By definition, Fitch graphs of 0/1-edge-labeled trees are loopless and
undirected. We call edges $e$ of $(T,\lambda)$ with label $\lambda(e)=1$
also 1-edges and, otherwise, 0-edges.
\begin{remark} Fitch graphs as defined here have been termed
  \emph{undirected} Fitch graphs \cite{Hellmuth:18a}, in contrast to the
  notion of the \emph{directed} Fitch graphs of 0/1-edge-labeled trees
  studied e.g.\ in \cite{Geiss:18a,Hellmuth:2019a}.
\end{remark}

\begin{cproposition}{\ref{prop:fitch}}{\cite{Hellmuth:18a,Zverovich:99}}
  The following statements are equivalent.
  \begin{enumerate}
    \item $G$ is the Fitch graph of a 0/1-edge-labeled tree.
    \item $G$ is a complete multipartite graph.
    \item $G$ does not contain $K_2+K_1$ as an induced subgraph.
  \end{enumerate}
\end{cproposition}

\begin{cdefinition}{\ref{def:rsFitchG}}[rs-Fitch graph]
  Let $\scen = (T,S,\sigma,\mu,\tT,\tS)$ be a relaxed scenario with
  HGT-labeling $\lambda_{\scen}$. We call the vertex colored graph
  $(\gfitch(\scen),\sigma) \coloneqq (\gfitch(T,\lambda_{\scen}),\sigma)$
  the \emph{Fitch graph of the scenario $\scen$.}\\
  A vertex colored graph $(G,\sigma)$ is a \emph{relaxed scenario Fitch
    graph} (\emph{rs-Fitch graph}) if there is a relaxed scenario
  $\scen = (T,S,\sigma,\mu,\tT,\tS)$  such that
  $G = \gfitch(\scen)$.
\end{cdefinition}

\begin{figure}[t]
  \begin{center}
    \includegraphics[width=0.85\textwidth]{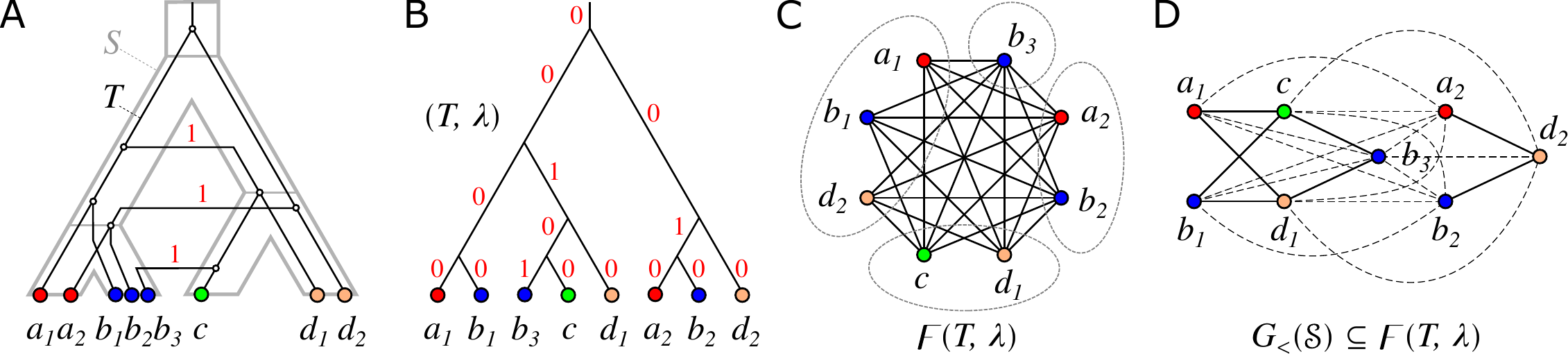}
  \end{center}
  \caption{(A) The relaxed scenario $\scen=(T,S,\sigma,\mu,\tT,\tS)$ as
    already shown in Fig.~\ref{fig:Gu-example}. (B) A 0/1-edge-labeled tree
    $(T,\lambda)$ satisfying $\lambda=\lambda_{\scen}$. (C) The
    corresponding Fitch graph $\gfitch(T,\lambda)$ drawn in a layout that
    emphasizes the property that $\gfitch(T,\lambda)$ is a complete
    multipartite graph. Independent sets are circled. (D) An alternative
    layout as in Fig.~\ref{fig:Gu-example} (top row) that emphasizes the
    relationship $\Gu(\scen)\subseteq \gfitch(\scen)=\gfitch(T,\lambda)$
    (cf.\ Thm.~\ref{thm:infer-fitch} below).  Edges that are not present in
    $\Gu(\scen)$ are drawn as dashed lines.}
  \label{fig:fitch-example}
\end{figure}

Fig.~\ref{fig:fitch-example} shows that rs-Fitch graphs are not necessarily
properly colored.  A subtle difficulty arises from the fact that Fitch
graphs of 0/1-edge-labeled trees are defined without a reference to the
vertex coloring $\sigma$, while the rs-Fitch graph is vertex colored.  This
together with Prop.~\ref{prop:fitch} implies
\begin{cfact}{\ref{obs:Fitch}}
  If $(G,\sigma)$ is an rs-Fitch graph then $G$ is a complete multipartite
  graph.
\end{cfact}
The ``converse'' of Obs.~\ref{obs:Fitch} is not true in general, as we
shall see in Thm.~\ref{thm:char-rsFitch} below. If, however, the coloring
$\sigma$ can be chosen arbitrarily, then every complete multipartite graph $G$
can be turned into an rs-Fitch graph $(G,\sigma)$ as shown in 
Prop.~\ref{prop:converse-obs-fitch}.

\begin{cproposition}{\ref{prop:converse-obs-fitch}}
  If $G$ is a complete multipartite graph, then there exists a relaxed
  scenario $\scen=(T,S,\sigma,\mu,\tT,\tS)$ such that $(G,\sigma)$ is an
  rs-Fitch graph.
\end{cproposition}

Although every complete multipartite graph can be colored in such a way
that it becomes an rs-Fitch graph (cf.\
Prop.~\ref{prop:converse-obs-fitch}), there are colored, complete
multipartite graphs $(G,\sigma)$ that are not rs-Fitch graphs, i.e., that
do not derive from a relaxed scenario (cf.\ Thm.~\ref{thm:char-rsFitch}).
We summarize this discussion in the following
\begin{cfact}{\ref{obs:01T-notScen}}
  There are (planted) 0/1-edge labeled trees $(T,\lambda)$ and colorings
  $\sigma\colon L(T)\to M$ such that there is no relaxed scenario
  $\scen = (T,S,\sigma,\mu,\tT,\tS)$ with $\lambda=\lambda_{\scen}$.
\end{cfact}
A subtle -- but important -- observation is that trees $(T,\lambda)$ with
coloring $\sigma$ for which Obs.~\ref{obs:01T-notScen} applies may still
encode an rs-Fitch graph $(\gfitch(T,\lambda),\sigma)$, see Example
\ref{ex:lst} and Fig.~\ref{fig:TreeClassesDistinct}.  The latter is due to
the fact that $\gfitch(T,\lambda) = \gfitch(T',\lambda')$ may be possible for
a different tree $(T',\lambda')$ for which there is a relaxed scenario
$\scen' = (T',S,\sigma,\mu,\tT,\tS)$ with $\lambda' = \lambda_{\scen}$.  In
this case, $(\gfitch(T,\lambda),\sigma) = (\gfitch(\scen'),\sigma)$ is an
rs-Fitch graph. We shall briefly return to these issues in the discussion
section~\ref{sect:concl}.

\begin{xmpl}
  \label{ex:lst}
  Consider the planted edge-labeled tree $(T,\lambda)$ shown in
  Fig.~\ref{fig:TreeClassesDistinct} with leaf set $L=\{a,b,b',c,d\}$,
  together with a coloring $\sigma$ where $\sigma(b)=\sigma(b')$ and
  $\sigma(a), \sigma(b), \sigma(c), \sigma(d)$
  are pairwise distinct.\\
  Assume, for contradiction, that there is a relaxed scenario
  $\scen = (T,S,\sigma,\mu,\tT,\tS)$ with
  $(T,\lambda) = (T,\lambda_{\scen})$. Hence, $\mu(v)$ and
  $\mu(b)=\sigma(b)$ as well as $\mu(u)$ and $\mu(b')=\sigma(b)$ must be
  comparable in $S$. Therefore, $\mu(u)$ and $\mu(v)$ must both be
  comparable to $\sigma(b)$ and thus, they are located on the path from
  $\rho_S$ to $\sigma(b)$. But this implies that $\mu(u)$ and $\mu(v)$ are
  comparable in $S$; a contradiction, since then
  $\lambda_{\scen}(u,v) = 0\neq \lambda(u,v) = 1$.
\end{xmpl}

\begin{figure}[t]
  \begin{center}
    \includegraphics[width=0.85\textwidth]{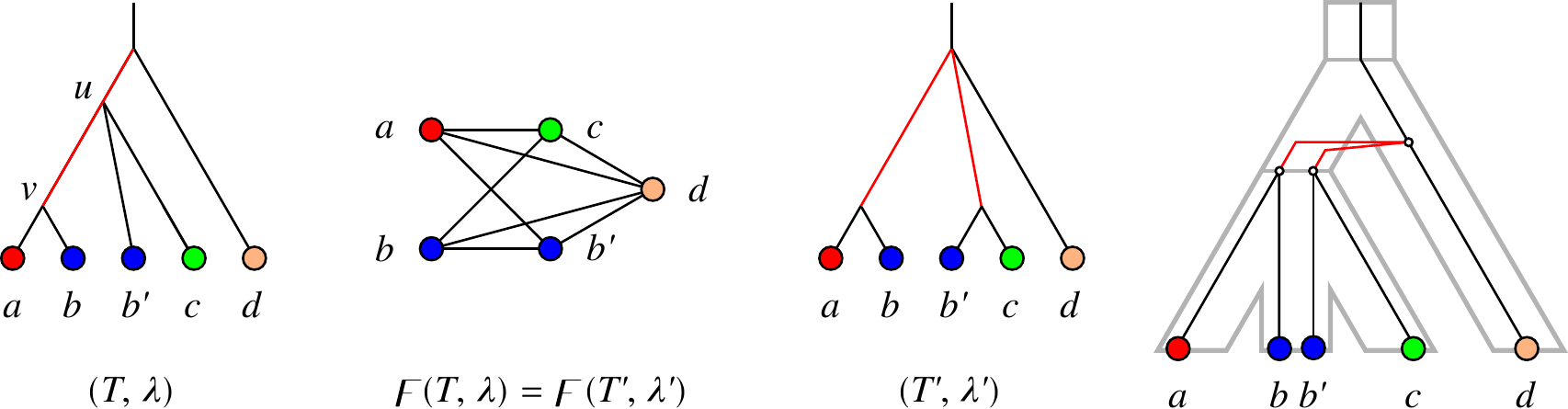}
  \end{center}
  \caption{0/1-edge-labeled tree $(T,\lambda)$ for which no relaxed
    scenario exists such that $(T,\lambda) = (T,\lambda_{\scen})$ (see
    Example~\ref{ex:lst}).  Red edges indicates 1-labeled edges.
    Nevertheless for $\gfitch\coloneqq\gfitch(T,\lambda)$ there is an
    alternative tree $(T',\lambda')$ for which a relaxed scenario
    $\scen = (T',S,\sigma,\mu,\tT,\tS)$ exists (right) such that
    $\gfitch = \gfitch(T',\lambda') = \gfitch(\scen)$. }
  \label{fig:TreeClassesDistinct}
\end{figure}

\subsection{LDT Graphs and rs-Fitch Graphs}

We proceed to investigate to what extent an LDT graph provides information
about an rs-Fitch graph.  As we shall see in
Thm.~\ref{thm:FitchRu-scenario} there is indeed a close connection between
rs-Fitch graphs and LDT graphs. We start with a useful relation between the
edges of rs-Fitch graphs and the reconciliation maps $\mu$ of their
scenarios.

\begin{clemma}{\ref{lem:independent-lca}}
  Let $\gfitch(\scen)$ be an rs-Fitch graph for some  relaxed 
  scenario $\scen$. Then,
  $ab\notin E(\gfitch(\scen))$ implies that
  $\lca_S(\sigma(a),\sigma(b)) \preceq_S \mu(\lca_T(a,b))$.
\end{clemma}

The next result shows that a subset of transfer edges can be inferred
immediately from LDT graphs:
\begin{ctheorem}{\ref{thm:infer-fitch}}
  If $(G,\sigma)$ is an LDT graph, then $G\subseteq \gfitch(\scen)$ for all
  relaxed scenarios $\scen$ that explain $(G,\sigma)$. 
\end{ctheorem}

Since we only have that $xy$ is an edge in $\gfitch(\scen)$ if the path
connecting $x$ and $y$ in the tree $T$ of $\scen$ contains a transfer edge,
Thm.~\ref{thm:infer-fitch} immediately implies

\begin{ccorollary}{\ref{cor:noHGT}}
  For every relaxed scenario $\scen=(T,S,\sigma,\mu,\tT,\tS)$ without
  transfer edges, it holds that $E(\Gu(\scen)) = \emptyset$.
\end{ccorollary}

\begin{figure}[t]
  \begin{center}
    \includegraphics[width=0.85\textwidth]{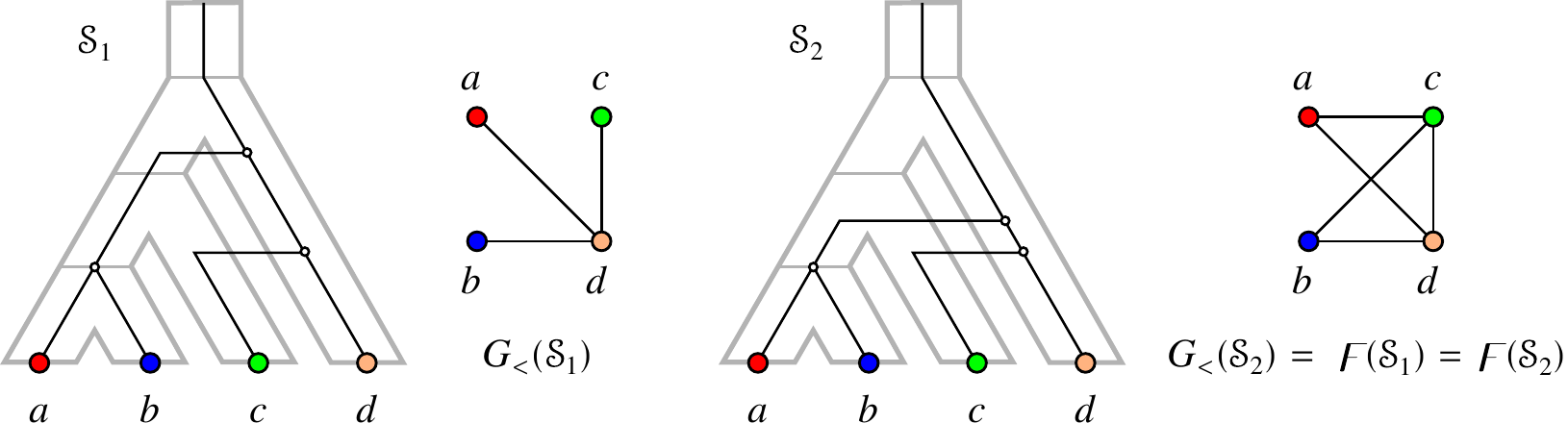}
  \end{center}
  \caption{Two relaxed scenarios $\scen_1$ and $\scen_2$ with the same
    rs-Fitch graph $\gfitch = \gfitch(\scen_1)=\gfitch(\scen_2)$ (right)
    and different LDT graphs $\Gu(\scen_1)\neq \gfitch$ and
    $\Gu(\scen_2)=\gfitch$.}
  \label{fig:Fitch-not-RU}
\end{figure}

Thm.~\ref{thm:infer-fitch} provides the formal justification for indirect
phylogenetic approaches to HGT inference that are based on the work of
\citet{Lawrence:92}, \citet{Clarke:02}, and \citet{Novichkov:04} by showing
that $(x,y)\in E(\Gu(\scen))$ can be explained only by HGT, irrespective of
how complex the true biological scenario might have been. However, it does
not cover all HGT events. Fig.\ \ref{fig:Fitch-not-RU} shows that there are
relaxed scenarios $\scen$ for which $\Gu(\scen) \neq \gfitch(\scen)$ even
though $\gfitch(\scen)$ is properly colored. Moreover, it is possible that
an rs-Fitch graph $(G,\sigma)$ contains edges $xy\in E(G)$ with
$\sigma(x)=\sigma(y)$. In particular, therefore, an rs-Fitch graph is not
always an LDT graph.

It is natural, therefore, to ask whether for every properly colored Fitch
graph there is a relaxed scenario $\scen$ such that
$\Gu(\scen) = \gfitch(\scen)$. An affirmative answer is provided by
\begin{ctheorem}{\ref{thm:FitchRu-scenario}}
  The following statements are equivalent.
  \begin{enumerate}
    \item $(G,\sigma)$ is a properly colored complete multipartite graph.
    \item There is a relaxed scenario $\scen=(T,S,\sigma,\mu,\tT,\tS)$ with
    coloring $\sigma$ such that $G=\Gu(\scen) = \gfitch(\scen)$.
    \item $(G,\sigma)$ is complete multipartite and an LDT graph. 
    \item $(G,\sigma)$ is properly colored and an rs-Fitch graph.
  \end{enumerate}
  In particular, for every properly colored complete multipartite graph
  $(G,\sigma)$ the triple set $\Sri(G,\sigma)$ is compatible.
\end{ctheorem}

relaxed  scenarios for which $(\gfitch(\scen),\sigma)$ is properly
colored do not admit two members of the same gene family that are separated
by a HGT event. While restrictive, such models are not altogether
unrealistic. Proper coloring of $(\gfitch(\scen),\sigma)$ is, in
particular, the case if every horizontal transfer is \emph{replacing},
i.e., if the original copy is effectively overwritten by homologous
recombination \cite{Thomas:05}, see also \cite{Choi:12} for a detailed case
study in \emph{Streptococcus}.  As a consequence of
Thm.~\ref{thm:FitchRu-scenario}, LDT graphs are sufficient to describe
replacing HGT. However, the incidence rate of replacing HGT decreases
exponentially with phylogenetic distance between source and target
\cite{Williams:12}, and additive HGT becomes the dominant mechanism between
phylogenetically distant organisms. Still, replacing HGTs may also be the
result of additive HGT followed by a loss of the (functionally redundant)
vertically inherited gene.

\subsection{rs-Fitch Graphs with General Colorings}

In scenarios with additive HGT, the rs-Fitch graph is no longer properly
colored and no-longer coincides with the LDT graph. Since not every
vertex-colored complete multipartite graph $(G,\sigma)$ is an rs-Fitch
graph (cf. Thm.~\ref{thm:char-rsFitch}), we ask whether an LDT
$(G,\sigma)$ that is not itself already an rs-Fitch graph imposes
constraints on the rs-Fitch graphs $(\gfitch(\scen),\sigma)$ that derive
from  relaxed  scenarios $\scen$ that explain $(G,\sigma)$. As a first step
towards this goal, we aim to characterize rs-Fitch graphs, i.e., to
understand the conditions imposed by the existence of an underlying
scenario $\scen$ on the compatibility of the collection of independent
sets $\mathcal{I}$ of $G$ and the coloring $\sigma$. As we shall see,
these conditions can be explained in terms of an auxiliary graph that we
introduce in a very general setting:
\begin{cdefinition}{\ref{def:auxfitch}}
  Let $L$ be a set, $\sigma\colon L\to M$ a map and
  $\mathcal{I}=\{I_1,\dots, I_k\}$ a set of subsets of $L$.  Then the graph
  $\auxfitch(\sigma,\mathcal{I})$ has vertex set $M$ and edges $xy$ if and
  only if $x\ne y$ and $x,y\in \sigma(I')$ for some $I'\in\mathcal{I}$. 
\end{cdefinition}
By construction $\auxfitch(\sigma,\mathcal{I'})$ is a subgraph of
$\auxfitch(\sigma,\mathcal{I})$ whenever
$\mathcal{I'}\subseteq\mathcal{I}$.  An extended version of Def.\
\ref{def:auxfitch} that contains also an edge-labeling of
$\auxfitch(\sigma,\mathcal{I})$ can be found in the Technical Part -- this
technical detail is not needed here.  As it turns out, rs-Fitch graphs
are characterized by the structure of their auxiliary graphs $\auxfitch$ as
shown in the next
\begin{ctheorem}{\ref{thm:char-rsFitch}}
  A graph $(G,\sigma)$ is an rs-Fitch graph if and only if (i) it is
  complete multipartite with independent sets
  $\mathcal{I}=\{I_1,\dots, I_k\}$, and (ii) if $k>1$, there is an
  independent set $I'\in \mathcal{I}$ such that
  $\auxfitch(\sigma,\mathcal{I}\setminus\{I'\})$ is disconnected.
\end{ctheorem}

As a consequence of Thm.~\ref{thm:char-rsFitch}, we obtain
\begin{ccorollary}{\ref{cor:auxfitch1}}
  rs-Fitch graphs can be recognized in polynomial time.
\end{ccorollary}

As for LDT graphs, the property of being an rs-Fitch graph is hereditary.
\begin{ccorollary}{\ref{cor:rsFitch-hereditary}}
  If $(G=(L,E),\sigma)$ is an rs-Fitch graph, then the colored vertex induced 
  subgaph
  $(G[W],\sigma_{|W})$ is an rs-Fitch graph for all non-empty subsets
  $W\subseteq L$.
\end{ccorollary}

\begin{figure}[ht]
  \begin{center}
    \includegraphics[width=0.85\textwidth]{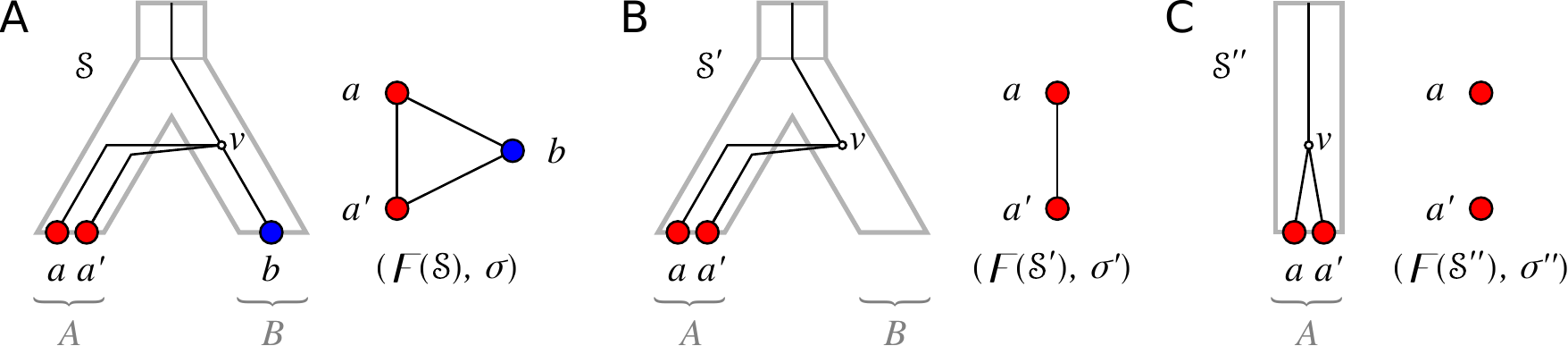}
  \end{center}
  \caption{Shown are three distinct  relaxed  scenarios $\scen$, $\scen'$
    and $\scen''$ with corresponding rs-Fitch graphs. Here
    $\sigma' = \sigma_{|\{a,a'\}}$ and
    $\sigma'' = \sigma_{|\{a,a'\},\{A\}}$ (cf.\
    Def.~\ref{def:sigma-restrictions}).  Putting
    $(G,\sigma) = (\gfitch(\scen),\sigma)$, one can observe that
    $(G[\{a,a'\}], \sigma') = (\gfitch(\scen'),\sigma')$ is an rs-Fitch
    graph.  In contrast, $\sigma''$ is restricted to the ``observable''
    part of species (consisting of $A$ alone), and
    $(G[\{a,a'\}], \sigma'')$ is not an rs-Fitch graph, see text for
    further details.}
  \label{fig:hereditary-surjective}
\end{figure}

Note, however, that Cor.~\ref{cor:rsFitch-hereditary} is not satisfied if
we restrict the codomain of $\sigma$ to the observable part of colors,
i.e., if we consider $\sigma_{|W,\sigma(W)}\colon W \to \sigma(W)$ instead
of $\sigma_{|W}\colon W\to M$, even if $\sigma$ is surjective. To see this
consider the vertex colored graph $(G,\sigma)$ with $V(G)=\{a,a',b\}$,
$E(G) = \{aa',ab,a'b\}$ and $\sigma \colon V(G)\to M = \{A,B\}$ where
$\sigma(a) = \sigma(a')=A \neq \sigma(b)=B$.  A possible relaxed scenario
$\scen$ for $(G,\sigma)$ is shown in
Fig.~\ref{fig:hereditary-surjective}(A).  The deletion of $b$ yields
$W=V(G)\setminus \{b\} = \{a,a'\}$ and the graph $(G[W],\sigma_{|W})$ for
which $\scen'$ with HGT-labeling $\lambda_{\scen'}$ as in
Fig.~\ref{fig:hereditary-surjective}(B) is a relaxed scenario that
satisfies $G[W] = \gfitch(T,\lambda_{\scen'})$.  However, if we restrict
the codomain of $\sigma$ to obtain
$\sigma_{|W,\{A\}}\colon \{a,a'\} \to \sigma(W) =\{A\}$, then there is no
relaxed scenario $\scen$ for which $G[W] = \gfitch(T,\lambda_{\scen})$,
since there is only a single species tree $S$ on $L(S)=\{A\}$
(Fig.~\ref{fig:hereditary-surjective}(C)) that consists of the single edge
$(0_T,A)$ and thus, $\mu(v)$ and $\mu(a)$ as well as $\mu(v)$ and $\mu(a')$
must be comparable in this scenario.

\subsection{Least Resolved Trees for Fitch graphs}

It is important to note that the characterization of rs-Fitch graphs in
Thm.~\ref{thm:char-rsFitch} does not provide us with a characterization of
rs-Fitch graphs that share a common  relaxed  scenario with a given LDT 
graph. As a
potential avenue to address this problem we investigate the structure of
least-resolved trees for Fitch graphs as possible source of additional
constraints.

\begin{cdefinition}{\ref{def:FLRT}}
  The edge-labeled tree $(T,\lambda)$ is \emph{Fitch-least-resolved}
  w.r.t.\ $\gfitch(T,\lambda)$, if for all trees $T'\neq T$ that are
  displayed by $T$ and every labeling $\lambda'$ of $T'$ it holds that
  $\gfitch(T,\lambda)\neq \gfitch(T',\lambda')$.
\end{cdefinition}

As shown in the Technical Part (Thm.~\ref{thm:LRT-rsFitch}),
Fitch-least-resolved trees can be characterized in terms of their
edge-labeling, a result that is very similar to the results for
``directed'' Fitch graphs of 0/1-edge-labeled trees in \cite{Geiss:18a}.
As a consequence of this characterization, Fitch-least-resolved trees can
be constructed in polynomial time. However, Fitch-least-resolved trees are
far from being unique. In particular, Fitch-least-resolved trees are only
of very limited use for the construction of relaxed scenarios
$\scen=(T,S,\sigma,\mu,\tT,\tS)$ from an underlying Fitch graph. In fact,
even though $(G,\sigma)$ is an rs-Fitch graph, Example~\ref{ex:FLRT-noScen}
in the Technical Part shows that it is possible that there is no relaxed
scenario $\scen=(T,S,\sigma,\mu,\tT,\tS)$ with HGT-labeling
$\lambda_{\scen}$ such that $(T,\lambda) = (T,\lambda_{\scen})$ for
\emph{any} of its Fitch-least-resolved trees $(T,\lambda)$.

\section{Editing Problems}
\label{sect:edit}

\subsection{Editing Colored Graphs to LDT Graphs and Fitch Graphs}

Empirical estimates of LDT graphs from sequence data are expected to
suffer from noise and hence to violate the conditions of
Thm.~\ref{thm:characterization}. It is of interest, therefore, to
consider the problem of correcting an empirical estimate $(G,\sigma)$ to
the closest LDT graph. We therefore briefly investigate the usual three
edge \emph{modification} problems for graphs: \emph{completion} only
considers the insertion of edges, for \emph{deletion} edges may only be
removed, while solutions to the \emph{editing} problem allow both
insertions and deletions, see e.g.\ \cite{Burzyn:06}.  

\begin{problem}[\PROBLEM{LDT-Graph-Modification (LDT-M)}]\ \\
  \begin{tabular}{ll}
    \emph{Input:}    & A colored graph $(G =(V,E),\sigma)$
    and an integer $k$.\\
    \emph{Question:} & Is there a subset $F\subseteq E$ such that $|F|\leq
    k$ and $(G'=(V,E\star F),\sigma)$  \\ &is an LDT graph  
    where $\star\in \{\setminus, \cup, \Delta\}$?
  \end{tabular}
\end{problem}

We write \PROBLEM{LDT-E}, \PROBLEM{LDT-C}, \PROBLEM{LDT-D} for the
editing, completion, and deletion version of \PROBLEM{LDT-M}. By virtue
of Thm.~\ref{thm:characterization}, the \PROBLEM{LDT-M} is closely
related to the problem of finding a compatible subset
$\mathscr{R}\subseteq \Sri(G_\mathscr{R},\sigma)$ with maximum
cardinality. The corresponding decision problem, \PROBLEM{MaxRTC}, is
known to be NP-complete \cite[Thm.~1]{Jansson:01}. In the technical part
we prove
\begin{ctheorem}{\ref{thm:LDT-M-NP}}
  \PROBLEM{LDT-M} is NP-complete. 
\end{ctheorem}

Even through at present it remains unclear whether rs-Fitch graphs can
be estimated directly, the corresponding graph modification problems are
at least of theoretical interest.

\begin{problem}[\PROBLEM{rs-Fitch Graph-Modification (rsF-M)}]\ \\
  \begin{tabular}{ll}
    \emph{Input:}    & A colored graph $(G =(V,E),\sigma)$
    and an integer $k$.\\
    \emph{Question:} & Is there a subset $F\subseteq E$ such that $|F|\leq
    k$ and $(G'=(V,E\star F),\sigma)$  \\ &is an rs-Fitch 
    graph  where $\star\in \{\setminus, \cup, \Delta\}$?
  \end{tabular}
\end{problem}

As above, we write \PROBLEM{rsF-E}, \PROBLEM{rsF-C}, \PROBLEM{rsF-D}
for the editing, completion, and deletion version of \PROBLEM{rsF-M}.
Since rs-Fitch graphs are complete multipartite, their complements are
disjoint unions of complete graphs. The problems \PROBLEM{rsF-M} are thus
closely related the cluster graph modification problems. Both
\PROBLEM{Cluster Deletion} and \PROBLEM{Cluster Editing} are NP-complete,
while \PROBLEM{Cluster Completion} is polynomial (by completing each
connected component to a clique, i.e., computing the transitive closure)
\cite{Shamir:04}. We obtain
\begin{ctheorem}{\ref{thm:rsF-M-NP}}
  \PROBLEM{rsF-C} and \PROBLEM{rsF-E} are NP-complete.
\end{ctheorem}
\PROBLEM{rsF-D} remains open since the complement of the transitive
closure of the complement of a colored graph $(G,\sigma)$ is not
necessarily an rs-Fitch graph. This is in particular the case if
$(G,\sigma)$ is complete multipartite but not an rs-Fitch graph.

\subsection{Editing LDT Graphs to Fitch Graphs}

Putative LDT graphs $(G,\sigma)$ can be estimated directly from sequence
(dis)similarity data.  The most direct approach was introduced by
\citet{Novichkov:04}, where, for (reciprocally) most similar genes $x$ and
$y$ from two distinct species $\sigma(x)=A$ and $\sigma(x)=B$,
dissimilarities $\delta(x,y)$ between genes and dissimilarities
$\Delta(A,B)$ of the underlying species are compared under the assumption
of a (gene family specific) clock-rate $r$, i.e., the expectation that
orthologous gene pairs satisfy $\delta(x,y)\approx r \Delta(A,B)$. In this
setting, $xy\in E(G)$ if $\delta(x,y)< r \Delta(A,B)$ at some level of
statistical significance.  The rate assumption can be relaxed to consider
rank-order statistics. For fixed $x$, differences in the orders of
$\delta(x,y)$ and $\Delta(\sigma(x),\sigma(y))$ assessed by rank-order
correlation measures have been used to identify $x$ as HGT candidate e.g.\
\cite{Lawrence:92,Clarke:02}. An interesting variation on the theme is
described by \citet{Sevillya:20}, who use relative synteny rather than
sequence similarity for the same purpose.  A more detailed account on
estimating $(G,\sigma)$ will be given elsewhere.

In contrast, it seems much more difficult to infer a Fitch graph
$(\gfitch,\sigma)$ directly from data. To our knowledge, no method for this
purpose has been proposed in the literature. However, $(\gfitch,\sigma)$ is
of much more direct practical interest because the independent sets of
$\gfitch$ determine the maximal HGT-free subsets of genes, which could be
analyzed separately by better-understood techniques. In this section, we
therefore focus on the aspects of $(\gfitch,\sigma)$ that are not captured
by LDT graphs $(G,\sigma)$. In the light of the previous section, these are
in particular non-replacing HGTs, i.e., HGTs that result in genes $x$ and
$y$ in the same species $\sigma(x)=\sigma(y)$. In this case,
$(\gfitch,\sigma)$ is no longer properly colored and thus $G\ne\gfitch$. To
get a better intuition on this case consider three genes $a$, $a'$, and $b$
with $\sigma(a)=\sigma(a')\ne\sigma(b)$ with $ab\notin E(G)$ and
$a'b\in E(G)$. By Lemma~\ref{lem:Ru-GeneTriple}, the gene tree $T$ of any
explaining relaxed scenario displays the triple $a'b|a$.
Fig.~\ref{fig:2plausibeScen} shows two relaxed scenarios with a single HGT
that explain this situation:
\begin{figure}[ht]
  \begin{center}
    \includegraphics[width=0.7\textwidth]{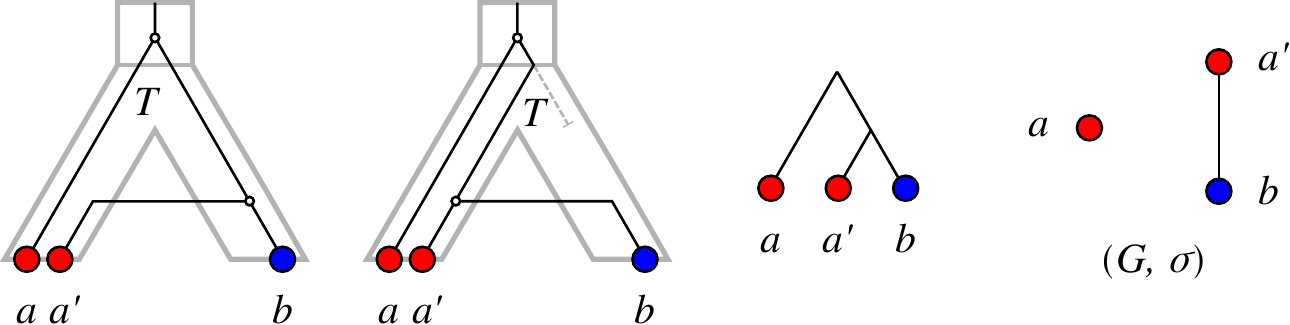}
  \end{center}
  \caption{Two relaxed scenarios with $T$ displaying the triple $a'b|a$ and
    explaining the same graph $(G,\sigma)$.}
  \label{fig:2plausibeScen}
\end{figure}
In the first, we have $aa'\in E(\gfitch)$, while the other implies
$aa'\notin E(\gfitch)$. Neither scenario is \emph{a priori} less
plausible than the other. Although the frequency of true homologous
replacement via crossover decreases exponentially with the phylogenetic
distance of donor and acceptor species \cite{Williams:12}, additive HGT
with subsequent loss of one copy is an entirely plausible scenario.

A pragmatic approach to approximate $(\gfitch,\sigma)$ is therefore to
consider the step from an LDT graph $(G,\sigma)$ to 
$(\gfitch,\sigma)$ as
a graph modification problem. First we note that
Algorithm~\ref{alg:Ru-recognition} explicitly produces a  relaxed  
scenario $\scen$
and thus implies a corresponding gene tree $T_{\scen}$ with HGT-labeling
$\lambda_{\scen}$, and thus an rs-Fitch graph $(\gfitch(\scen),\sigma)$.
However, Algorithm~\ref{alg:Ru-recognition} was designed primarily as proof
device. It produces neither a unique  relaxed  scenario nor necessarily the 
most plausible or a most parsimonious one. Furthermore, both the LDT graph
$(G,\sigma)$ and the desired rs-Fitch graph $(\gfitch,\sigma)$ are
consistent with a potentially very large number of scenarios. It thus
appears preferable to altogether avoid the explicit construction of
scenarios at this stage.

Since every LDT graph $(G,\sigma)$ is explained by some $\scen$, it is also
a spanning subgraph of the corresponding rs-Fitch graph
$(\gfitch(\scen),\sigma)$. The step from an LDT graph $(G,\sigma)$ to an
rs-Fitch graph $(\gfitch,\sigma)$ can therefore be viewed as an
edge-completion problem. The simplest variation of the problem is
\begin{problem}[Fitch graph completion]
  Given an LDT graph $(G,\sigma)$, find a minimum cardinality set $Q$ of
  possible edges such that $((V(G),E(G)\cup Q),\sigma)$ is a complete
  multipartite graph.
  \label{problem:Fcomp}
\end{problem}
A close inspection of Problem~\ref{problem:Fcomp} shows that the coloring
is irrelevant in this version, and the actual problem to be solved is the
problem \textsc{Complete Multipartite Graph Completion} with a cograph as
input. We next show that this task can be performed in linear time.
The key idea is to consider the complementary problem, i.e., the problem of
deleting a minimum set of edges from the complementary cograph
$\overline{G}$ such that the end result is a disjoint union of complete
graphs. This is known as \textsc{Cluster Deletion} problem
\cite{Shamir:04}, and is known to have a greedy solution for cographs
\cite{Gao:13}.

\begin{clemma}{\ref{lem:editing}}
  There is a linear-time algorithm to solve Problem \ref{problem:Fcomp}
  for every cograph $G$.
\end{clemma}
All maximum clique partitions of a cograph $G$ have the same sequence of
cluster sizes \cite[Thm.~1]{Gao:13}. However, they are not unique as
partitions of the vertex set $V(G)$. Thus the minimal editing set $Q$ that
needs to be inserted into a cograph to reach a complete multipartite graphs
will not be unique in general. In the Technical Part, we briefly sketch a
recursive algorithm operating on the cotree of $\overline{G}$.

However, an optimal solution to Problem~\ref{problem:Fcomp} with input
$(G,\sigma)$ does not necessarily yield an rs-Fitch graph or an rs-Fitch
graph $(\gfitch(\scen),\sigma)$ such that $G=\Gu(\scen)$, see
Fig.~\ref{fig:optimal-edit-no-rs-Fitch}.  In particular, there are LDT
graphs $(G,\sigma)$ for which more edges need to be added to obtain an
rs-Fitch graph than the minimum required to obtain a complete multipartite
graph, see Fig.~\ref{fig:optimal-rs-Fitch-no-min-compl}.

\begin{figure}[ht]
  \begin{center}
    \includegraphics[width=0.7\textwidth]{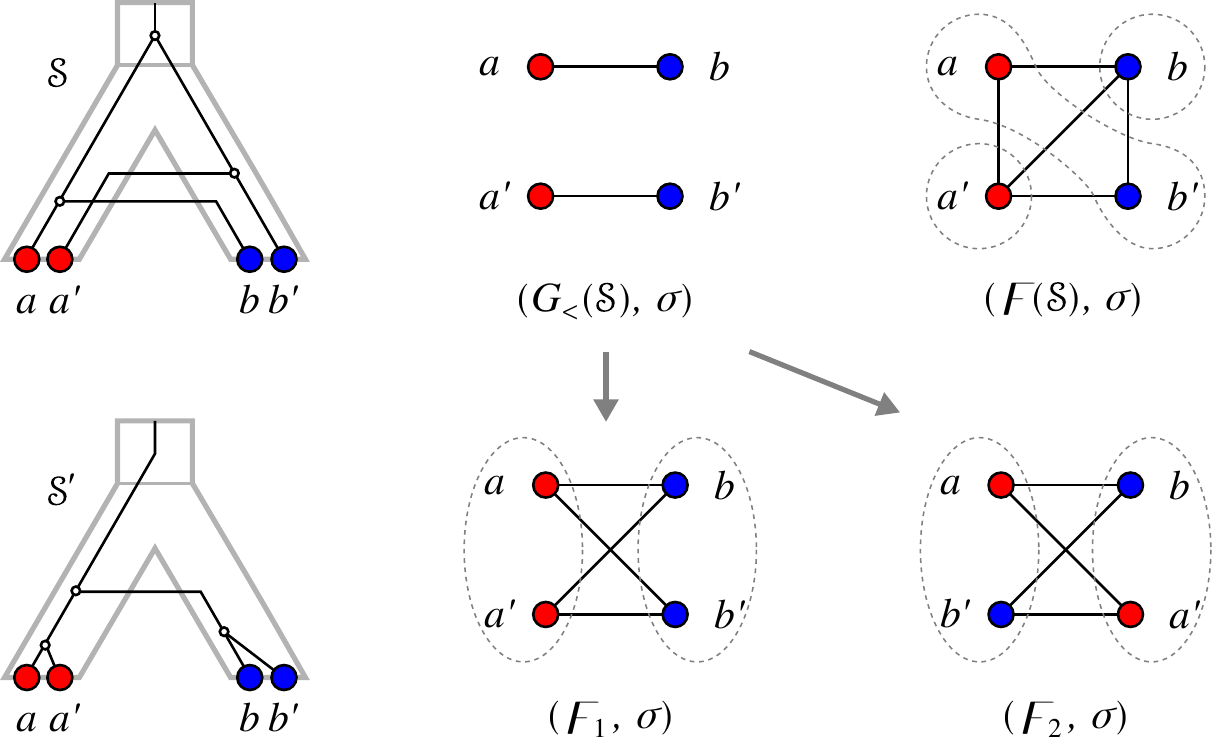}
  \end{center}
  \caption{Upper panel: a  relaxed  scenario $\scen$ with LDT graph
    $(\Gu(\scen),\sigma)$ and rs-Fitch graph $(\gfitch(\scen),\sigma)$.
    There are two minimum edge completion sets that yield the complete
    multipartite graphs $(\gfitch_1,\sigma)$ and $(\gfitch_2,\sigma)$
    (lower part). By Thm.~\ref{thm:char-rsFitch}, $(\gfitch_2,\sigma)$ is
    not an rs-Fitch graph. The graph $(\gfitch_1,\sigma)$ is an rs-Fitch
    graph for the  relaxed  scenario $\scen'$.  However, $\Gu(\scen)\ne 
    \Gu(\scen')$ for all scenarios $\scen'$ with
    $(\gfitch(\scen'),\sigma) = (\gfitch_1,\sigma)$.  To see this, note that
    the gene tree $T=((a,b),(a',b'))$ in $\scen$ is uniquely determined by
    application of Lemma~\ref{lem:2order}
    and~\ref{lem:Ru-GeneTriple}. Assume that there is any edge-labeling
    $\lambda$ such that $\gfitch(T,\lambda) = \gfitch_1$. The none-edges in
    $\gfitch_1$ imply that along the two paths from $a$ to $a'$ and $b$ to
    $b'$ there is no transfer edge, that is, there cannot be any transfer
    edge in $T$; a contradiction.}
  \label{fig:optimal-edit-no-rs-Fitch}
\end{figure}

\begin{figure}[ht]
  \begin{center}
    \includegraphics[width=0.85\textwidth]{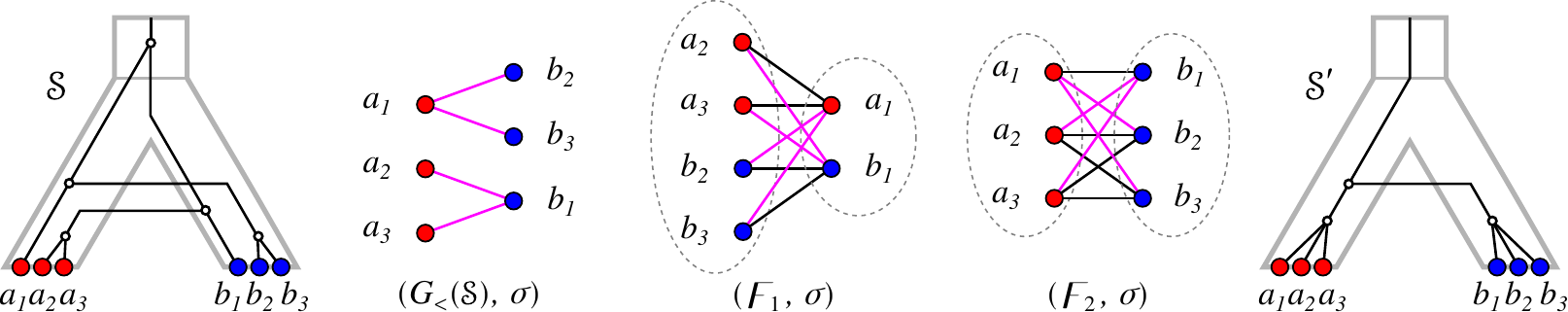}
  \end{center}
  \caption{The LDT graph $(\Gu(\scen),\sigma)$ for the  relaxed  scenario
    $\scen$ has a
    unique minimum edge completion set (as determined by full enumeration),
    resulting in the complete multipartite graph
    $(\gfitch_1,\sigma)$. However, Thm.~\ref{thm:char-rsFitch} implies that
    $(\gfitch_1,\sigma)$ is not rs-Fitch graph.  An edge completion set
    with more edges must be used to obtain an rs-Fitch graph, for instance
    $(\gfitch_2,\sigma)$, which is explained by the scenario $\scen'$.}
  \label{fig:optimal-rs-Fitch-no-min-compl}
\end{figure}

A more relevant problems for our purposes, therefore is
\begin{problem}[rs-Fitch graph completion]
  Given an LDT graph $(G,\sigma)$ find a minimum cardinality set $Q$ of
  possible edges such that $((V(G),E(G)\cup Q),\sigma)$ is an rs-Fitch graph.	
  \label{problem:rsFcomp}
\end{problem}
The following, stronger version is what we ideally would like to
solve:
\begin{problem}[strong rs-Fitch graph completion]
  Given an LDT graph $(G,\sigma)$ find a minimum cardinality set $Q$ of
  possible edges such that $\gfitch = ((V(G),E(G)\cup Q),\sigma)$ is an
  rs-Fitch graph and there is a common  relaxed 
  scenario $\scen$, that is, $\scen$
  satisfies $G = \Gu(\scen)$ and $\gfitch = \gfitch(\scen)$.
  \label{problem:strong-rsFcomp}
\end{problem}
The computational complexity of Problems \ref{problem:rsFcomp} and
\ref{problem:strong-rsFcomp} is unknown. We conjecture, however, that both
are NP-hard.  In contrast to the application of graph modification problems
to correct possible errors in the originally estimated data, the
minimization of inserted edges into an LDT graph lacks a direct biological
interpretation.  Instead, most-parsimonious solutions in terms of
evolutionary events are usually of interest in biology. In our framework,
this translates to
\begin{problem}[Min Transfer Completion]
  Let $(G,\sigma)$ be an LDT graph and $\mathbb{S}$ be the set of all
  relaxed  scenarios $\scen$ with $G=\Gu(\scen)$.  Find a  relaxed  
  scenario $\scen'\in\mathbb{S}$ that has a minimal number of transfer edges 
  among all elements in $\mathbb{S}$ and the corresponding rs-Fitch graph
  $\gfitch(\scen')$.
  \label{problem:strong-Tcomp}
\end{problem}

One way to address this problem might be as follows: Find edge-completion
sets for the given LDT graph $(G,\sigma)$ that minimize the number of
independent sets in the resulting rs-Fitch graph
$\gfitch = ((V(G),E(G)\cup Q),\sigma)$.  The intuition behind this idea is
that, in this case, the number of pairs within the individual independent
sets is maximized and thus, we get a maximized set of gene pairs without
transfer along their connecting path in the gene tree. It remains an open
question whether this idea always yields a solution for
Problem~\ref{problem:strong-Tcomp}.

\section{Simulation Results}
\label{sect:simul}

Evolutionary scenarios covering a wide range of HGT frequencies were
generated with the simulation library \texttt{AsymmeTree}
\cite{Stadler:20a}. The tool generates a planted species tree $S$ with time
map $\tS$. A constant-rate birth-death process then generates a gene tree
$(\widetilde T,\widetilde\tT)$ with additional branching events producing
copies at inner vertex $u$ of $S$ propagating to each descendant lineage of
$u$. To model HGT events, a recipient branch of $S$ is selected at
random. The simulation is event-based in the sense that each node of the
``true'' gene tree other than the planted root is one of speciation, gene
duplication, horizontal gene transfer, gene loss, or a surviving gene. 
Here, the lost as well as the surviving genes form the leaf set of 
$\widetilde T$.

We used the following parameter settings for \texttt{AsymmeTree}: Planted
species trees with a number of leaves between 10 and 50 (randomly drawn in
each scenario) were generated using the Innovation Model \cite{Keller:12}
and equipped with a time map as described in
\cite{Stadler:20a}. Multifurcations were introduced into the species tree
by contraction of inner edges with a common probability $p=0.2$ per edge to
simulate. Gene trees therefore are also not binary in general. We used
multifurcations to model the effects of limited phylogenetic resolution.
Duplication and HGT events, however, always result in bifurcations in the
gene tree $\widetilde T$.  We considered different combinations of
duplication, loss, and HGT event rates (indicated on the horizontal axis in
Figs.~\ref{fig:dataset-stats}--\ref{fig:fitch-approx-bp}). For
each combination of event rates, we simulated 1000 scenarios per event rate
combination.  Fig.~\ref{fig:dataset-stats} summarizes basic statistics
of the simulated data sets.

\begin{figure}[ht]
  \begin{center}
    \includegraphics[width=0.85\textwidth]{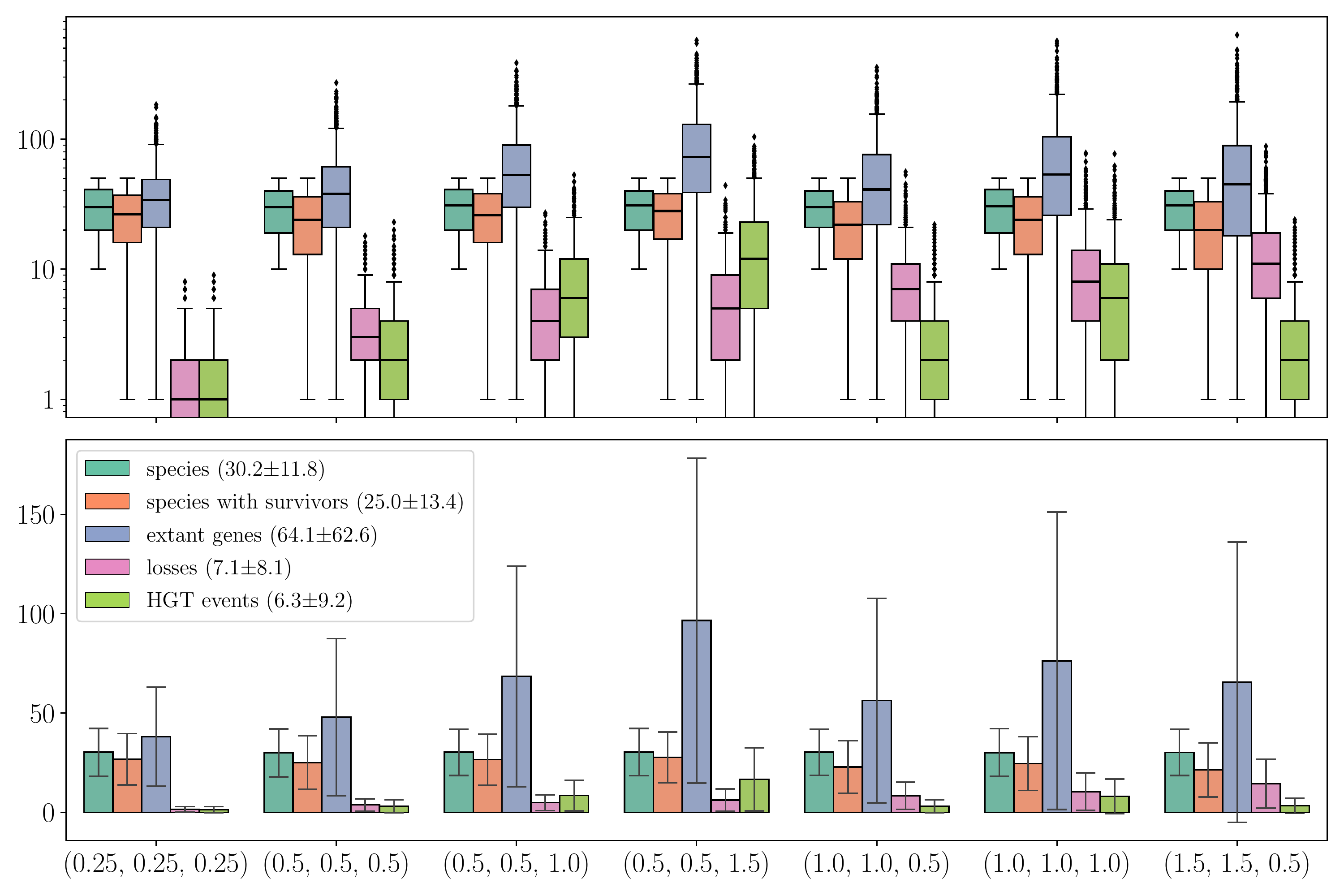}
  \end{center}
  \caption{Top panel: Distribution of the numbers of species (i.e.\ species
    tree leaves), species thereof that contain at least one surviving
    genes, surviving genes in total (non-loss leaves in the gene trees),
    loss events (loss leaves), and horizontal transfer events (inner
    vertices that are HGT events).  Bottom panel: Mean and standard
    deviation of these quantities.  The numbers in the legend indicate the
    mean and standard deviation taken over all event rate combinations.
    The tuples on the horizontal axis give the rates for duplication, loss,
    and horizontal transfer.}
  \label{fig:dataset-stats}
\end{figure}

The simulation also determines the set of surviving genes
$L\subseteq L(\widetilde{T})$, the reconciliation map
$\widetilde\mu\colon V(\widetilde{T})\to V(S)\cup E(S)$ and the coloring
$\sigma\colon L\to L(S)$ representing the species in which each surviving
gene resides.  From the true tree $\widetilde T$, the observable gene tree
$T=\widetilde{T}_{|L}$ is obtained by recursively removing leaves that
correspond to loss events, i.e.\ $L(\widetilde{T})\setminus L$, and
suppressing inner vertices with a single child and setting
$\tT(x)=\widetilde\tT(x)$ and $\mu(x)=\widetilde\mu(x)$ for all
$x\in V(T)$. This defines a relaxed scenario
$\scen=(T,S,\sigma,\mu,\tT,\tS)$.  From the scenario $\scen$, we can
immediately determine the associated HGT map $\lambda_{\scen}$, the Fitch
graph $\gfitch(\scen)$, and the LDT graph $\Gu(\scen)$.  We also consider
$\widetilde\scen=(\widetilde T, S,\sigma,\widetilde\mu,\widetilde\tT,\tS)$
which, from a formal point of view, is not a relaxed scenario, see
Fig.~\ref{fig:transfer-edges-plot}.  In this example, the gene-species
association $\sigma \colon L \to L(S)$ is not a map for the entire leaf set
$L(\widetilde T)$. Still, we can define the \emph{true LDT graph}
$\Gu(\widetilde \scen)$ and the \emph{true Fitch graph}
$\gfitch(\widetilde\scen)$ of $\widetilde\scen$ in the same way as LDT
graphs using Defs.~\ref{def:LDTgraph}, \ref{def:Gu-scen}, and
\ref{def:rsFitchG}, respectively.  Note that this does not guarantee that
every true Fitch graph is also an rs-Fitch graph. The example in
Fig.~\ref{fig:transfer-edges-plot} shows, furthermore, that
$\gfitch(\widetilde\scen)[L] \neq \gfitch(\scen)$ is possible.  For the LDT
graphs, on the other hand, we have $\Gu(\scen) = \Gu(\widetilde \scen)$
because $\widetilde \scen$ and $\scen$ are based on the same time maps.

\begin{figure}[ht]
  \begin{center}
    \includegraphics[width=0.5\textwidth]{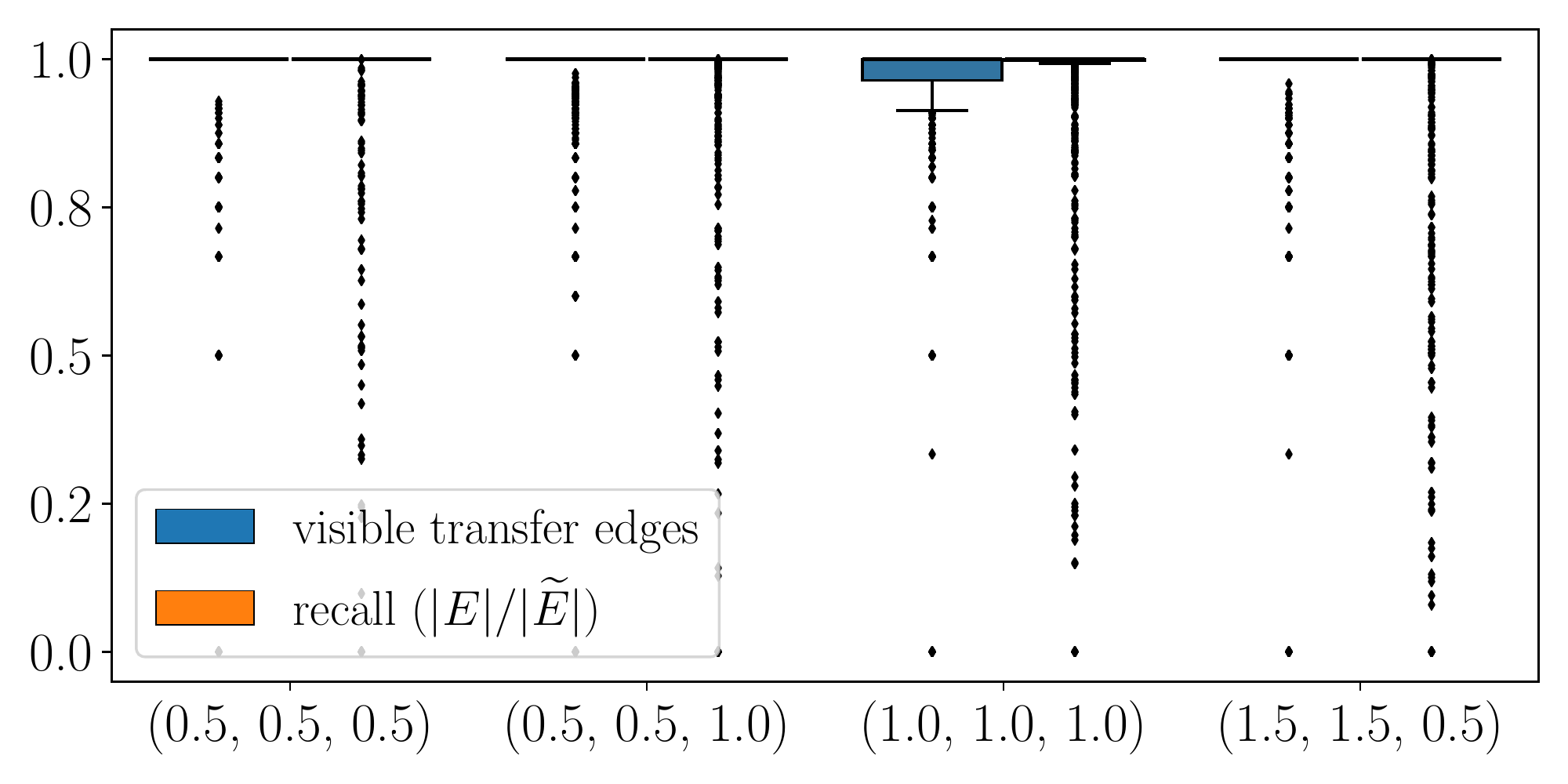}
    \includegraphics[width=0.35\textwidth]{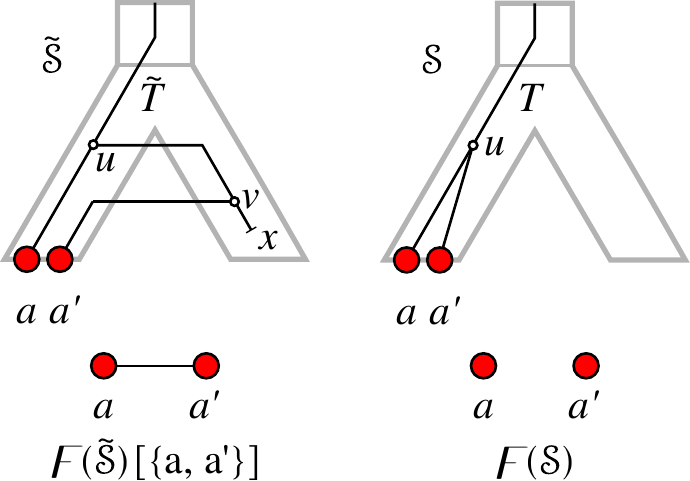}
  \end{center}
  \caption{Left: Fraction of ``visible'' transfer edges among the ``true''
    transfer edges in $T$ in the simulated scenarios, i.e., the edges that
    correspond to a path in $\widetilde T$ containing at least one transfer
    edge w.r.t.\ $\widetilde{\scen}$ (see also the explanation in the
    text). The tuples on the horizontal axis give the rates for
    duplication, loss, and horizontal transfer.  Since
    $E\coloneqq E(\gfitch(\scen)) \subseteq \widetilde{E} \coloneqq
    E(\gfitch(\widetilde\scen)[L(T)])$, we also show the ratio
    $|E|/|\widetilde E|$.  Right: A relaxed scenario
    $\scen=(T,S,\sigma,\mu,\tT,\tS)$ with an ``invisible'' transfer edge
    $(u,a')$ (as determined by the knowledge of
    $\widetilde\scen=(\widetilde
    T,S,\sigma,\widetilde\mu,\widetilde\tT,\tS)$). In this example we have
    $\gfitch(\widetilde\scen)[L(T)=\{a,a'\}] \neq \gfitch(\scen)$.}
  \label{fig:transfer-edges-plot}
\end{figure}

The distinction between the true graph $\gfitch(\widetilde\scen)[L]$ and
the rs-Fitch graph $\gfitch(\scen)$ is closely related to the definition of
transfer edges.  So far, we only took into account transfer edges $(u,v)$
in the (observable) gene trees $T$, for which $u$ and $v$ are mapped to
incomparable vertices or edges of the species trees $S$ (cf.\
Def.~\ref{def:HGT-label}).  Thus, given the knowledge of the  relaxed  
scenario $\scen=(T,S,\sigma,\mu,\tT,\tS)$, these transfer edges are in that 
sense ``visible''.  However, given
$\widetilde\scen=(\widetilde T, S,\sigma,\widetilde\mu,\widetilde\tT,\tS)$,
which still contains all loss branches, it is possible that a non-transfer
edge in $T$ corresponds to a path in $\widetilde T$ which contains a
transfer edge w.r.t.\ $\widetilde\scen$, i.e., some edge
$(u,v)\in E(\widetilde{T})$ such that $\widetilde{\mu}(u)$ and
$\widetilde{\mu}(v)$ are incomparable in $S$.  In particular, this is the
case whenever a gene is transferred into some recipient branch followed by
a back-transfer into the original branch and a loss in the recipient branch
(see Fig.~\ref{fig:transfer-edges-plot}, right).
Fig.~\ref{fig:transfer-edges-plot} shows that, in the majority of the
simulated scenarios, the HGT information is preserved in the observable
data. In fact, $\gfitch(\scen)=\gfitch(\widetilde\scen)$ in $86.7\%$ of
simulated scenarios. Occasionally, however, we also encounter scenarios in
which large fractions of the xenologous pairs are hidden from inference by
the LDT-based approach.

In the following, we will only be concerned with estimating a Fitch graph
$\gfitch(\scen)$, i.e., the graph resulting from the ``visible'' transfer
edges. These were edgeless in about $17.7\%$ of the observable scenarios
$\scen$ (all parameter combinations taken into account). In these cases the
LDT and thus also the inferred Fitch graphs are edgeless. These scenarios
were excluded from further analysis.

\begin{figure}[tb]
  \begin{center}
    \includegraphics[width=0.85\textwidth]{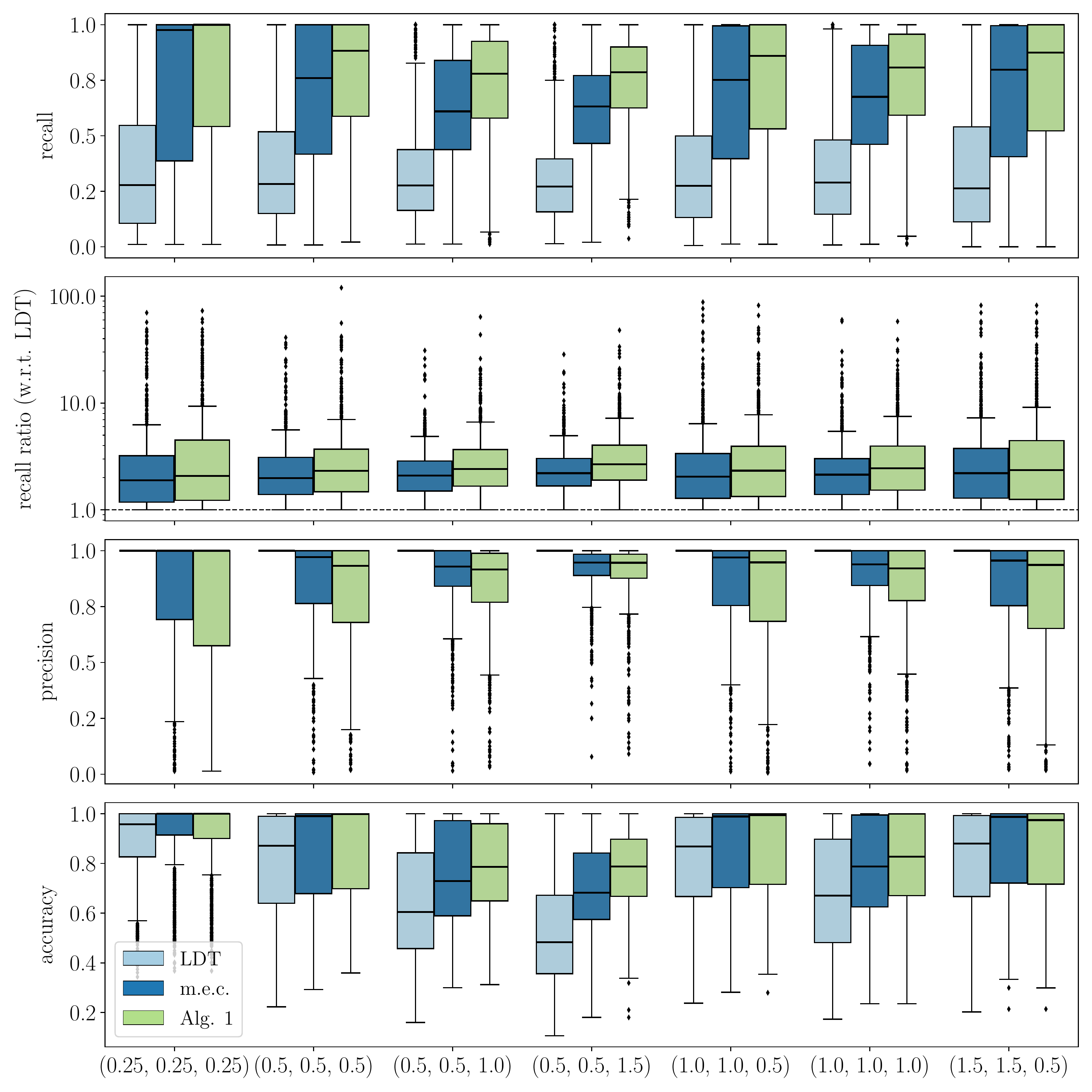}
  \end{center}
  \caption{Xenologs inferred from LDT graphs. Only observable scenarios
    $\scen$ whose LDT graph $(\Gu(\scen),\sigma)$ contains at least one edge are
    included (82.3\% of all scenarios). The tuples on the horizontal
    axis give the rates for duplication, loss, and horizontal transfer.
    Top panel: Recall.  Fraction of
    edges in $\gfitch(\scen)$ represented in $\Gu(\scen)$ (light blue). As
    an alternative, the fraction of edges in a ``minimum edge completion''
    (m.e.c.) to the ``closest'' complete multipartite graph is shown in
    dark blue. We observe a substantial increase in the fraction of
    inferred edges. The Fitch graph $\gfitch(\scen')$ obtained from the
    scenario $\scen'$ produced by Alg.~\ref{alg:Ru-recognition} with input
    $(\Gu(\scen),\sigma)$ yields an even better recall (light green).
    Second panel: Increase in the number of correctly inferred edges
    relative to the LDT graph $\Gu(\scen)$.  Third panel: Precision.  In
    contrast to LDT graphs, which by Thm.~\ref{thm:infer-fitch} cannot
    contain false positive edges, this is not the case for the estimated
    Fitch graphs obtained as m.e.c.\ and by Alg.~\ref{alg:Ru-recognition}.
    While false positive edges are typically rare, occasionally very poor
    estimates are observed. Bottom panel: Accuracy.}
  \label{fig:fitch-approx-bp}
\end{figure}

We first ask how well the LDT graph $\Gu(\scen)$ approximates the Fitch
graph $\gfitch(\scen)$. As shown in Fig.~\ref{fig:fitch-approx-bp}, the
recall is limited. Over a broad range of parameters, the LDT graph contains
about a third of the xenologous pairs. This begs the question whether the
solution of the editing Problem~\ref{problem:Fcomp}, obtained using the
exact recursive algorithm detailed in Sec.~\ref{app:edit} in the Technical
Part, leads to a substantial improvement. We find that recall indeed
increases substantially, at very moderate levels of false positives. The
editing approach achieves a median precision of well above 90\% in most
cases and a median recall of at least 60\%, it provides results that are at
the very least encouraging. We find that minimal edge completion
(Problem~\ref{problem:Fcomp}) already yields an rs-Fitch graph in the vast
majority of cases ($99.8$\%, scenarios of all parameter combinations taken
into account), even if we restrict the color set to $M'\coloneqq \sigma(L)$
(instead of $L(S)$) and thus force surjectivity of the coloring $\sigma$.
We note that the original LDT graph and the minimal edge completion may not
always be explained by a common scenario. This suggests that it will be
worthwhile to consider the more difficult editing problems for rs-Fitch
graphs with a  relaxed  scenario $\scen$ that at the same time explains the 
LDT graph.

Alg.~\ref{alg:Ru-recognition} provides a means to obtain an rs-Fitch graph
satisfying the latter constraint but without giving any guarantees for
optimality in terms of a minimal edge completion. An implementation is
available in the current release of the \texttt{AsymmeTree} package.
For the rs-Fitch graphs $\gfitch(\scen')$ of the scenarios $\scen'$
constructed by Alg.~\ref{alg:Ru-recognition} with $(\Gu(\scen),\sigma)$ as
input, we observe another moderate increase of recall when compared with
the minimal edge completion results. This comes, however, at the expense of
a loss in precision.  This is not surprising, since $\gfitch(\scen')$ by
construction contains at least as many edges as any minimal edge completion
of $\Gu(\scen)$.  Therefore, the number of both true positive and false
positive edges in $\gfitch(\scen')$ can be expected to be higher, resulting
in a higher recall and lower precision, respectively.

The recall is given by $TP / (TP + FN)$, and $|E(\gfitch(\scen))|= TP + FN$
in terms of true positives $TP$ and false negatives $FN$. Moreover,
$\Gu(\scen)$ is a subgraph of the Fitch graphs $\gfitch_{\textrm{m.e.c.}}$
and $\gfitch(\scen')$ inferred with editing or with
Alg.~\ref{alg:Ru-recognition}, respectively.  The ratio
$|E(\gfitch(\scen)) \cap E(\gfitch^*)| / |E(\gfitch(\scen) \cap
E(\Gu(\scen)))|$ with
$\gfitch^*\in \{\gfitch_{\textrm{m.e.c.}}, \gfitch(\scen') \}$ therefore
directly measures the increase in the number of correctly predicted
xenologous pairs relative to the LDT. It is equivalent to the ratio of the
respective recalls.  By construction, the ratio is always $\ge1$. This is
summarized as the second panel in Fig.~\ref{fig:fitch-approx-bp}.

\section{Discussion and Future Directions}
\label{sect:concl}

In this contribution, we have introduced  \emph{later-divergence-time 
  (LDT) graphs} as a model capturing the subset of horizontal transfer 
  detectable
through the pairs of genes that have diverged later than their respective
species. Within the setting of relaxed scenarios, 
LDT graphs $(G,\sigma)$ are
exactly the properly colored cographs with a consistent triple set 
$\Sri(G,\sigma)$. We further showed that LDT graphs
describe a sufficient set of HGT events if and only if they are complete
multipartite graphs. This corresponds to scenarios in which all HGT events
are replacing. Otherwise, additional HGT events exist that separate genes
from the same species. To better understand these, we investigated
scenario-derived rs-Fitch graphs and characterized them as those complete
multipartite graphs that satisfy an additional constraint on the coloring
(expressed in terms of an auxiliary graph).  Although the information
contained in LDT graphs is not sufficient to unambiguously determine the
missing HGT edges, we arrive at an efficiently solvable graph editing
problem from which a ``best guess'' can be obtained. To our knowledge, this
is the first detailed mathematical investigation into the power and
limitation of an implicit phylogenetic method for HGT inference.

From a data analysis point of view, LDT graphs appear to be an attractive
avenue to infer HGT in practice. While existing methods to estimate them
from (dis)similarity data certainly can be improved, it is possible to use
their cograph structure to correct the initial estimate in the same way as
orthology data \cite{Hellmuth:15a}. Although the LDT modification
problems are NP-complete (Thm.~\ref{thm:LDT-M-NP}), it does not appear
too difficult to modify efficient cograph editing heuristics
\cite{Crespelle:19x,Hellmuth:20b} to accommodate the additional coloring
constraints.

LDT graphs by themselves clearly do no contain sufficient information
to completely determine a relaxed scenario. Additional information, e.g.\
a best match graph \cite{Geiss:19a,Geiss:20b} will certainly be
required. The most direct practical use of LDT information is to infer
the Fitch graph, whose independent sets correspond to maximal HGT-free
subsets of genes. These subsets can be analyzed separately
\cite{Hellmuth:2017} using recent results to infer gene family histories,
including orthology relations from best match data
\cite{Geiss:20b,Schaller:20x}. The main remaining unresolved question is
whether the resulting HGT-free subtrees can be combined into a complete
scenario using only relational information such as best match data.  One
way to attack this is to employ the techniques used by \citet{LH:20} to
characterize the conditions under which a fully event-labeled gene tree
can be reconciled with unknown species trees. These not only resulted in
an polynomial-time algorithm but also establishes additional constraints
on the HGT-free subtrees.  An alternative, albeit mathematically less
appealing approach is to adapt classical phylogenetic methods to
accommodate the HGT-free subtrees as constraints. We suspect that best
match data can supply further, stringent constraints for this task. We
will pursue this avenue elsewhere.

Several alternative routes can be followed to obtain Fitch graphs from LDT
graphs. The most straightforward approach is to elaborate on the editing
problems briefly discussed in Sec.~\ref{sect:edit}.  A natural question
arising in this context is whether there are non-LDT edges that are shared
by all minimal completion sets $Q$, and whether these ``obligatory
Fitch-edges'' can be determined efficiently. A natural alternative is to
modify Algorithm~\ref{alg:Ru-recognition} to incorporate some form of cost
function to favor the construction of biologically plausible scenarios. In
a very different approach, one might also consider to use LDT graphs as
constraints in probabilistic models to reconstruct scenarios, see e.g.\
\cite{Sjostrand:14,Khan:16}.

Although we have obtained characterizations of both LDT graphs and rs-Fitch
graphs, many open questions and avenues for future research remain.
\paragraph{Reconciliation maps.} The notion of \emph{relaxed reconciliation
  maps} used here appears to be at least as general as alternatives that
have been explored in the literature. It avoids the concurrent definition
of event types and thus allows situations that may be excluded in a more
restrictive setting. For example, relaxed scenarios may have two or more
vertically inherited genes $x$ and $y$ in the same species with
$u\coloneqq \lca_T(x,y)$ mapping to a vertex of the species trees. In the
usual interpretation, $u$ correspond to a speciation event (by virtue of
$\mu(u)\in V^0(S)$); on the other hand, the descendants $x$ and $y$
constitute paralogs in most interpretations. Such scenarios are explicitly
excluded e.g.\ in \cite{Stadler:20a}. Lemma~\ref{lem:NoR=} suggests that
relaxed scenarios are sufficiently flexible to make it possible to replace
a scenario $\scen$ that is ``forbidden'' in response to such inconsistent
interpretations of events by an ``allowed'' scenario $\scen'$ with the same
$\sigma$ such that $\Gu(\scen)=\Gu(\scen')$. Whether this is indeed true,
or whether a more restrictive definition of reconciliation imposes
additional constraints of LDT graphs will of course need to be checked in
each case.

The restriction of a $\mu$-free scenario to a subset $L'$ of leaves of $T$
and to a subset $M'$ of leaves of $S$ is well defined as long as
$\sigma(L')\subseteq M'$.  One can also define a corresponding restriction
of the reconciliation map $\mu$. Most importantly, the deletion of some
leaves of $T$ may leave inner vertices in $T$ with only a single child,
which are then suppressed to recover a phylogenetic tree. This replaces paths
in $T$ by single edges and thus affects the definition of the HGT map
$\lambda_{\scen}$ since a path in $T$ that contains two adjacent vertices
$u_1$, $u_2$ with incomparable images $\mu(u_1)$ and $\mu(u_2)$ 
may be
replaced by an edge with comparable end points in the restricted scenario
$\scen'$. This means that HGT events may become invisible, and thus
$\gfitch(\scen')$ is not necessarily an \emph{induced} subgraph of
$\gfitch(\scen)$, but a subgraph that may lack additional edges. Note that
this is in contrast to the \emph{assumptions} made in the analysis of
(directed) Fitch graphs of 0/1-edge-labeled graphs
\cite{Geiss:18a,Hellmuth:2019a}, where the information on horizontal
transfers is inherited upon restriction of $(T,\lambda)$.

\paragraph{Observability.} The latter issue is a special case of the more
general problem with \emph{observability} of events. Conceptually, we
assume that evolution followed a \emph{true scenario} comprising discrete
events (speciations, duplications, horizontal transfer, gene losses, and
possibly other events such as hybridization which are not considered here).
In computer simulations, of course we know this true scenario, as well as
all event types. Gene loss not only renders some leaves 
invisible
but also erases the evidence of all subtrees without surviving
leaves. Removal of these vertices in general results in a non-phylogenetic
gene tree that contains inner vertices with a single child. In the absence
of horizontal transfer, this causes little problems and the
\emph{unobservable vertices} can be be removed as described in the previous
paragraph, see e.g.\ \cite{HernandezRosales:12a}. The situation is more
complicated with HGT. In \cite{Nojgaard:18a}, an HGT-vertex is deemed
observable if it has both a horizontally and a vertically inherited
descendant. In our present setting, the scenario retains an HGT-edge by
virtue of consecutive vertices in $T$ with incomparable $\mu$-images,
irrespective of whether an HGT-vertex is retained. This type of
``vertex-centered'' notion of xenology is explored further in
\cite{Hellmuth:17a}.  We suspect that these different points of view can be
unified only when gene losses are represented explicitly or when gene and
species tree trees are not required to be phylogenetic (with single-child
vertices implicating losses). Either extension of the theory, however,
requires a more systematic understanding of which losses need to be
represented and what evidence can be acquired to ``observe'' them.

\paragraph{Impact of Orthology.} Pragmatically, one would define two genes
$x$ and $y$ to be \emph{orthologs} if $\mu(\lca_T(x,y))\in V^0(S)$, i.e.,
if $x$ and $y$ are the product of a speciation event. Lemma~\ref{lem:NoR=}
implies that there is always a scenario without any orthologs that
explains a given LDT graph $(G,\sigma)$. In particular, therefore,
$(G,\sigma)$ makes no implications on orthology. Conversely, however,
orthology information is available and additional information on HGT might
become available. In a situation akin to Fig.~\ref{fig:2plausibeScen}
(with the ancestral duplication moved down to the speciation), knowing that
$a$ and $b$ are orthologs in the more restrictive sense that
$\mu(\lca_T(a,b))=\lca_S(\sigma(a),\sigma(b))$ excludes the r.h.s.\
scenario and implies that $a'$ is the horizontally inherited child, and
therefore also that $a$ and $a'$ are xenologs. This connection of
orthology and xenology will be explored elsewhere.

\paragraph{Other types of implicit phylogenetic information.} LDT graphs
are not the only conceivable type of accessible xenology information. A
large class of methods is designed to assess whether a single gene is
\emph{a xenolog}, i.e., whether there is evidence that it has been
horizontally inserted into the genome of the recipient species. The
main subclasses evaluate nucleotide composition patterns, the phyletic
distribution of best-matching genes, or combination thereof. A recent
overview can be found e.g.\ in \cite{SanchezSoto:20}. It remains an open
question how this information can be utilized in conjunction with other
types of HGT information, such as LDT graphs. It seems reasonable to
expect that it can provide not only additional constraints to infer
rs-Fitch graphs but also provides directional information that may help
to infer the directed Fitch graphs studied by
\cite{Geiss:18a,Hellmuth:2019a}. Complementarily, we may ask whether it
is possible to gain direct information on HGT edges between pairs of genes
in the same genome, and if so, what needs to be measured to extract this
information efficiently.

We also have to leave open several mathematical questions. Regarding
0/1-edge labeled trees $(T,\lambda)$, it would be of interest to know
whether there is always a relaxed scenario
$\scen = (T,S,\sigma,\mu,\tT,\tS)$ such that
$(T,\lambda) = (T,\lambda_{\scen})$ for a suitable choice of $\sigma$.
Elaborating on Thm.~\ref{thm:FitchRu-scenario}, it would be interesting
to characterize the leaf colorings $\sigma$ for $(T,\lambda)$ such that
there is a relaxed scenario $\scen$ with
$\gfitch(T,\lambda) = \gfitch(\scen)$.

\subsection*{Acknowledgments}
  We thank the three anonymous referees for their valuable comments that
  helped to significantly improve the paper.  This work was funded in part
  by the Deutsche Forschungsgemeinschaft (proj.\ CO1 within CRG 1423, no.\
  421152132 and proj. MI439/14-2), and by the Natural Sciences and
  Engineering Research Council of Canada (NSERC, grant
  RGPIN-2019-05817).

\begin{appendix}
  
  \section*{Technical Part}
  
  \section{Later-Divergence-Time Graphs}
  \label{TP:sect:LDT}
  
  \subsection{LDT Graphs and Evolutionary Scenarios}
  
  In the absence of horizontal gene transfer, the last common ancestor of two
  species $A$ and $B$ should mark the latest possible time point at which two
  genes $a$ and $b$ residing in $\sigma(a)=A$ and $\sigma(b)=B$,
  respectively, may have diverged. Situations in which this constraint is
  violated are therefore indicative of HGT.
  
  \begin{definition}[{$\mathbf{\mu}$}-free scenario]
    Let $T$ and $S$ be planted trees, $\sigma\colon L(T)\to L(S)$ be a map
    and $\tT$ and $\tS$ be time maps of $T$ and $S$, respectively, such that
    $\tT(x) = \tS(\sigma(x))$ for all $x\in L(T)$.  Then,
    $\mfscen=(T,S,\sigma,\tT,\tS)$ is called a \emph{$\mu$-free scenario}.
    \label{def:mu-free}
  \end{definition}
  
  The condition that $\tT(x) = \tS(\sigma(x))$ for all $x\in L(T)$ is mostly
  a technical convenience that makes $\mu$-free scenarios easier to
  interpret. Nevertheless, by Lemma~\ref{lem:arbitrary-tT}, given the time
  map $\tS$, one can easily construct a time map $\tT$ such that
  $\tT(x) = \tS(\sigma(x))$ for all $x\in L(T)$.  In particular, when
  constructing relaxed scenarios explicitly, we may simply choose $\tT(u)=0$
  and $\tS(x)=0$ as common time for all leaves $u\in L(T)$ and $x\in L(S)$.
  
  \begin{definition}[LDT  graph]
    For a $\mu$-free scenario $\mfscen=(T,S,\sigma,\tT,\tS)$, we define
    $\Gu(\mfscen) = \Gu(T,S,\sigma,\tT,\tS) = (V,E)$ as the graph with vertex
    set $V\coloneqq L(T)$ and edge set
    \begin{equation*}
      E \coloneqq \{ab\mid a,b\in L(T),
      \tT(\lca_T(a,b))<\tS(\lca_S(\sigma(a),\sigma(b))). \}
      \label{eq:Ru-def}
    \end{equation*}
    A vertex-colored graph $(G,\sigma)$ is a \emph{later-divergence-time graph}
    (\emph{LDT} graph), if there is a $\mu$-free scenario
    $\mfscen=(T,S,\sigma,\tT,\tS)$  such that
    $G=\Gu(\mfscen)$.  In this case, we say that $\mfscen$ \emph{explains}
    $(G,\sigma)$. \label{def:LDTgraph}
  \end{definition}
  
  It is easy to see that the edge set of $\Gu(\mfscen)$ defines an
  \emph{undirected} graph and that there are no edges of the form $aa$, since
  $\tT(\lca_T(a,a)) = \tT(a) = \tS(\sigma(a))
  =\tS(\lca_S(\sigma(a),\sigma(a)))$. Hence $\Gu(\mfscen)$ is a simple graph.
  
  By definition, every relaxed scenario $\scen =(T,S,\sigma,\mu,\tT,\tS)$
  satisfies $\tT(x)=\tS(\sigma(x))$ all $x \in L(T)$. Therefore, removing
  $\mu$ from $\scen$ yields a $\mu$-free scenario
  $\mfscen=(T,S,\sigma,\tT,\tS)$. Thus, we will use the following simplified
  notation.
  \begin{definition}
    We put $\Gu(\scen) \coloneqq \Gu(T,S,\sigma,\tT,\tS)$ for a given relaxed
    scenario $\scen =(T,S,\sigma,\mu,\tT,\tS)$ and the underlying $\mu$-free
    scenario $(T,S,\sigma,\tT,\tS)$ and say, by slight abuse of notation,
    that $\scen$ \emph{explains} $(\Gu(\scen),\sigma)$.
    \label{def:Gu-scen}
  \end{definition}
  
  \begin{lemma}
    For every $\mu$-free scenario $\mfscen=(T,S,\sigma,\tT,\tS)$, there is a
    relaxed scenario $\scen=(T,S,\sigma,\mu,\widetilde\tT,\widetilde\tS)$ for
    $T, S$ and $\sigma$ such that
    $(\Gu(\mfscen),\sigma) = (\Gu(\scen), \sigma)$.
    \label{lem:mfscen}
  \end{lemma}
  \begin{proof}
    Let $\mfscen=(T,S,\sigma,\tT,\tS)$ be a $\mu$-free scenario.  In order to
    construct a relaxed scenario
    $\scen=(T,S,\sigma,\mu,\widetilde\tT,\widetilde\tS)$ that satisfies
    $\Gu(\scen)=\Gu(\mfscen)$, we start with a time map $\widetilde\tT$ for
    $T$ satisfying $\widetilde\tT(0_T)=\max(\tT(0_T),\tS(0_S))$ and
    $\widetilde\tT(v)=\tT(v)$ for all $v\in
    V(T)\setminus\{0_T\}$. Correspondingly, we introduce a time map
    $\widetilde\tS$ for $S$ such that
    $\widetilde\tS(0_S)=\max(\tT(0_T),\tS(0_S))$ and
    $\widetilde\tS(v)=\tS(v)$ for all $v\in V(S)\setminus\{0_S\}$.  By
    construction, we have
    $t_{\max,T}\coloneqq \max\{\tT(v) \mid v\in V(T)\}=\tT(0_T)=\tS(0_S)$.
    Moreover, we have
    $t_{\min,S}\coloneqq\min\{\tS(v) \mid v\in V(S)\} \le \min\{\tT(v) \mid
    v\in V(T)\}\eqqcolon t_{\min,T}$.  To see this, we can choose $x\in V(T)$
    such that $\tT(v)=t_{\min,T}$. By the definition of time maps and
    minimality of $\tT(v)$, the vertex $x$ must be a leaf.  Hence, since
    $\mfscen$ is a $\mu$-free scenario, we have $\tT(x)=\tS(\sigma(x))$ with
    $X\coloneqq\sigma(x)\in L(S)\subset V(S)$.  Therefore, it must hold that
    $t_{\min,S}\le t_{\min,T}$.  We now define
    $P\coloneqq\{p\in V(S)\cup E(S) \mid X\preceq_{S} p\}$, i.e., the set of
    all vertices and edges on the unique path in $S$ from $0_S$ to the leaf
    $X$.  Since $\tS(X)= t_{\min,T} < t_{\max,T} = \tS(0_S)$, we find, for
    each $v\in V(T)$, \emph{either} a vertex $u\in P$ such that
    $\tT(v)=\tS(u)$ \emph{or} an edge $(u,w)\in P$ such that
    $\tS(w)<\tT(v)<\tS(u)$.  Hence, we can specify the reconciliation map
    $\mu$ by defining, for every $v\in V(T)$,
    \begin{equation*}
      \mu(v) \coloneqq
      \begin{cases}
        0_S &\text{if } v=0_T,\\
        \sigma(v) &\text{if } v\in L(T),\\
        u &\text{if there is some vertex } u\in P \textrm{ with } 
        \tT(v)=\tS(u),\\
        (u,w) &\text{if there is some edge } (u,w)\in P \textrm{ with }
        \tS(w)<\tT(v)<\tS(u).
      \end{cases}
    \end{equation*}
    For each $v\in V^0(T)$, exactly one of the two alternatives for $P$
    applies, hence $\mu$ is well-defined. It is now an easy task to verify
    that all conditions in Definitions~\ref{def:tc-map}
    and~\ref{def:relaxed-reconc} are satisfied for
    $\scen=(T,S,\sigma,\mu,\widetilde\tT,\widetilde\tS)$ by construction.
    Hence, by Def.~\ref{def:relaxed-scenario}, $\scen$ is a relaxed scenario.
    
    It remains to show that $\Gu(\mfscen)=\Gu(\scen)$.  Let $a,b\in L(T)$ be
    arbitrary. Clearly, neither $\lca_T(a,b)$ nor
    $\lca_S(\sigma(a),\sigma(b))$ equals the planted root $0_T$ or $0_S$,
    respectively. Since we have only changed the timing of the roots $0_T$ or
    $0_S$, we obtain $ab\in E(\Gu(\scen))$ if and only if
    $\widetilde\tT(\lca_T(a,b)) = \tT(\lca_T(a,b)) <
    \widetilde\tS(\lca_S(\sigma(a),\sigma(b))) =
    \tS(\lca_S(\sigma(a),\sigma(b)))$ if and only if $ab\in E(\Gu(\mfscen))$,
    which completes the proof.  
  \end{proof}
  
  \begin{theorem}
    $(G,\sigma)$ is an LDT graph if and only if there is a relaxed scenario
    $\scen = (T,S,\sigma,\mu,\tT,\tS)$  such that
    $(G,\sigma) = (\Gu(\scen),\sigma)$.
    \label{thm:LDT-scen}
  \end{theorem}
  \begin{proof}
    By definition, $(G,\sigma)$ is an LDT graph for every relaxed scenario
    $\scen$ with coloring $\sigma$ that satisfies
    $(G,\sigma) = (\Gu(\scen),\sigma)$.  Now suppose that $(G,\sigma)$ is an
    LDT graph. By definition, there is a $\mu$-free scenario
    $\mfscen=(T,S,\sigma,\tT,\tS)$ with coloring $\sigma$ such that
    $(G,\sigma)=(\Gu(\mfscen),\sigma)$.  By Lemma~\ref{lem:mfscen}, there is
    a relaxed scenario $\scen=(T,S,\sigma,\mu,\widetilde\tT,\widetilde\tS)$
    for $T, S$ and $\sigma$ such that $(G,\sigma) = (\Gu(\scen), \sigma)$.
  \end{proof}
  
  \begin{remark}
    From here on, we omit the explicit reference to Lemma~\ref{lem:mfscen}
    and Thm \ref{thm:LDT-scen} and assume that the reader is aware of the
    fact that every LDT graph is explained by some relaxed scenario $\scen$
    and that for every $\mu$-free scenario $\mfscen=(T,S,\sigma,\tT,\tS)$,
    there is a relaxed scenario $\scen$ for $T, S$ and $\sigma$ such that
    $(\Gu(\mfscen),\sigma) = (\Gu(\scen), \sigma)$.
  \end{remark}
  
  We now derive some simple properties of $\mu$-free and relaxed
  scenarios.  It may be surprising at first glance that ``the speciation
  nodes'', i.e., vertices $u\in V^0(T)$ with $\mu(u)\in V(S)$ do not play
  a special role in determining LDT graphs.
  \begin{lemma}
    For every relaxed scenario $\scen =(T,S,\sigma,\mu,\tT,\tS)$ there exists
    a relaxed scenario
    $\widetilde{\scen} = (T,S,\sigma,\widetilde\mu,\widetilde\tT,\tS)$ such
    that $\Gu(\widetilde{\scen})=\Gu(\scen)$ and for all distinct
    $x,y\in L(T)$ with $xy\notin E(\Gu(\scen))$ holds
    $\widetilde \tT(\lca_T(x,y))>\tS(\lca_S(\sigma(x),\sigma(y)))$.
    \label{lem:NoR=}
  \end{lemma}
  \begin{proof}
    For the relaxed scenario $\scen =(T,S,\sigma,\mu,\tT,\tS)$ we write
    $V^0(S)\coloneqq V(S)\setminus (L(S)\cup \{0_S\})$ and define
    \begin{align*}
      D_S &\coloneqq \{|\tS(y)-\tS(x)| \colon x,y\in V(S),\tS(x)\neq\tS(y)\}
      \textrm{,}\\
      D_T &\coloneqq \{|\tT(y)-\tT(x)| \colon x,y\in V(T),\tT(x)\neq\tT(y)\}
      \textrm{, and} \\
      D_{TS} &\coloneqq \{|\tT(x)-\tS(y)| \colon x\in V(T),\, y\in V(S),
      \tT(x)\neq \tS(y)\}.
    \end{align*}
    We have $D_S\ne\emptyset$ and $D_T\ne\emptyset$ since we do not consider
    empty trees, and thus, at least the ``planted'' edges $0_S\rho_S$ and
    $0_T\rho_T$ always exist. By construction, all values in $D_T$, $D_S$, and
    $D_{TS}$ are strictly positive. Now define
    \begin{equation*}
      \epsilon \coloneqq \frac{1}{2}\min (D_{ST}\cup D_S\cup D_T).
    \end{equation*}
    Since $D_S$ and $D_T$ are not empty, $\epsilon$ is well-defined
    and, by construction, $\epsilon>0$.
    Next we set, for all $v\in V(T)$, 
    \begin{equation*}
      \begin{split}
        \widetilde\tT(v) &\coloneqq
        \begin{cases}
          \tT(v)+\epsilon, \text{ if } v\in V^0(T)\\
          \tT(v), \text{ otherwise,} 
        \end{cases}\\
        \widetilde\mu(v) &\coloneqq
        \begin{cases}
          (\parent(x),x), \text{ if } \mu(v) = x\in V^0(S)\\
          \mu(v), \text{ otherwise.} 
        \end{cases} \\
      \end{split}
    \end{equation*}
    \begin{Xclaim}
      $\widetilde{\scen} \coloneqq (T,S,\sigma,\widetilde\mu,\widetilde\tT,\tS)$
      is a relaxed scenario.
    \end{Xclaim}
    \begin{claim-proof}
      By construction, if $\mu(v)\in (L(S)\cup \{0_S\})$ and thus,
      $\mu(v)\notin V^0(S)$, $\mu(v)$ and $\widetilde\mu(v)$
      coincide. Therefore, (G0) and (G1) are trivially 
      satisfied for
      $\widetilde\mu$. In order to show (G2), we first note that
      $\widetilde \tT(v)= \tT(v) = \tS(\sigma(v))$ holds for all $v \in L(T)$
      by Def.\ \ref{def:tc-map}.
      
      We next argue that $\widetilde\tT$ is a time map. To this end, let
      $x,y\in V(T)$ with $x\prec_T y$.  Hence, $\tT(x)<\tT(y)$ and, in
      particular, $\tT(y)-\tT(x)\geq 2\epsilon$.  Assume for contradiction
      that $\widetilde \tT(x) \geq \widetilde \tT(y)$.  This implies
      $\widetilde \tT(x) = \tT(x)+\epsilon$ and $\widetilde \tT(y) =\tT(y)$,
      since $\tT(x)<\tT(y)$ and $\epsilon>0$ always implies
      $\tT(x)+\epsilon <\tT(y) +\epsilon$ and $\tT(x) <\tT(y) +\epsilon$.
      Therefore,
      $\widetilde \tT(y) - \widetilde \tT(x) = \tT(y)-(\tT(x) + \epsilon)
      \geq \epsilon>0$ and thus, $\widetilde \tT(y) > \widetilde \tT(x)$; a
      contradiction.
      
      We continue with showing that the two time maps $\widetilde\tT$ and
      $\tS$ are time-consistent w.r.t.\ $\widetilde{\scen}$.  To see that
      Condition (C1) is satisfied, observe that, by construction,
      $\widetilde\mu(v)\in V(S)$ does hold only in case
      $\mu(v)\notin E(S)\cup V^0(S)$ and thus, $\mu(v)\in L(S) \cup
      \{0_S\}$. In this case, $\widetilde\mu(v) = \mu(v)$ and since $\mu(v)$
      satisfies (G1) we have $v\in L(T)\cup \{0_T\}$.  Thus, $v\notin V^0(T)$
      and, therefore, $\widetilde \tT(v) =\tT(v) = \tS(\mu(v))$. Therefore,
      Condition (C1) is satisfied.
      
      Now consider Condition (C2).  As argued above,
      $\widetilde \mu(v)\in E(S)$ holds for all
      $v\in V^0(T) = V(T)\setminus (L(T)\cup \{0_T\})$.  By construction,
      $\widetilde \tT(v) = \tT(v)+\epsilon$.  There are two cases:
      $\mu(v)=x\in V^0(S)$, or $\mu(v)=(y,x)\in E(S)$ with $y = \parent(x)$.
      The following arguments hold for both cases: We have
      $\widetilde \mu(v) = (y,x)\in E(S)$. Moreover,
      $\tS(x) \leq \tT(v)< \widetilde \tT(v)$ since $\tT$ and $\tS$ satisfy
      (C1) and (C2).  Furthermore, $\tT(v)<\tS(y)$ and, by construction,
      $\tS(y)-\tT(v)\geq 2\epsilon$.  This immediately implies that
      $\tS(y) \geq \tT(v) + 2\epsilon = \widetilde \tT(v) + \epsilon >
      \widetilde \tT(v)$.  In summary, $\tS(x) < \widetilde{\tT}(v) < \tS(y)$
      whenever $\widetilde \mu(v) = (y,x)\in E(S)$.  Therefore, Condition
      (C2) is satisfied for $\widetilde{\scen}$.
    \end{claim-proof}
    
    \begin{Xclaim}\label{claim:subset-edges}
      $E(\Gu(\scen)) \subseteq E(\Gu(\widetilde{\scen}))$.
    \end{Xclaim}
    \begin{claim-proof}
      Let $xy$ be an edge in $\Gu(\scen)$ and thus $x\ne y$, and set
      $v_T\coloneqq \lca_T(x,y)$ and
      $v_S\coloneqq \lca_S(\sigma(x),\sigma(y))$.  By definition, we have
      $\tT(v_T)<\tS(v_S)$.  Therefore, we have $\tS(v_S)-\tT(v_T)\in D_{TS}$
      and, hence, $\tS(v_S)-\tT(v_T)\ge 2\epsilon$.  Since $x\ne y$,
      $v_T=\lca_T(x,y)$ is an inner vertex of $T$.  By construction,
      therefore, $\widetilde{\tT}(v_T)=\tT(v_T)+\epsilon$.  The latter
      arguments together with the fact that $\tS$ remains unchanged imply
      that $\tS(v_S)-\widetilde{\tT}(v_T)\ge \epsilon>0$, and thus,
      $\widetilde{\tT}(v_T)<\tS(v_S)$.  Therefore, we conclude that $xy$ is
      an edge in $\Gu(\widetilde{\scen})$.
    \end{claim-proof}
    It remains to show
    \begin{Xclaim}
      For all distinct $x,y\in L(T)$ with $xy\notin E(\Gu(\scen))$, we have
      $\widetilde \tT(\lca_T(x,y))>\tS(\lca_S(\sigma(x),\sigma(y)))$.
    \end{Xclaim}
    \begin{claim-proof}
      Suppose $xy\notin E(\Gu(\scen))$ for two distinct $x,y\in L(T)$, and
      set $v_T\coloneqq \lca_T(x,y)$ and
      $v_S\coloneqq \lca_S(\sigma(x),\sigma(y))$.  By definition, this
      implies $\tT(v_T)\ge \tS(v_S)$.  Since $x\ne y$, we clearly have that
      $v_T=\lca_T(x,y)$ is an inner vertex of $T$, and hence,
      $\widetilde{\tT}(v_T)=\tT(v_T)+\epsilon$.  The latter two argument
      together with $\epsilon>0$ and the fact that $\tS$ remains unchanged
      imply that $\widetilde{\tT}(v_T)>\tS(v_S)$.
    \end{claim-proof}
    In particular, therefore, $xy\notin E(\Gu(\scen))$ implies that
    $xy\notin E(\Gu(\widetilde\scen))$ and therefore,
    $E(\Gu(\widetilde{\scen}))\subseteq E(\Gu(\scen))$. Together with Claim
    \ref{claim:subset-edges} and the fact that both $\Gu(\scen)$ and
    $\Gu(\widetilde\scen)$ have vertex set $L(T)$, we conclude that
    $\Gu(\scen) = \Gu(\widetilde{\scen})$, which completes the proof.  
  \end{proof}
  
  Since the relaxed scenario
  $\widetilde{\scen} = (T,S,\sigma,\widetilde\mu,\widetilde\tT,\tS)$ as
  constructed in the proof of Lemma~\ref{lem:NoR=} satisfies
  $\widetilde\mu(v)\notin V^0(S)$ we obtain
  \begin{corollary}
    For every relaxed scenario $\scen =(T,S,\sigma,\mu,\tT,\tS)$ there exists
    a relaxed scenario
    $\widetilde{\scen} = (T,S,\sigma,\widetilde\mu,\widetilde\tT,\tS)$ such
    that $\Gu(\widetilde{\scen})=\Gu(\scen)$ and
    $\widetilde\mu(v)\notin V^0(S)$ for all $v\in V(T)$.
  \end{corollary}
  Lemma~\ref{lem:NoR=}, however, does not imply that one can always find a
  relaxed scenario with a reconciliation map $\widetilde\mu$ for given trees
  $T$ and $S$ satisfying
  $\widetilde\mu(\lca_T(x,y))\succ_S\lca_S(\sigma(x),\sigma(y))$ for all
  distinct $x,y \in L(T)$ with $xy\notin E(\Gu(\scen))$, as shown in
  Example~\ref{ex:widetildeMu}.
  
  \begin{xmpl}\label{ex:widetildeMu}
    Consider the LDT graph $(\Gu(\scen),\sigma)$ with corresponding relaxed
    scenario $\scen$ as shown in Fig.~\ref{fig:counterexample-comparable-mu}.
    Note first that $v=\lca_T(a,b)=\lca_{T}(c,d)$ and $ab,cd\notin
    E(\Gu)$. To satisfy both
    $\widetilde\mu(v)\succ_S \lca_S(\sigma(a),\sigma(b))$ and
    $\widetilde\mu(v)\succ_S \lca_S(\sigma(c),\sigma(d))$, we clearly need
    that $\widetilde{\mu}(v)\succeq_S \rho_S$, and thus
    $\widetilde\tT(v)\ge \widetilde\tS(\rho_S)$.  However, $ad'\in E(\Gu)$
    and $\lca_{T}(a,d')=u$ imply that
    $\widetilde\tT(u)<\tS(\sigma(a),\sigma(d))=\tS(\rho_S)$.  Hence, we
    obtain $\widetilde\tT(u)<\tS(\rho_S)\le\widetilde\tT(v)$; a contradiction
    to $(u,v)\in E(T)$ and $\widetilde{\tT}$ being a time map for $T$.
    Therefore, there is no relaxed scenario
    $\widetilde{\scen} = (T,S,\sigma,\widetilde\mu,\widetilde\tT,\tS)$ such
    that $\Gu(\widetilde{\scen})=\Gu(\scen)$ and such that
    $\widetilde \mu(\lca_T(x,y))\succ_S\lca_S(\sigma(x),\sigma(y))$ for all
    distinct $x,y\in L(T)$ with $xy\notin E(\Gu(\scen))$.
  \end{xmpl}
  
  \begin{figure}[t]
    \begin{center}
      \includegraphics[width=0.6\textwidth]{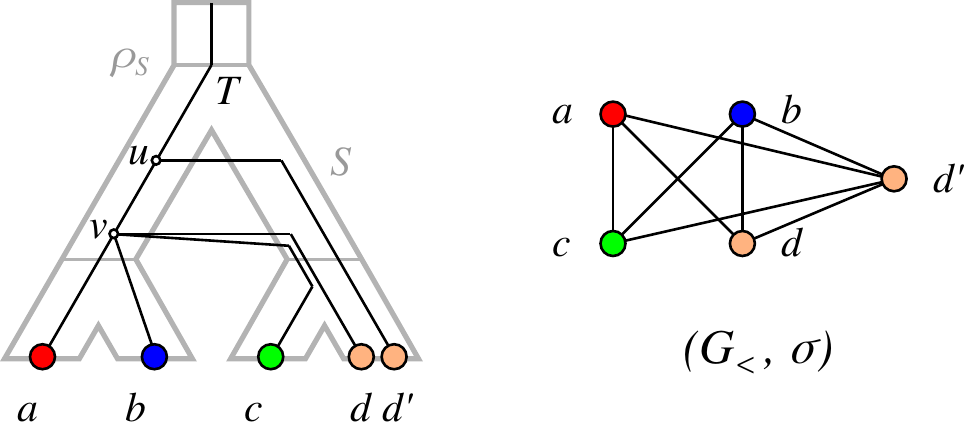}
    \end{center}
    \caption{Left a relaxed scenario $\scen =(T,S,\sigma,\mu,\tT,\tS)$ with
      corresponding graph $(\Gu(\scen),\sigma)$ (right). For
      $(\Gu(\scen),\sigma)$ there is no relaxed scenario
      $\widetilde{\scen} = (T,S,\sigma,\widetilde\mu,\widetilde\tT,\tS)$ such
      that $\Gu(\widetilde{\scen})=\Gu(\scen)$ and for all distinct
      $x,y\in L(T)$ with $xy\notin E(\Gu(\scen))$ it holds that
      $\widetilde \mu(\lca_T(x,y))\succ_S\lca_S(\sigma(x),\sigma(y))$, see
      Example \ref{ex:widetildeMu}.}
    \label{fig:counterexample-comparable-mu}
  \end{figure}
  
  For the special case that the graph under consideration has no edges we
  have
  \begin{lemma}
    For an edgeless graph $G$ and for any choice of ~$T$ and $S$ with
    $L(T)=V(G)$ and $\sigma(L(T))=L(S)$ there is a relaxed scenario
    $\scen =(T,S,\sigma,\mu,\tT,\tS)$ that satisfies $G = \Gu(\scen)$.
    \label{lem:Rempty}
  \end{lemma}
  \begin{proof}
    Given $T$ and $S$ we construct a relaxed scenario as follows. Let $\tS$
    be an arbitrary time map on $S$. Then we can choose $\tT$ such that
    $\tS(\rho_S)<\tT(u)<\tS(0_S)$ for all $u\in V^0(T)$. Each leaf
    $u\in L(T)$ then has a parent in $T$ located above the last common
    ancestor $\rho_S$ of all species in which case $\Gu(\scen)$ is edgeless.
  \end{proof}
  Lemma~\ref{lem:Rempty} is reminiscent of the fact that for DL-only
  scenarios any given gene tree $T$ can be reconciled with an arbitrary
  species tree as long as $\sigma(L(T))=L(S)$ \cite{Guigo:96,Geiss:20b}.
  
  \subsection{Properties of LDT Graphs} 
  
  \begin{proposition}
    Every LDT graph $(G,\sigma)$ is properly colored.
    \label{prop:properCol}
  \end{proposition}
  \begin{proof}
    Let $\mfscen=(T,S,\sigma,\tT,\tS)$ be a $\mu$-free scenario such that
    $(G,\sigma) = (\Gu(\mfscen),\sigma)$ and recall that every $\mu$-free
    scenario satisfies $\tT(x) = \tS(\sigma(x))$ for all $x\in L(T)$ with
    $\sigma(x)\in L(S)$.  Let $a,b\in L(T)$ be distinct and suppose that
    $\sigma(a)=\sigma(b)=A$.  Since $a$ and $b$ are distinct we have
    $a,b\prec_T \lca_T(a,b)$ and hence, by Def.~\ref{def:time-map},
    $\tT(a) < \tT(\lca_T(a,b))$.  This implies that
    $\tT(a) = \tS(A) = \tS(\lca_S(A,A)) <\tT(\lca_T(a,b))$.  Therefore,
    $ab\notin E(G)$.  Consequently, $ab\in E(G)$ implies
    $\sigma(a)\neq \sigma(b)$, which completes the proof.  
  \end{proof}
  
  Extending earlier work of
  Dekker (\citeyear{Dekker:86ma}), Bryant and Steel
  (\citeyear{Bryant:95}) derived conditions under which two triples $r_1,r_2$
  imply a third triple $r_3$ that must be displayed by any tree that displays
  $r_1,r_2$. In particular, we make frequent use of the following
  \begin{lemma}
    If a tree $T$ displays $xy|z$ and $zw|y$ then $T$ displays $xy|w$ and
    $zw|x$. In particular $T_{|\{x,y,z,w\}} = ((x,y),(z,w))$ (in \emph{Newick}
    format).
    \label{lem:2order}
  \end{lemma}
  
  \begin{definition}
    For every graph $G=(L,E)$, we define the set of triples on $L$
    \begin{equation*}
      \Tri(G) \coloneqq \{xy|z \; \colon
      x,y,z\in L \text{ are pairwise distinct, }
      xy\in E,\; xz,yz\notin E\} \,.
    \end{equation*}
    If $G$ is endowed with a coloring $\sigma\colon L\to M$ we also define a
    set of color triples
    \begin{align*}
      \Sri(G,\sigma) \coloneqq \{\sigma(x)\sigma(y)|\sigma(z)\; \colon 
      & x,y,z\in L,\, \sigma(x),\sigma(y),\sigma(z) \text{ are pairwise 
        distinct},\\
      &xz, yz\in E,\; xy\notin E\}.
    \end{align*}
    \label{def:infoTriples}
  \end{definition}
  
  \begin{lemma}
    If a graph $(G,\sigma)$ is an LDT graph then $\Sri(G,\sigma)$ is
    compatible and $S$ displays $\Sri(G,\sigma)$ for every $\mu$-free
    scenario $\mfscen=(T,S,\sigma,\tT,\tS)$ that explains $(G,\sigma)$.
    \label{lem:Ru-SpeciesTriple}
  \end{lemma}  
  \begin{proof}
    Suppose that $(G=(L,E),\sigma)$ is an LDT graph and let
    $\mfscen=(T,S,\sigma,\tT,\tS)$ be a $\mu$-free scenario that explains
    $(G,\sigma)$. In order to show that $\Sri(G,\sigma)$ is compatible it
    suffices to show that $S$ displays every triple in $\Sri(G,\sigma)$.
    
    Let $AB|C\in \Sri(G,\sigma)$.  By definition, $A,B,C$ are pairwise
    distinct and there must be vertices $a,b,c\in L$ with $\sigma(a)=A$,
    $\sigma(b)=B$, and $\sigma(c)=C$ such that $ab \notin E$ and
    $bc,ac \in E$.  First, $ab \notin E$ and $bc,ac \in E$ imply
    $\tT(\lca_T(a,b))\geq \tS(\lca_S(A,B))$,
    $\tT(\lca_T(b,c))<\tS(\lca_S(B,C))$, and
    $\tT(\lca_T(a,c))<\tS(\lca_S(A,C))$. Moreover, for any three vertices
    $a,b,c$ in $T$ it holds that
    $1 \leq |\{\lca_T(a,b),\lca_T(a,c),\lca_T(b,c)\}| \leq 2$.
    
    Therefore we have to consider the following four cases: (1)
    $u\coloneqq \lca_T(a,b)=\lca_T(b,c)=\lca_T(a,c)$, (2)
    $u\coloneqq \lca_T(a,b)=\lca_T(a,c)\neq\lca_T(b,c)$ and (3)
    $u\coloneqq\lca_T(a,b)=\lca_T(b,c)\neq\lca_T(a,c)$, (4)
    $\lca_T(a,b)\neq u\coloneqq \lca_T(b,c)=\lca_T(a,c)$.  Note, for any
    three vertices $x,y,z$ in $T$, $\lca_T(x,y)\neq \lca_T(x,z)=\lca_T(y,z)$
    implies that $\lca_T(x,y)\prec_T \lca_T(x,z)=\lca_T(y,z)$. In Cases (1)
    and (2), we find $\tS(\lca_S(A,C)) > \tT(u) \geq \tS(\lca_S(A,B))$.
    Together with the fact that $\lca_S(A,C)$ and $\lca_S(A,B)$ are
    comparable in $S$, this implies that $AB|C$ is displayed by $S$.  In Case
    (3), we obtain $\tS(\lca_S(B,C)) > \tT(u) \geq \tS(\lca_S(A,B))$ and, by
    analogous arguments, $AB|C$ is displayed by $S$.  Finally, in Case (4),
    the tree $T$ displays the triple $ab|c$.  Thus,
    $\tS(\lca_S(A,B))\leq \tT(\lca_T(a,b)) < \tT(u) < \tS(\lca_S(A,C))$.
    Again, $AB|C$ is displayed by $S$.  
  \end{proof}
  
  The next lemma shows that induced $K_2+K_1$ subgraphs in LDT graphs implies
  triples that must be displayed by $T$.
  
  \begin{lemma}
    If $(G,\sigma)$ is an LDT graph, then $\Tri(G)$ is compatible and $T$
    displays $\Tri(G)$ for every $\mu$-free scenario
    $\mfscen=(T,S,\sigma,\tT,\tS)$ that explains $(G,\sigma)$.
    \label{lem:Ru-GeneTriple}
  \end{lemma}
  \begin{proof}
    Suppose that $(G=(L,E),\sigma)$ is an LDT graph and let
    $\mfscen=(T,S,\sigma,\tT,\tS)$ be a $\mu$-free scenario that explains
    $(G,\sigma)$. In order to show that $\Tri(G)$ is compatible it suffices
    to show that $T$ displays every triple in $\Tri(G,\sigma)$.
    
    Let $ab|c \in \Tri(G)$. By definition, $a,b,c\in L(T)$ are distinct, and
    $ab\in E$ and $ac,bc\not\in E$. Since $ab \in E$, we have
    $A\coloneqq\sigma(a)\neq \sigma(b)\eqqcolon B$ by 
    Prop.~\ref{prop:properCol}.
    
    There are two cases, either $\sigma(c)\in \{A,B\}$ or not.  Suppose first
    that w.l.o.g.\ $\sigma(c)=A$.  In this case, $ab \in E$ and $bc \notin E$
    together imply $\tT(\lca_T(a,b))<\tS(\lca_S(A,B))\leq \tT(\lca_T(b,c))$.
    This and the fact that $\lca_T(a,b)$ and $\lca_T(b,c)$ are comparable in
    $T$ implies that $T$ displays $ab|c$.
    
    Suppose now that $\sigma(c)=C\notin \{A,B\}$.  We now consider the four
    possible topologies of $S'=S_{|ABC}$: (1) $S'$ is a star, (2)
    $S'=AB|C$, (3) $S'=AC|B$, and (4) $S'=BC|A$.
    
    In Cases~(1), (2) and (4), we have
    $\tS(\lca_S(A,B)) \leq \tS(\lca_S(A,C))$, where equality holds only in
    Cases~(1) and (4). This together with $ab \in E$ and $ac \notin E$
    implies
    $\tT(\lca_T(a,b))<\tS(\lca_S(A,B)) \leq \tS(\lca_S(A,C)) \leq
    \tT(\lca_T(a,c))$.  This and the fact that $\lca_T(a,b)$ and
    $\lca_T(a,c)$ are comparable in $T$ implies that $T$ displays $ab|c$.  In
    Case (3), $ab \in E$ and $bc \notin E$ imply
    $\tT(\lca_T(a,b))<\tS(\lca_S(A,B)) = \tS(\lca_S(B,C)) \leq
    \tT(\lca_T(b,c))$.  By analogous arguments as before, $T$ displays
    $ab|c$.  
  \end{proof}
  
  We note, finally, that the Aho graph of the triple set $[\Tri(G),L]$ in a
  sense recapitulates $G$. More precisely, we have:
  \begin{proposition}
    Let $(G=(L,E),\sigma)$ be a vertex-colored graph. If for all edges $xy\in E$
    there is a vertex $z$ such that $xz,yz\notin E$ (and thus, in particular,
    in case that $G$ is disconnected), then $[\Tri(G),L]=G$.
    \label{prop:Aho-G}
  \end{proposition}
  \begin{proof}
    Clearly, the vertex sets of $[\Tri(G),L]$ and $G$ are the same, that
    is, $L$.  Let $xy\in E$ and thus, we have $x\neq y$.  There is a vertex
    $z\neq x,y$ in $G$ with $xz,yz\notin E$ if and only if
    $xy|z\in \Tri(G)$ and thus, if and only if $xy$ is an edge in
    $[\Tri(G),L]=G$.  
  \end{proof}
  
  \begin{definition}
    For a vertex-colored graph $(G,\sigma)$, we will use the shorter notation
    $x_1-x_2-\dots-x_n$ and $X_1-X_2-\dots-X_n$ for a path $P_n$ that is
    induced by the vertices $\{x_i\mid 1\leq i\leq n\}$ with colors
    $\sigma(x_i)=X_i$, $1\leq i\leq n$ and edges $x_ix_{i+1}$,
    $1\leq i\leq n-1$.
  \end{definition}
  
  \begin{lemma}
    Every LDT graph $(G,\sigma)$ is a properly colored cograph.
    \label{lem:propcolcograph}
  \end{lemma}
  \begin{proof}
    Let $\mfscen=(T,S,\sigma,\tT,\tS)$ be a $\mu$-free scenario that explains
    $(G,\sigma)$.  By Prop.~\ref{prop:properCol}, $(G,\sigma)$ is properly
    colored.  To show that $G=(L,E)$ is a cograph it suffices to show that
    $G$ does not contain an induced path on four vertices (cf.\ 
    Prop.~\ref{prop:cograph}).  Hence, assume for contradiction that $G$ 
    contains
    an induced $P_4$.
    
    First we observe that for each edge $ab$ in this $P_4$ it holds that
    $\sigma(a)\neq \sigma(b)$ since, otherwise, by Prop.~\ref{prop:properCol}, 
    $ab\notin E$. Based on possible colorings of the
    $P_4$ w.r.t.\ $\sigma$ and up to symmetry, we have to consider four
    cases: (1) $A-B-C-D$, (2) $A-B-C-A$, (3) $A-B-A-C$ and (4) $A-B-A-B$.
    
    In Case (1) the $P_4$ is of the form $a-b-c-d$ with $\sigma(a)=A$,
    $\sigma(b)=B$, $\sigma(c)=C$, $\sigma(d)=D$.  By Lemma
    \ref{lem:Ru-SpeciesTriple}, the species tree $S$ must display both $AC|B$
    and $BD|C$. Hence, by Lemma~\ref{lem:2order}, $S_{|ABCD} = ((A,C),(B,D))$
    in \emph{Newick} format. Let
    $x \coloneqq \lca_S(A,B,C,D) = \rho_{S_{|ABCD}}$.  Note, $x$
    ``separates'' $A$ and $C$ from $B$ and $D$.  Now, $ab\in E$ and
    $ad\notin E$ implies that $\tT(\lca_T(a,b))<\tS(x)\leq \tT(\lca_T(a,d))$.
    This and the fact that $\lca_T(a,b)$ and $\lca_T(a,d)$ are comparable in
    $T$ implies that $T$ displays $ab|d$.  Similarly, $cd\in E$ and
    $ad\notin E$ implies that $T$ displays $cd|a$ is displayed by $T$. By
    Lemma~\ref{lem:2order}, $T_{|abcd} = ((a,b),(c,d))$.  Let
    $y \coloneqq \lca_T(a,b,c,d) = \rho_{T_{|abcd}}$.  Now, $bc\in E$,
    $\lca_T(b,c)=y$, and $\lca_S(B,C)=x$ implies $\tT(y)<\tS(x)$.  This and
    $\lca_T(a,d)=y$ and $\lca_S(A,D)=x$ imply that $ad\in E$, and thus
    $a,b,c,d$ do not induce a $P_4$ in $G$; a contradiction.
    
    Case (2) can be directly excluded, since Lemma~\ref{lem:Ru-SpeciesTriple}
    implies that, in this case, $S$ must display $AC|B$ and $AB|C$; a
    contradiction.
    
    Now consider Case (3), that is, the $P_4$ is of the form $a-b-a'-c$ with
    $\sigma(a)=\sigma(a')=A$, $\sigma(b)=B$ and $\sigma(c)=C$.  By Lemma
    \ref{lem:Ru-SpeciesTriple}, the species tree $S$ must display $BC|A$ and
    thus $x\coloneqq\lca_S(A,B)=\lca_S(A,C)$.  Since $ab\in E$ and
    $ac\notin E$ we observe $\tT(\lca_T(a,b))<\tS(x)\leq \lca_T(a,c)$ and, as
    in Case (1) we infer that $T$ displays $ab|c$. By similar arguments,
    $a'c\in E$ and $ac\notin E$ implies that $T$ displays $a'c|a$. By Lemma
    \ref{lem:2order}, $T_{|abcd} = ((a,b),(a',c))$ and thus,
    $y\coloneqq \lca_T(a',b) = \lca_T(a,c)$ and $a'b\in E$ implies that
    $\tT(y)<\tS(x)$.  Since $y= \lca_T(a,c)$ and
    $\tT(y)<\tS(x)=\tS(\lca_S(A,C))$, we can conclude that $ac\in E$. Hence,
    $a,b,c,d$ do not induce a $P_4$ in $G$; a contradiction.
    
    In Case (4) the $P_4$ is of the form $a-b-a'-b'$ with
    $\sigma(a)=\sigma(a')=A$ and $\sigma(b)=\sigma(b')=B$.  Now,
    $ab,a'b'\in E$ and $ab'\notin E$ imply that
    $\tT(\lca_T(a,b)), \tT(\lca_T(a',b')) < \tS(\lca_S(A,B))\leq
    \tT(\lca_T(a,b'))$.  Hence, by similar arguments as above, $T$ must
    display $ab|b'$ and $a'b'|a$.  By Lemma~\ref{lem:2order},
    $T_{abcd} = ((a,b),(a',b'))$ and thus,
    $y\coloneqq \lca_T(a'b) = \lca_T(a,b')$.  However, $a'b\notin E$ implies
    that $\tT(y)<\tS(\lca_S(A,B))$; a contradiction to
    $\tS(\lca_S(A,B))\leq \tT(\lca_T(a,b'))$.  
  \end{proof}
  
  The converse of Lemma~\ref{lem:propcolcograph} is not true in general.
  To see this, consider the properly-colored cograph $(G,\sigma)$ with vertex
  $V(G)=\{a,a',b,b',c,c'\}$, edges $ab,bc, a'b',a'c' $ and coloring
  $\sigma(a)=\sigma(a')=A$ $\sigma(b)=\sigma(b')=B$, $\sigma(c)=\sigma(c')=C$
  with $A,B,C$ being pairwise distinct. In this case, $\Sri(G,\sigma)$
  contains the triples $AC|B$ and $BC|A$.  By Lemma
  \ref{lem:Ru-SpeciesTriple}, the tree $S$ in every $\mu$-free scenario
  $\mfscen=(T,S,\sigma,\tT,\tS)$ or relaxed scenario
  $\scen=(T,S,\sigma,\mu,\tT,\tS)$ explaining $(G,\sigma)$ displays $AC|B$
  and $BC|A$. Since no such scenario can exist, $(G,\sigma)$ is not an LDT
  graph.
  
  \subsection{Recognition and Characterization of LDT Graphs}
  
  \begin{definition}
    Let $(G=(L,E),\sigma)$ be a graph with coloring $\sigma\colon L\to M$.
    Let $\mathcal{C}$ be a partition of $M$, and $\mathcal{C}'$ be the set of
    connected components of $G$.  We define the following binary relation
    $\rel(G, \sigma, \mathcal{C})$ by setting
    \begin{align*}
      (x,y)\in \rel(G, \sigma, \mathcal{C}) \iff  x,y\in L,\;
      \sigma(x), \sigma(y) & \in  C 
      \text{ for some } C\in\mathcal{C},
      \text{ and } \\
      x,y & \in C' \text{ for some } C'\in\mathcal{C}'.
    \end{align*}
    \label{def:rel}
  \end{definition}
  
  In words, two vertices $x,y\in L$ are in relation
  $\rel(G, \sigma, \mathcal{C})$ whenever they are in the same connected
  component of $G$ and their colors $\sigma(x), \sigma(y)$ are contained in
  the same set of the partition of $M$.
  
  \begin{lemma}
    Let $(G=(L,E),\sigma)$ be a graph with coloring $\sigma\colon L\to M$ and
    $\mathcal{C}$ be a partition of $M$.  Then,
    $\rel\coloneqq\rel(G, \sigma, \mathcal{C})$ is an equivalence relation
    and every equivalence class of $\rel$, or short $\rel$-class, is
    contained in some connected component of $G$. In particular, each
    connected component of $G$ is the disjoint union of $\rel$-classes.
    \label{lem:KinCC}
  \end{lemma}
  \begin{proof}
    It is easy to see that $\rel$ is reflexive and symmetric.  Moreover,
    $xy,yz\in \rel$ implies that $\sigma(x), \sigma(y), \sigma(z)$ must
    be contained in the same set of the partition $\mathcal{C}$, and
    $x,y,z$ must be contained in the same connected component of
    $G$. Therefore, $xy\in \rel$ and thus, $\rel$ is transitive.  In
    summary, $\rel$ is an equivalence relation.
    
    We continue with showing that every $\rel$-class $K$ is entirely
    contained in some connected component of $G$. Clearly, there is a
    connected component $C$ of $G$ such that $C\cap K\neq
    \emptyset$. Assume, for contradiction, that $K\not\subseteq C$. Hence,
    $G$ must be disconnected and, in particular, there is a second
    connected component $C'$ of $G$ such that $C'\cap K\neq
    \emptyset$. Hence, there is a pair $xy\in K$ such that $x\in C\cap K$
    and $y\in C'\cap K$. But then $x$ and $y$ are in different connected
    components of $G$ violating the definition of $\rel$; a contradiction.
    Hence, every $\rel$-class is entirely contained in some connected
    component of $G$. This and the fact the $\rel$-classes are disjoint
    implies that each connected component of $G$ is the disjoint union
    of $\rel$-classes.
  \end{proof}
  
  The following partition of the leaf sets of subtrees of a tree $S$ rooted
  at some vertex $u\in V(S)$ will be useful:
  \begin{align*}
    &\text{If } u 
    \textrm{ is not a leaf, then } &\partS(u)&\coloneqq
    \{L(S(v)) \mid v\in\child_S(u)\} \\
    & \textrm{and, otherwise, } &\partS(u)&\coloneqq \{\{u\}\}.
  \end{align*}
  One easily verifies that, in both cases, $\partS(u)$ yields a valid
  partition of the leaf set $L(S(u))$.  Recall that
  $\sigma_{|L',M'}\colon L'\to M'$ was defined as the ``submap'' of $\sigma$
  with $L'\subseteq L$ and $\sigma(L') \subseteq M' \subseteq M$.
  
  \begin{lemma}\label{lem:xy-iff-Ks-in-same-CC}
    Let $(G=(L,E),\sigma)$ be a properly colored cograph. Suppose that the
    triple set $\Sri(G,\sigma)$ is compatible and let $S$ be a tree on $M$
    that displays $\Sri(G,\sigma)$. Moreover, let $L'\subseteq L$ and
    $u\in V(S)$ such that $\sigma(L') \subseteq L(S(u))$. \
    Finally, set $\rel\coloneqq \rel(G[L'],\sigma_{|L',L(S(u))},\partS(u))$.\\
    Then, for all distinct $\rel$-classes $K$ and $K'$, either $xy\in E$ for
    all $x\in K$ and $y\in K'$, or $xy\notin E$ for all $x\in K$ and
    $y\in K'$.  In particular, for $x\in K$ and $y\in K'$, it holds that
    \begin{equation*}
      xy\in E \iff K, K' \text{ are contained in the same connected 
        component of } G[L'].
    \end{equation*}
  \end{lemma}
  \begin{proof}
    Let $\sigma \colon L\to M$ and put $\Sri = \Sri(G,\sigma)$.  Since $\Sri$
    is a compatible triple set on $M$, there is a tree $S$ on $M$ that
    displays $\Sri$.  Moreover, the condition
    $\sigma(L') \subseteq L(S(u))\subseteq M$ together with the fact that
    $\partS(u)$ is a partition of $L(S(u))$ ensures that $\rel$ is
    well-defined.
    
    Now suppose that $K$ and $K'$ are distinct \rel-classes.  As a
    consequence of Lemma~\ref{lem:KinCC}, we have exactly the two cases:
    either (i) $K$ and $K'$ are contained in the same connected component $C$
    of $G[L']$ or (ii) $K\subseteq C$ and $K'\subseteq C'$ for distinct
    components $C$ and $C'$ of $G[L']$.
    
    Case (i). Assume, for contradiction, that there are two vertices $x\in K$
    and $y\in K'$ with $xy\notin E$. Note that $C\subseteq L'$ and thus,
    $G[C]$ is an induced subgraph of $G[L']$. By Prop.~\ref{prop:cograph},
    both induced subgraphs $G[L']$ and $G[C]$ are cographs. Now we can again
    apply Prop.~\ref{prop:cograph} to conclude that
    $\mathrm{diam}(G[C])\leq 2$.  Hence, there is a vertex $z\in C$ such that
    $xz,zy\in E$.  Since $x$ and $y$ are in distinct classes of $\rel$ but in
    the same connected component $C$ of $G[L']$, $\sigma(x)$ and $\sigma(y)$
    must lie in distinct sets of $\partS(u)$.  In particular, it must hold that
    $\sigma(x)\neq \sigma(y)$.  The fact that $G[L']$ is properly colored
    together with $xz, yz \in E$ implies that
    $\sigma(z)\neq \sigma(x),\sigma(y)$.  By definition and since $G[L']$ is
    an induced subgraph of $G$, we obtain that
    $\sigma(x)\sigma(y)|\sigma(z)\in\Sri$.  In particular,
    $\sigma(x)\sigma(y)|\sigma(z)$ is displayed by $S$.  Since $\sigma(x)$
    and $\sigma(y)$ lie in distinct sets of $\partS(u)$, $u$ must be an inner
    vertex, and we have $\sigma(x)\in L(S(v))$ and $\sigma(y)\in L(S(v'))$
    for distinct $v, v'\in\child_S(u)$.  In particular, it must hold that
    $\lca_S(\sigma(x),\sigma(y))=u$.  Moreover, $z\in C\subseteq L'$ and
    $\sigma(L')\subseteq L(S(u))$ imply that $\sigma(z)\in L(S(u))$.  Taken
    together, the latter two arguments imply that $S$ cannot display the
    triple $\sigma(x)\sigma(y)|\sigma(z)$; a contradiction.
    
    Case~(ii). By assumption, the $\rel$-classes $K$ and $K'$ are in distinct
    connected components of $G[L']$, which immediately implies $xy\notin E$
    for all $x\in K$, $y\in K'$.
    
    In summary, either $xy\in E$ for all $x\in K$ and $y\in K'$, or
    $xy\notin E$ for all $x\in K$ and $y\in K'$.  Moreover, Case (i)
    establishes the \emph{if}-direction and Case (ii) establishes, by means
    of contraposition, the \emph{only-if}-direction of the final statement.
  \end{proof}
  
  Lemma~\ref{lem:xy-iff-Ks-in-same-CC} suggests a recursive strategy to
  construct a  relaxed  scenario $\scen=(T,S,\sigma,\mu,\tT,\tS)$ for a given
  properly-colored cograph $(G,\sigma)$, which is outlined in the main part
  of this paper and described more formally in
  Algorithm~\ref{alg:Ru-recognition}.  We proceed by proving the correctness
  of Algorithm~\ref{alg:Ru-recognition}.
  
  \begin{theorem}
    \label{thm:algo-works}
    Let $(G,\sigma)$ be a properly colored cograph, and assume that the
    triple set $\Sri(M,G)$ is compatible.  Then
    Algorithm~\ref{alg:Ru-recognition} returns a relaxed scenario
    $\scen=(T,S,\sigma,\mu,\tT,\tS)$  such that
    $\Gu(\scen)=G$ in polynomial time. 
  \end{theorem}
  
  \begin{proof}
    Let $\sigma\colon L\to M$ and put $\Sri\coloneqq\Sri(G,\sigma)$. By a
    slight abuse of notation, we will simply write $\mu$ and $\tT$ also for
    restrictions to subsets of $V(T)$. Observe first that due to Line
    \ref{line:if-false}, the algorithm continues only if $(G,\sigma)$ is a
    properly colored cograph and $\Sri$ is compatible, and returns a tuple
    $\scen=(T,S,\sigma,\mu,\tT,\tS)$ in this case.  In particular, a tree $S$
    on $M$ that displays $\Sri$ exists, and can e.g.\ be constructed using
    \texttt{BUILD} (Line \ref{line:S}).  By Lemma~\ref{lem:arbitrary-tT}, we
    can always construct a time map $\tS$ for $S$ satisfying $\tS(x)=0$ for
    all $x\in L(S)$ (Line~\ref{line:tS}).  By definition, $\tS(y)>\tS(x)$
    must hold for every edge $(y,x)\in E(S)$, and thus, we obtain
    $\epsilon>0$ in Line~\ref{line:epsilon}.  Moreover, the recursive
    function \texttt{BuildGeneTree} maintains the following invariant:
    \begin{Xclaim}
      \label{clm:color-subset}
      In every recursion step of the function \texttt{BuildGeneTree}, we have 
      $\sigma(L')\subseteq L(S(u_S))$.
    \end{Xclaim}
    \begin{claim-proof}
      Since $S$ (with root $\rho_S$) is a tree on $M$ by construction and
      thus $L(S(\rho_S))=M$, the statement holds for the top-level recursion
      step on $L$ and $\rho_S$.  Now assume that the statement holds for an
      arbitrary step on $L'$ and $u_S$.  If $u_S$ is a leaf, there are no
      deeper recursion steps.  Thus assume that $u_S$ is an inner vertex.
      Recall that $\partS(u_S)$ is a partition of $L(S(u_S))$ (by
      construction), and that $\rel= \rel(G[L'], \sigma_{|L',L(S(u))}, 
      \partS(u_S))$ is an equivalence relation (by Lemma~\ref{lem:KinCC}). This 
      together with the definition of $\rel$ and
      $\sigma(L')\subseteq L(S(u_S))$, implies that there is a child
      $v_S\in \child_S(u_{S})$ such that $\sigma(K)\subseteq L(S(v_S))$ for
      all $\rel$-classes $K$. In particular, therefore, the statement is true
      for all recursive calls on $K$ and $v_S$ in
      Line~\ref{line:recursive-call}.  Repeating this argument top-down along
      the recursion hierarchy proves the claim.
    \end{claim-proof}
    
    Note, that we are in the \emph{else}-condition in Line \ref{line:rel}
    only if $u_S$ is not a leaf. Therefore and as a consequence of
    Claim~\ref{clm:color-subset} and by similar arguments as in its proof,
    there is a vertex $v^*_S\in\child_S(u_S)$ such that
    $\sigma(C)\cap L(S(v^*_S))\ne\emptyset$ for every connected component $C$
    of $G[L']$ in Line~\ref{line:choose-v-S}, and a vertex
    $v_S\in \child_S(u_{S})$ such that $\sigma(K)\subseteq L(S(v_S))$ for
    every $\rel$-class $K$ in Line~\ref{line:choose-v-S-for-class}.
    Moreover, $\parent_S(u_{S})$ is always defined since we have $u_S=\rho_S$
    and thus $\parent_S(u_S)=0_S$ in the top-level recursion step, and
    recursively call the function \texttt{BuildGeneTree} on vertices $v_S$
    such that $v_S\prec_S u_S$.
    
    In summary, all assignments are well-defined in every recursion step.  It
    is easy to verify that the algorithm terminates since, in each recursion
    step, we either have that $u_S$ is a leaf, or we recurse on vertices
    $v_{S}$ that lie strictly below $u_S$.  We argue that the resulting tree
    $T'$ is a \emph{not necessarily phylogenetic} tree on $L$ by observing
    that, in each step, each $x\in L'$ is either attached to the tree as a
    leaf if $u_S$ is a leaf, or, since $\rel$ forms a partition of $L'$ by
    Lemma~\ref{lem:KinCC}, passed down to a recursion step on $K$ for some
    $\rel$-class $K$. Nevertheless, $T'$ is turned into a phylogenetic tree
    $T$ by suppression of degree-two vertices in Line~\ref{line:Tphylo}.
    Finally, $\mu(x)$ and $\tT(x)$ are assigned for all vertices
    $x\in L(T')=L$ in Line~\ref{line:mu-tT-leaves}, and for all newly created
    inner vertices in Lines~\ref{line:mu-tT-inner1}
    and~\ref{line:mu-tT-inner2}.
    
    Recall that $\tS$ is a valid time map satisfying $\tS(x)=0$ for all
    $x\in L(S)$ by construction.  Before we continue to show that $\scen$ is
    a relaxed scenario, we first show that the conditions for time maps and
    time consistency are satisfied for $(T',\tT, S, \tS,\mu)$:
    
    \begin{Xclaim}
      \label{clm:tT-mu-in-T-prime}
      For all $x,y \in V(T')$ with $x\prec_{T'} y$, we have $\tT(x)<\tT(y)$.
      Moreover, for all $x\in V(T')$, the following statements are true:
      \vspace{-0.02in}
      \begin{description}
        \item[(i)] if $\mu(x)\in V(S)$, then $\tT(x)=\tS(\mu(x))$, and
        \item[(ii)] if $\mu(x)=(a,b)\in E(S)$, then $\tS(b)<\tT(x)<\tS(a)$.
      \end{description}
    \end{Xclaim}
    \begin{claim-proof}
      Recall that we always write an edge $(u,v)$ of a tree $T$ such that
      $v\prec_T u$.  For the first part of the statement, it suffices to show
      that $\tT(x)<\tT(y)$ holds for every edge $(y,x)\in E(T')$, and thus to
      consider all vertices $x\neq \rho_{T'}$ in $T'$ and their unique
      parent, which will be denoted by $y$ in the following.  Likewise, we
      have to consider all vertices $x\in V(T')$ including the root to show
      the second statement.  The root $\rho_{T'}$ of $T'$ corresponds to the
      vertex $u_T$ created in Line~\ref{line:create-uT} in the top-level
      recursion step on $L$ and $\rho_{S}$.  Hence, we have
      $\mu(\rho_{T'})=(\parent_S(\rho_S)=0_S,\rho_S)\in E(S)$ and
      $\tT(\rho_{T'})=\tS(\rho_S) +\epsilon$ (cf.\ 
      Line~\ref{line:mu-tT-inner1}).
      Therefore, we have to show~(ii).  Since $\epsilon>0$, it holds that
      $\tS(\rho_S)<\tT(\rho_{T'})$.  Moreover,
      $\tS(0_S)-\tS(\rho_{S})\ge 3\epsilon$ holds by construction, and thus
      $\tS(0_S)-(\tT(\rho_{T'})-\epsilon)\ge 3\epsilon$ and
      $\tS(0_S)-\tT(\rho_{T'})\ge 2\epsilon$, which together with
      $\epsilon>0$ implies $\tT(\rho_{T'})<\tS(0_S)$.
      
      We now consider the remaining vertices
      $x\in V(T')\setminus\{\rho_{T'}\}$. Every such vertex $x$ is introduced
      into $T'$ in some recursion step on $L'$ and $u_S$ in one of the
      Lines~\ref{line:create-uT}, \ref{line:attach-leaf},
      \ref{line:create-vT} or~\ref{line:recursive-call}.  There are exactly
      the following three cases: (a) $x\in L(T')$ is a leaf attached to some
      inner vertex $u_T$ in Line~\ref{line:attach-leaf}, (b) $x=v_T$ as
      created in Line~\ref{line:create-vT}, and (c) $x=w_T$ as assigned in
      Line~\ref{line:recursive-call}.  Note that if $x=u_T$ as created in
      Line~\ref{line:create-uT}, then $u_T$ is either the root of $T'$, or
      equals a vertex $w_T$ as assigned in Line~\ref{line:recursive-call} in
      the ``parental'' recursion step.
      
      In Case~(a), we have that $x\in L(T')$ is a leaf and attached to some
      inner vertex $y=u_T$. Since $u_S$ must be a leaf in this case, and thus
      $\tS(u_S)=0$, we have $\tT(y)=0+\epsilon=\epsilon$ and $\tT(x)=0$ (cf.\
      Lines~\ref{line:mu-tT-inner1} and~\ref{line:mu-tT-leaves}).  Since
      $\epsilon>0$, this implies $\tT(x)<\tT(y)$.  Moreover, we have
      $\mu(x)=\sigma(x)\in L(S)\subset V(S)$ (cf.\
      Line~\ref{line:mu-tT-leaves}), and thus have to show Subcase~(i).
      Since $u_S$ is a leaf and $\sigma(L')\subseteq L(S(u_S))$, we conclude
      $\sigma(x)=u_S$.  Thus we obtain $\tT(x)=0=\tS(u_S)=\tS(\mu(x))$.
      
      In Case~(b), we have $x=v_T$ as created in Line~\ref{line:create-vT},
      and $x$ is attached as a child to some vertex $y=u_T$ created in the
      same recursion step.  Thus, we have $\tT(y)=\tS(u_S)+\epsilon$ and
      $\tT(x)=\tS(u_S)-\epsilon$ (cf.\ Lines~\ref{line:mu-tT-inner1}
      and~\ref{line:mu-tT-inner2}).  Therefore and since $\epsilon>0$, it
      holds $\tT(x)<\tT(y)$.  Moreover, we have $\mu(x)=(u_S,v^*_S)\in E(S)$
      for some $v^*_S\in\child_S(u_S)$.  Hence, we have to show Subcase~(ii).
      By a similar calculation as before, $\epsilon>0$,
      $\tS(u_S)-\tS(v^*_S)\ge 3\epsilon$ and $\tT(x)=\tS(u_S)-\epsilon$ imply
      $\tS(v^*_S)<\tT(x)<\tS(u_S)$.
      
      In Case~(c), $x=w_T$ as assigned in Line~\ref{line:recursive-call} is
      equal to $u_T$ as created in Line~\ref{line:create-uT} in some
      next-deeper recursion step with $u'_S\in\child_S(u_S)$. Thus, we have
      $\tT(x)=\tS(u'_S)+\epsilon$ and $\mu(x)=(u_S,u'_S)\in E(S)$ (cf.\
      Line~\ref{line:mu-tT-inner1}). Moreover, $x$ is attached as a child of
      some vertex $y=v_T$ as created in Line~\ref{line:create-vT}. Thus, we
      have $\tT(y)=\tS(u_S)-\epsilon$.  By construction and since
      $(u_S,u'_S)\in E(S)$, we have $\tS(u_S)-\tS(u'_S)\ge 3\epsilon$.
      Therefore, $(\tT(y)+\epsilon) - (\tT(x)-\epsilon) \ge 3\epsilon$ and
      thus $\tT(y)- \tT(x) \ge \epsilon$. This together with $\epsilon>0$
      implies $\tT(x)<\tT(y)$.  Moreover, since $\mu(x)=(u_S,u'_S)\in E(S)$
      for some $u'_S\in\child_S(u_S)$, we have to show Subcase~(ii).  By a
      similar calculation as before, $\epsilon>0$,
      $\tS(u_S)-\tS(u'_S)\ge 3\epsilon$ and $\tT(x)=\tS(u'_S)+\epsilon$ imply
      $\tS(u'_S)<\tT(x)<\tS(u_S)$.
    \end{claim-proof}
    
    \begin{Xclaim}
      \label{clm:relaxed-scenario}
      $\scen=(T,S,\sigma,\mu,\tT,\tS)$ is a relaxed scenario.
    \end{Xclaim}
    \begin{claim-proof}
      The tree $T$ is obtained from $T'$ by first adding a planted root $0_T$
      (and connecting it to the original root) and then suppressing all inner
      vertices except $0_T$ that have only a single child in Line 
      \ref{line:Tphylo}.  In particular,
      $T$ is a planted phylogenetic tree by construction.  The root
      constraint (G0) $\mu(x)=0_S$ if and only if $x=0_T$ also holds by
      construction (cf.\ Line~\ref{line:mu-tT-planted-root}).  Since we
      clearly have not contracted any outer edges $(y,x)$, i.e.\ with
      $x\in L(T')$, we conclude that $L(T')=L(T)=L$.  As argued before, we
      have $\tT(x)=0$ and $\mu(x)=\sigma(x)$ whenever $x\in L(T')=L(T)$ (cf.\
      Line~\ref{line:mu-tT-leaves}).  Since all other vertices are either
      $0_T$ or mapped by $\mu$ to some edge of $S$ (cf.\
      Lines~\ref{line:mu-tT-planted-root}, \ref{line:mu-tT-inner1}
      and~\ref{line:mu-tT-inner2}), the leaf constraint (G1)
      $\mu(x)=\sigma(x)$ is satisfied if and only if $x\in L(T)$.
      
      By construction, we have $V(T)\setminus \{0_T\} \subseteq V(T')$.
      Moreover, suppression of vertices clearly preserves the
      $\preceq$-relation between all vertices $x,y\in V(T)\setminus \{0_T\}$.
      Together with Claim~\ref{clm:tT-mu-in-T-prime}, this implies
      $\tT(x)<\tT(y)$ for all vertices $x,y\in V(T)\setminus \{0_T\}$ with
      $x\prec_{T} y$.  For the single child $\rho_T$ of $0_T$ in $T$, we have
      $\tT(\rho_T)\le \tS(\rho_S)+\epsilon$ where equality holds if the root
      of $T'$ was not suppressed and thus is equal to $\rho_T$.  Moreover,
      $\tT(0_T)=\tS(0_S)$ and $\tS(0_S)-\tS(\rho_S)\ge 3\epsilon$ hold by
      construction.  Taken together the latter two arguments imply that
      $\tT(\rho_T)<\tT(0_T)$.  In particular, we obtain $\tT(x)<\tT(y)$ for
      all vertices $x,y\in V(T)$ with $x\prec_{T} y$.  Hence, $\tT$ is a time
      map for $T$, which, moreover, satisfies $\tT(x)=0$ for all $x\in L(T)$.
      
      To show that $\scen=(T,S,\sigma,\mu, \tT,\tS)$ is a relaxed scenario, it 
      remains to show that
      $\mu$ is time-consistent with the time maps $\tT$ and $\tS$.  In case
      $x\in L(T)\subset V(T)$, we have $\mu(x)=\sigma(x)\in L(S)\subset V(S)$
      and thus $\tT(x)=0=\tS(\sigma(x))=\tS(\mu(x))$.  For $0_T$, we have
      $\tT(0_T)=\tS(0_S)=\tS(\mu(0_T))$.  The latter two arguments imply that
      all vertices $x\in L(T)\cup \{0_T\}$ satisfy (C1) in the
      Def.~\ref{def:tc-map}.  The remaining vertices of $T$ are all vertices
      of $T'$ as well.  In particular, they are all inner vertices that are
      mapped to some edge of $S$ (cf.\ Lines~\ref{line:mu-tT-inner1}
      and~\ref{line:mu-tT-inner2}).  The latter two arguments together with
      Claim~\ref{clm:tT-mu-in-T-prime} imply that, for all vertices
      $x\in V(T)\setminus (L(T)\cup \{0_T\})$, we have $\mu(x)=(a,b)\in E(S)$
      and $\tS(b)<\tT(x)<\tS(a)$.  Therefore, every such vertex satisfies
      (C2) in Def.~\ref{def:tc-map}.  It follows that the time consistency
      constraint (G2) is also satisfied, and thus $\scen$ is a relaxed
      scenario.
    \end{claim-proof}
    
    \begin{Xclaim}
      \label{clm:cotree}
      Every vertex $v\in V^0(T)$ was either created in
      Line~\ref{line:create-uT} or in Line~\ref{line:create-vT}. In
      particular, it holds for all $x,y\in L(T)$ with $\lca_T(x,y)=v$:
      \vspace{-0.02in}
      \begin{description}
        \item[(1)] If $v$ was created in Line~\ref{line:create-uT}, then
        $xy\notin E(G)$ and $xy\notin E(\Gu(\scen))$.
        \item[(2)] If $v$ was created in Line~\ref{line:create-vT}, then
        $xy\in E(G)$ and $xy\in E(\Gu(\scen))$.
      \end{description}
      Furthermore, $G$ is a cograph with cotree $(T,t)$ where $t(v) = 0$ if
      $v$ was created in Line~\ref{line:create-uT} and $t(v) = 1$, otherwise.
    \end{Xclaim}
    \begin{claim-proof}
      Since $T$ is phylogenetic, every vertex $v\in V^0(T)$ is the last
      common ancestor of two leaves $x,y\in L\coloneqq L(T)$. Let
      $v\in V^0(T)$ be arbitrary and choose arbitrary leaves $x,y\in L$ such
      that $\lca_T(x,y)=v$.  Since $v\in V^0(T)$, the leaves $x$ and $y$ must
      be distinct.
      
      Note that $v\notin L(T)\cup\{0_T\}$, and thus, $v$
      is also an inner vertex in $T'$.  Therefore, we have exactly the two
      cases (1) $v=u_T$ is created in Line~\ref{line:create-uT}, and (2)
      $v=v_T$ is created in Line~\ref{line:create-vT}. Similar as before, the
      case that $v=w_K$ is assigned in Line~\ref{line:recursive-call} is
      covered by Case~(a), since, in this case, $w_K$ is created in a deeper
      recursion step.
      
      We consider the recursion step on $L'$ and $u_S$, in which $v$ was
      created.  Clearly, it must hold that $x,y\in L'$.  Before we continue,
      set $\rel\coloneqq\rel(G[L'], \sigma_{|L',L(S(u))}, \partS(u_S))$ as in
      Line~\ref{line:rel}. Note, since $\scen$ is a relaxed scenario, the
      graph $(\Gu(\scen),\sigma)$ is well-defined.
      
      For Statement~(1), suppose that $v=u_T$ was created in
      Line~\ref{line:create-uT}.  Hence, we have the two cases~(i) the vertex
      $u_S$ of $S$ in this recursion step is a leaf, and (ii) $u_S$ is an
      inner vertex.  In Case~(i), we have $L(S(u_S))=\{u_S\}$.  Together with
      Claim~\ref{clm:color-subset} and $\sigma(x),\sigma(y)\in\sigma(L')$,
      this implies $\sigma(x)=\sigma(x)=u_S$.  By assumption, $(G,\sigma)$ is
      properly colored. By Prop.~\ref{prop:properCol} $(\Gu(\scen),\sigma)$
      must be properly colored as well.  Hence, we conclude that
      $xy\notin E(G)$ and $xy\notin E(\Gu(\scen))$, respectively.  In
      Case~(ii), $u_S$ is not a leaf.  Therefore, $\lca_{T}(x,y)=v=u_T$ is
      only possible if $x$ and $y$ lie in distinct connected components of
      $G[L']$.  This immediately implies $xy\notin E(G)$. Moreover, we have
      $\sigma(x),\sigma(y)\in L(S(u_S))$ and thus
      $\lca_S(\sigma(x),\sigma(y))\preceq_{S} u_S$. Since $\tS$ is a time map
      for $S$, it follows that
      $\tS(\lca_S(\sigma(x),\sigma(y)))\le \tS(u_S)$.  Together with
      $\tT(u_T)=\tS(u_S)+\epsilon$ (cf.\ Line~\ref{line:mu-tT-inner1}) and
      $\epsilon>0$, this implies
      $\tS(\lca_S(\sigma(x),\sigma(y))) < \tT(v)=\tT(\lca_T(x,y))$.  Hence,
      $xy\notin E(\Gu(\scen))$.
      
      For Statement~(2), suppose that $v=v_T$ was created in
      Line~\ref{line:create-vT}.  Therefore, $\lca_{T}(x,y)=v=v_T$ is only
      possible if $x$ and $y$ lie in the same connected components of $G[L']$
      but in distinct $\rel$-classes.  Now, we can apply
      Lemma~\ref{lem:xy-iff-Ks-in-same-CC} to conclude that $xy\in E(G)$.
      Moreover, the fact that $x$ and $y$ lie in the same connected component
      of $G[L']$ but in distinct $\rel$-classes implies that $\sigma(x)$ and
      $\sigma(y)$ lie in distinct sets of $\partS(u_S)$. Hence, there are
      distinct $v_S,v'_S\in\child_S(u)$ such that $\sigma(x)\preceq_{S}v_S$
      and $\sigma(y)\preceq_{S} v'_S$. In particular,
      $\lca_S(\sigma(x),\sigma(y))=u_S$.  In Line~\ref{line:mu-tT-inner2}, we
      assign $\tT(\lca_T(x,y))=\tT(v_T)=\tS(u_S)-\epsilon$.  Together with
      $\epsilon>0$, the latter two arguments imply
      $\tT(\lca_T(x,y))<\tS(u_S)=\tS(\lca_S(\sigma(x),\sigma(y)))$.
      Therefore, we have $xy\in E(\Gu(\scen))$.
      
      By the latter arguments, the cotree $(T,t)$ as defined above is
      well-defined and, for all $v\in V^0(T)$, we have $t(v)=1$ if and only
      if $xy\in E(G)$ for all $x,y\in L$ with $\lca_T(x,y)=v$.  Hence,
      $(T,t)$ is a cotree for $G$.
    \end{claim-proof}
    
    \begin{Xclaim}
      \label{clm:Ru-scen-equals-Ru}
      The relaxed scenario $\scen$ satisfies $\Gu(\scen)=G$.
    \end{Xclaim}
    \begin{claim-proof}
      Since $L(T)=L$, the two undirected graphs $\Gu(\scen)$ and $G$ have the
      same vertex set.  By Claim~\ref{clm:cotree}, we have, for all distinct
      $x,y\in L$, either $xy\notin E(G)$ and $xy\notin E(\Gu(\scen))$, or
      $xy\in E(G)$ and $xy\in E(\Gu(\scen))$.
    \end{claim-proof}
    
    Together, Claims~\ref{clm:relaxed-scenario}
    and~\ref{clm:Ru-scen-equals-Ru} imply that
    Algorithm~\ref{alg:Ru-recognition} returns a relaxed scenario
    $\scen=(T,S,\sigma,\mu,\tT,\tS)$ with coloring $\sigma$ such that
    $\Gu(\scen)=G$.
    
    To see that Algorithm \ref{alg:Ru-recognition} runs in polynomial time,
    we first note that the function $\mathtt{BuildGeneTree()}$ operates in
    polynomial time. This is clear for the setup and the $\mathbf{if}$ part.
    The construction of $\rel$ in the $\mathbf{else}$ part involves the
    computation of connected components and the evaluation of
    Def.~\ref{def:rel}, both of which can be achieved in polynomial time. This
    is also true for the comparisons of color classes required to identify
    $v_S^*$ and $v_S$. Since the sets $K$ in recursive calls of
    $\mathtt{BuildGeneTree()}$ form a partition of $L'$, and the $v_S$ are
    children of $u_S$ in $S$ and the depth of the recursion is bounded by
    $O(|L(S)|)$, the total effort remains polynomial.  
  \end{proof}
  
  \begin{theorem}
    A graph $(G,\sigma)$ is an LDT graph if and only if it is a properly
    colored cograph and $\Sri(G,\sigma)$ is compatible.
    \label{thm:characterization}
  \end{theorem}
  \begin{proof}
    By Lemma~\ref{lem:Ru-SpeciesTriple} and~\ref{lem:propcolcograph}, if
    $(G,\sigma)$ is an LDT graph then it is a properly colored cograph and
    $\Sri(G,\sigma)$ is compatible.  Now suppose that $(G,\sigma)$ is a
    properly colored cograph and $\Sri(G,\sigma)$ is compatible. Then, by
    Thm.~\ref{thm:algo-works}, Algorithm~\ref{alg:Ru-recognition} outputs
    a relaxed scenario $\scen=(T,S,\sigma,\mu,\tT,\tS)$ such that
    $\Gu(\scen)=G$. By definition, this in particular implies that
    $(G,\sigma)$ is an LDT graph.  
  \end{proof}
  
  \begin{corollary}
    LDT graphs can be recognized in polynomial time.
    \label{cor:LDTpoly}
  \end{corollary}
  \begin{proof}
    Cographs can be recognized in linear time \cite{Corneil:81a}, the proper
    coloring can be verified in linear time, the triple set $\Sri(G,\sigma)$
    contains not more than $|V(G)|\cdot|E(G)|$ triples and can be constructed
    in $O(|V(G)|\cdot|E(G)|)$ time, and compatibility of $\Sri(G,\sigma)$
    can be checked in
    $O(\min(|\Sri|\log^2 |V(G)|, |\Sri| + |V(G)|^2\ln |V(G)|))$ time
    \cite{Jansson:05}.
  \end{proof}
  
  \begin{corollary}
    The property of being an LDT graph is hereditary,
    that is, if $(G,\sigma)$ is an LDT graph then each of its vertex induced
    subgraphs is an LDT graph.
    \label{cor:LDT-here}
  \end{corollary}
  \begin{proof}
    Let $(G=(V,E),\sigma)$ be an LDT graph. It suffices to show that
    $(G-x, \sigma_{|V\setminus \{x\}})$ is an LDT graph, where $G-x$ is
    obtained from $G$ by removing $x\in V$ and all its incident edges.  By
    Prop.~\ref{prop:cograph}, $G-x$ is a cograph that clearly remains
    properly colored. Moreover, every induced path on three vertices in $G-x$
    is also an induced path on three vertices in $G$. This implies that if
    $xy|z \in \Sri' = \Sri(G-x,\sigma_{|V\setminus \{x\}})$, then
    $xy|z \in \Sri(G,\sigma)$. Hence, $\Sri' \subseteq \Sri(G,\sigma)$.  By
    Thm.~\ref{thm:characterization}, $\Sri(G,\sigma)$ is compatible.
    Hence, any tree that displays all triples in $\Sri(G,\sigma)$, in
    particular, displays all triples in $\Sri'$. Therefore, $\Sri'$ is
    compatible.  In summary, $(G-x, \sigma_{|V\setminus \{x\}})$ is a
    properly colored cograph and $\Sri'$ is compatible. By 
    Thm.~\ref{thm:characterization} it is an LDT graph.  
  \end{proof}
  
  The relaxed scenarios $\scen$ explaining an LDT graph $(G,\sigma)$ are far
  from being unique.  In fact, we can choose from a large set of trees
  $(S,\tS)$ that is determined only by the triple set $\Sri(G,\sigma)$:
  \begin{corollary}
    If $(G=(L,E),\sigma)$ is an LDT graph with coloring
    $\sigma\colon L\to M$, then for all planted trees $S$ on $M$ that display
    $\Sri(G,\sigma)$ there is a relaxed scenario
    $\scen=(T,S,\sigma,\mu,\tT,\tS)$ that contains $\sigma$ and $S$ and that
    explains $(G,\sigma)$.
    \label{cor:manyT}
  \end{corollary}
  \begin{proof}
    If $(G,\sigma)$ is an LDT graph, then the species tree $S$ assigned in
    Line~\ref{line:S} in Algorithm~\ref{alg:Ru-recognition} is an arbitrary
    tree on $M$ displaying $\Sri(G,\sigma)$.  
  \end{proof}
  
  \begin{corollary}\label{cor:displayed-cotree} 
    If $(G,\sigma)$ is an LDT graph, then there exists a relaxed scenario
    $\scen=(T,S,\sigma,\mu,\tT,\tS)$ explaining $(G,\sigma)$ such that $T$
    displays the discriminating cotree $T_{G}$ of $G$.
  \end{corollary} 
  \begin{proof}
    Suppose that $(G,\sigma)$ is an LDT graph.  By
    Thm.~\ref{thm:characterization}, $(G,\sigma)$ must be a properly colored
    cograph and $\Sri(G,\sigma)$ is comparable.  Hence,
    Thm.~\ref{thm:algo-works} implies that Algorithm~\ref{alg:Ru-recognition}
    constructs a relaxed scenario $\scen=(T,S,\sigma,\mu,\tT,\tS)$ explaining
    $(G,\sigma)$.  In particular, the tree $T$ together with labeling $t$ as
    specified in Claim~\ref{clm:cotree} is a cotree for $G$.  Since the
    unique discriminating cotree $(T_{G},\hat t)$ of $G$ is obtained from any
    other cotree by contraction of edges in $T$, the tree $T$ must display
    $T_{G}$.  
  \end{proof}
  
  Although, Cor.~\ref{cor:displayed-cotree} implies that there is always a
  relaxed  scenario $\scen$ where the tree $T$ displays the discriminating 
  cotree $T_{G}$ of $G=G(\scen)$, this is not true for all relaxed scenarios 
  $\scen$ with $G=G(\scen)$. Fig.~\ref{fig:Ru-cotree-not-displ} shows a relaxed
  scenario $\scen' = (T',S',\sigma,\mu',\tT',\tS')$ with $G = G(\scen')$ for
  which $T'$ does not display $T_G$.
  
  \begin{figure}[t]
    \centering
    \includegraphics[width=0.8\textwidth]{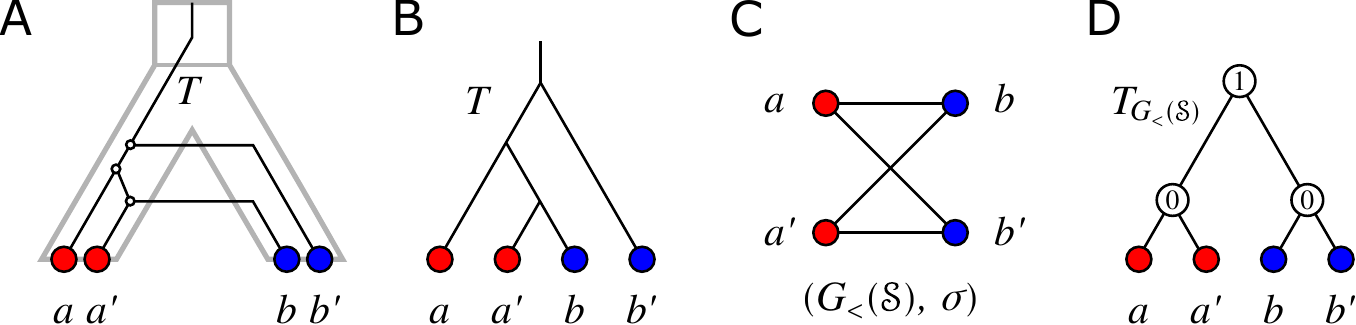}
    \caption{A relaxed scenario $\scen$ (A) with gene tree $T$ (B) and its
      associated graph $(\Gu(\scen),\sigma)$ (C). The discriminating cotree
      $T_{\Gu(\scen)}$ (D) is not displayed by $T$.}
    \label{fig:Ru-cotree-not-displ}
  \end{figure}
  
  Cor.~\ref{cor:displayed-cotree} enables us to relate connectedness of LDT
  graphs to properties of the relaxed scenarios by which it can be
  explained.
  \begin{lemma}
    \label{lem:Gu-connected}
    An LDT graph $(G=(L,E),\sigma)$ with $|L|>1$ is connected if and
    only if for every relaxed scenario
    $\scen=(T,S,\sigma,\mu,\tT,\tS)$ that explains $(G,\sigma)$, we have
    $\tT(\rho_T)<\tS(\lca_S(\sigma(L)))$.
  \end{lemma}
  \begin{proof}
    By contraposition, suppose first that there is a relaxed scenario
    $\scen=(T,S,\sigma,\mu,\tT,\tS)$ that explains $(G,\sigma)$ such that
    $\tT(\rho_T) \geq \tS(\lca_S(\sigma(L)))$. Since $|L(T)|=|L|>1$, the root
    $\rho_{T}$ is not a leaf. To show that $G$ is disconnected we consider
    two distinct children $v,w\in\child(\rho_T)$ of the root and leaves
    $x\in L(T(v))$ and $y\in L(T(w))$ and verify that $x$ and $y$ cannot be
    adjacent in $G$. If $\sigma(x)=\sigma(y)$, then $xy\notin E$ since
    $(G,\sigma)$ is properly colored (cf.\ Lemma~\ref{lem:propcolcograph}).
    Hence, suppose that $\sigma(x)\neq \sigma(y)$.  By construction,
    $\lca_T(x,y)=\rho_T$ and thus, by assumption,
    $\tT(\lca_T(x,y)) = \tT(\rho_T) \geq \tS(\lca_S(\sigma(L)))$.  Now
    $\lca_S(\sigma(L))\succeq_S \lca_S(\sigma(x),\sigma(y))$ implies that
    $\tS(\lca_S(\sigma(L)))\geq \tS(\lca_S(\sigma(x),\sigma(y)))$ and thus,
    $\tT(\lca_T(x,y))\geq \tS(\lca_S(\sigma(x),\sigma(y)))$.  Hence,
    $xy\notin E$. Consequently, for all distinct children
    $v,w\in\child(\rho_T)$, none of the vertices in $L(T(v))$ are adjacent to
    any of the vertices in $L(T(w))$ and thus, $G$ is disconnected.
    
    Conversely, suppose that $G$ is disconnected.  We consider
    Alg.~\ref{alg:Ru-recognition} with input $(G,\sigma)$.  By
    Thms.~\ref{thm:algo-works} and~\ref{thm:characterization}, the algorithm
    constructs a  relaxed  scenario $\scen=(T,S,\sigma,\mu,\tT,\tS)$ that
    explains $(G,\sigma)$.  Consider the top-level recursion step on $L$ and
    $\rho_S$.  Since $G$ is disconnected, the vertex $u_T$ created in
    Line~\ref{line:create-uT} of this step equals the root $\rho_T$ of the
    final tree $T$.  To see this, assume first that $\rho_S$ is a leaf.
    Then, we attach the $|L|>1$ elements in $L$ as leaves to $u_T$ (cf.\
    Line~\ref{line:attach-leaf}).  Now assume that $\rho_S$ is not a leaf.
    Since $G[L]=G$ has at least two components, we attach at least two
    vertices $v_T$ created in Line~\ref{line:create-vT} to $u_T$.  Hence
    $u_T$ is not suppressed in Line~\ref{line:Tphylo} and thus $\rho_T=u_T$.
    By construction, therefore, we have
    $\tT(\rho_T)=\tT(u_T)=\tS(u_S)+\epsilon=\tS(\rho_S)+\epsilon$ for some
    $\epsilon>0$. From $\sigma(\rho_S)\succeq_S \lca_S(\sigma(L))$ and the
    definition of time maps, we obtain
    $\tS(\rho_S)\ge\tS(\lca_S(\sigma(L)))$. Therefore, we have
    $\tT(\rho_T)\ge \tS(\lca_S(\sigma(L)))+\epsilon>\tS(\lca_S(\sigma(L)))$,
    which completes the proof. Therefore, we have shown so-far that if all
    relaxed scenarios $\scen=(T,S,\sigma,\mu,\tT,\tS)$ that explain
    $(G,\sigma)$ satisfy $\tT(\rho_T)\leq\tS(\lca_S(\sigma(L)))$, then
    $(G,\sigma)$ must be connected.  However,
    $\tT(\rho_T) = \tS(\lca_S(\sigma(L)))$ cannot occur, since we can reuse
    the same arguments as in the beginning of this proof to show that, in
    this case, $G$ is disconnected.  
  \end{proof}

  \subsection{Least Resolved Trees for LDT graphs}
  
  As we have seen e.g.\ in Cor.~\ref{cor:manyT}, there are in general many
  trees $S$ and $T$ forming  relaxed  scenarios $\scen$ that explain a
  given LDT graph $(G,\sigma)$.  This begs the question to what extent these
  trees are determined by ``representatives''. For $S$, we have seen that $S$
  always displays $\Sri(G,\sigma)$, suggesting to consider the role of
  $S=\Aho(\Sri(G,\sigma), M)$. This tree is least resolved in the sense
  that there is no  relaxed  scenario explaining the LDT graph
  $(G,\sigma)$ with a tree $S'$ that is obtained from $S$ by
  edge-contractions. The latter is due to the fact that any edge contraction
  in $\Aho(\Sri(G,\sigma), M)$ yields a tree $S'$ that does not display
  $\Sri(G,\sigma)$ any more \cite{Jansson:12}. By
  Prop.~\ref{lem:Ru-SpeciesTriple}, none of the  relaxed  scenarios
  containing $S'$ explain the LDT $(G,\sigma)$.
  
  \begin{definition}
    Let $\scen=(T,S,\sigma,\mu,\tT,\tS)$ be a relaxed scenario explaining the
    LDT graph $(G,\sigma)$. The planted tree $T$ is \emph{least resolved} for
    $(G,\sigma)$ if no relaxed scenario $(T',S',\sigma',\mu',\tT',\tS')$ with
    $T'<T$ explain $(G,\sigma)$.
    \label{def:LRT-LDT}
  \end{definition}
  In other words, $T$ is least resolved for $(G,\sigma)$ if no scenario with
  a gene tree $T'$ obtained from $T$ by a series of edge contractions
  explains $(G,\sigma)$. The examples in Fig.~\ref{fig:LRT-not-unique} show
  that there is not always a unique least resolved tree.
  
  As outlined in the main part of this paper, the examples in Fig.\
  \ref{fig:LRT-not-unique} show that LDT graphs are in general not
  accompanied by unique least resolved trees and the example in
  Fig.~\ref{fig:cotree-not-resolved-enough} shows that the unique
  discriminating cotree $T_G$ of an LDT graph $(G,\sigma)$ is not always
  ``sufficiently resolved''.
  
  \section{Horizontal Gene Transfer and Fitch Graphs}
  \label{TP:sect:HGT}
  
  \subsection{HGT-Labeled Trees and rs-Fitch Graphs}
  
  As alluded to in the introduction, the LDT graphs are intimately related
  with horizontal gene transfer. To formalize this connection we first define
  transfer edges. These will then be used to encode Walter Fitch's concept of
  xenologous gene pairs \cite{Fitch:00,Darby:17} as a binary relation, and
  thus, the edge set of a graph.
  \begin{definition}
    Let $\scen = (T,S,\sigma,\mu,\tT,\tS)$ be a relaxed scenario.  An edge
    $(u,v)$ in $T$ is a \emph{transfer edge} if $\mu(u)$ and $\mu(v)$ are
    incomparable in $S$. The \emph{HGT-labeling} of $T$ in $\scen$ is the
    edge labeling $\lambda_{\scen}: E(T)\to\{0,1\}$ with $\lambda(e)=1$ if and 
    only if $e$ is a transfer edge.
    \label{def:HGT-label}
  \end{definition}
  The vertex $u$ in $T$ thus corresponds to an HGT event, with $v$ denoting
  the subsequent event, which now takes place in the ``recipient'' branch of
  the species tree. Note that $\lambda_{\scen}$ is completely determined by
  $\scen$.  In general, for a given a gene tree $T$, HGT events correspond to
  a labeling or coloring of the edges of $T$.
  
  \begin{definition}[Fitch graph] 
    Let $(T,\lambda)$ be a tree $T$ together with a map
    $\lambda\colon E(T)\to \{0,1\}$.  The \emph{Fitch graph}
    $\gfitch(T,\lambda) = (V,E)$ has vertex set $V\coloneqq L(T)$
    and edge set
    \begin{equation*}
      E \coloneqq \{xy \mid  x,y\in L,
      \text{ the unique path connecting  }
      x \text{ and } y \text{ in } T 
      \text{ contains an edge }
      e \text{ with } \lambda(e)=1. \}
    \end{equation*}
    \label{def:FitchG}
  \end{definition}
  By definition, Fitch graphs of 0/1-edge-labeled trees are loop-less and
  undirected. We call edges $e$ of $(T,\lambda)$ with label $\lambda(e)=1$
  also 1-edges and, otherwise, 0-edges.
  \begin{remark} Fitch graphs as defined here have been termed
    \emph{undirected} Fitch graphs \cite{Hellmuth:18a}, in contrast to the
    notion of the \emph{directed} Fitch graphs of 0/1-edge-labeled trees
    studied e.g.\ in \cite{Geiss:18a,Hellmuth:2019a}.
  \end{remark}
  
  \begin{proposition}{\cite{Hellmuth:18a,Zverovich:99}}
    The following statements are equivalent.
    \begin{enumerate}
      \item $G$ is the Fitch graph of a 0/1-edge-labeled tree.
      \item $G$ is a complete multipartite graph.
      \item $G$ does not contain $K_2+K_1$ as an induced subgraph.
    \end{enumerate}
    \label{prop:fitch}
  \end{proposition}
  
  A natural connection between LDT graphs and complete multipartite graphs is
  suggested by the definition of triple sets $\Tri(G)$, since each forbidden 
  induced
  subgraph $K_2+K_1$ of a complete multipartite graphs corresponds to a
  triple in an LDT graph. More precisely, we have:
  \begin{lemma} 
    $(G,\sigma)$ is a properly colored complete multipartite if and only if it
    is properly colored and $\Tri(G) = \emptyset$.
    \label{lem:mulip-triples}
  \end{lemma}
  \begin{proof}
    The equivalence between the statements can be seen by observing that $G$
    is a complete multipartite graph if and only if $G$ does not contain an
    induced $K_2+K_1$ (cf.\ Prop.~\ref{prop:fitch}).  By definition of
    $\Tri(G)$, this is the case if and only if $\Tri(G)=\emptyset$.  
  \end{proof}
  
  \begin{definition}[rs-Fitch graph]
    Let $\scen = (T,S,\sigma,\mu,\tT,\tS)$ be a relaxed scenario with
    HGT-labeling $\lambda_{\scen}$. We call the vertex colored graph
    $(\gfitch(\scen),\sigma) \coloneqq (\gfitch(T,\lambda_{\scen}),\sigma)$
    the \emph{Fitch graph of the scenario $\scen$.}\\
    A vertex colored graph $(G,\sigma)$ is a \emph{relaxed scenario Fitch
      graph} (\emph{rs-Fitch graph}) if there is a relaxed scenario
    $\scen = (T,S,\sigma,\mu,\tT,\tS)$  such that
    $G = \gfitch(\scen)$.
    \label{def:rsFitchG}
  \end{definition}
  
  Fig.~\ref{fig:fitch-example} shows that rs-Fitch graphs are not necessarily
  properly colored.  A subtle difficulty arises from the fact that Fitch
  graphs of 0/1-edge-labeled trees are defined without a reference to the
  vertex coloring $\sigma$, while the rs-Fitch graph is vertex colored.
  \begin{fact}
    If $(G,\sigma)$ is an rs-Fitch graph then $G$ is a complete multipartite
    graph.
    \label{obs:Fitch}
  \end{fact}
  The ``converse'' of Obs.~\ref{obs:Fitch} is not true in general, as we
  shall see in Thm.~\ref{thm:char-rsFitch} below. If, however, the coloring
  $\sigma$ can be chosen arbitrarily, then every complete multipartite graph $G$
  can be turned into an rs-Fitch graph $(G,\sigma)$ as shown in 
  Prop.~\ref{prop:converse-obs-fitch}.
  
  \begin{proposition}
    If $G$ is a complete multipartite graph, then there  exists a relaxed
    scenario $\scen=(T,S,\sigma,\mu,\tT,\tS)$ such that $(G,\sigma)$ is an
    rs-Fitch graph.
    \label{prop:converse-obs-fitch}
  \end{proposition}
  \begin{proof}
    Let $G$ be a complete multipartite graph and set $L\coloneqq V(G)$ and
    $R\coloneqq E(G)$. If $R=\emptyset$, then the relaxed scenario $\scen$
    constructed in the proof of Lemma~\ref{lem:Rempty} shows that
    $E(G)=E(\gfitch(\scen)) = \emptyset$.  Hence, we assume that
    $R\neq \emptyset$ and explicitly construct a relaxed scenario
    $\scen = (T,S,\sigma,\mu,\tT,\tS)$ such that $(G,\sigma)$ is an rs-Fitch
    graph.
    
    We start by specifying the coloring $\sigma\colon L\to M$.  Since $G$ is
    a complete multipartite graph it is determined by its independent sets
    $I_1,\dots,I_k$, which form a partition of $L$.  We set
    $M\coloneqq\{1,2,\ldots,k\}$ and color every $x\in I_j$ with color
    $\sigma(x)=j$, $1\leq j\leq k$. By construction, $(G,\sigma)$ is properly
    colored, and $\sigma(x)=\sigma(y)$ whenever $xy\notin R$, i.e., whenever
    $x$ and $y$ lie in the same independent set. Therefore, we have
    $\Sri(G,\sigma) = \emptyset$. Let $S$ be the planted star tree with leaf
    set $L(S)=\{1,\dots,k\} = M$ and $\child_S(\rho_S)=M$. Since
    $R\neq \emptyset$, we have $k\geq 2$, and thus, $\rho_S$ has at least two
    children and is, therefore, phylogenetic.  We choose the time map $\tS$
    by putting $\tS(0_S)=2$, $\tS(\rho_S)=1$ and $\tS(x)=0$ for all
    $x\in L(S)$.
    
    Finally, we construct the planted phylogenetic tree $T$ with planted root
    $0_T$ and root $\rho_T$ as follows: Vertex $\rho_T$ has $k$ children
    $u_1,\dots, u_k$. If $I_j=\{x_j\}$ consists of a single element, then we
    put $u_j\coloneqq x_j$ as a leaf or $T$, and otherwise, vertex $u_j$ has
    exactly $|I_j|$ children where $\child(u_j)=I_j$. Now label, for all
    $i\in \{2,\dots, k\}$, the edge $(\rho_T,u_i)$ with ``$1$'', and all
    other edges with ``$0$''. Since $k\ge 2$, the tree $T$ is also
    phylogenetic by construction.
    
    \begin{figure}[ht]
      \centering
      \includegraphics[width=0.85\textwidth]{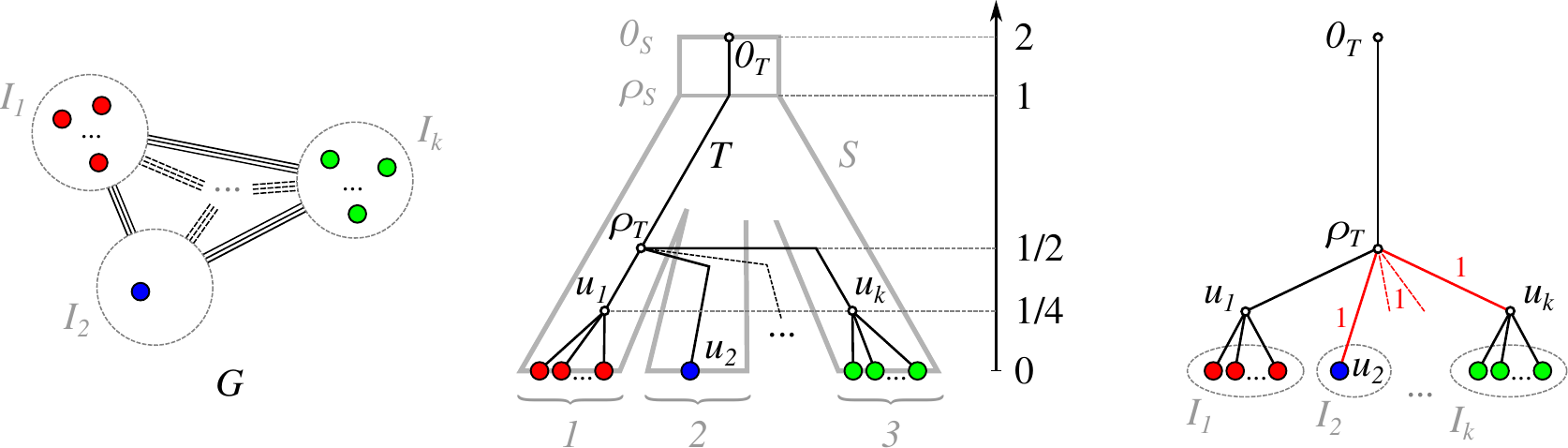}
      \caption{Construction in the proof of
        Prop.~\ref{prop:converse-obs-fitch}.}
      \label{fig:prop:converse-obs-fitch}
    \end{figure}
    
    We specify the time map $\tT$ and the reconciliation map $\mu$ by
    defining, for every $v\in V(T)$,
    \begin{equation*}
      \tT(v) \coloneqq
      \begin{cases}
        2=\tS(0_S) \\
        0 \\
        1/2 \\
        1/4 
      \end{cases} 
      \mu(v) \coloneqq
      \begin{cases}
        0_S &\text{if } v=0_T,\\
        \sigma(v) &\text{if } v\in L(T),\\
        (\rho_S,1) &\text{if } v = \rho_T, \textrm{ and}\\
        (\rho_S,i) &\text{if } v=u_i\not\in L(T), 1\leq i\leq k.
      \end{cases}
    \end{equation*}
    With the help of Fig.~\ref{fig:prop:converse-obs-fitch}, it is now easy
    to verify that (i) $\tT$ is a time map for $T$, (ii) the reconciliation
    map $\mu$ is time-consistent, and (iii) $\lambda_{\scen} = \lambda$.  In
    summary, $\scen = (T,S,\sigma,\mu,\tT,\tS)$ is a relaxed scenario, and
    $(G,\sigma) = (\gfitch(\scen),\sigma)$ is an rs-Fitch graph.  
  \end{proof}
  
  Although every complete multipartite graph can be colored in such a way
  that it becomes an rs-Fitch graph (cf.\
  Prop.~\ref{prop:converse-obs-fitch}), there are colored, complete
  multipartite graphs $(G,\sigma)$ that are not rs-Fitch graphs, i.e., that
  do not derive from a relaxed scenario (cf.\ Thm.~\ref{thm:char-rsFitch}).
  We summarize this discussion in the following
  \begin{fact}
    There are (planted) 0/1-edge labeled trees $(T,\lambda)$ and colorings
    $\sigma\colon L(T)\to M$ such that there is no relaxed scenario
    $\scen = (T,S,\sigma,\mu,\tT,\tS)$ with $\lambda=\lambda_{\scen}$.
    \label{obs:01T-notScen}
  \end{fact}
  A subtle -- but important -- observation is that trees $(T,\lambda)$ with
  coloring $\sigma$ for which Obs.~\ref{obs:01T-notScen} applies may still
  encode an rs-Fitch graph $(\gfitch(T,\lambda),\sigma)$, see Example
  \ref{ex:lst} and Fig.~\ref{fig:TreeClassesDistinct}.  The latter is due to
  the fact that $\gfitch(T,\lambda) = \gfitch(T',\lambda')$ may be possible for
  a different tree $(T',\lambda')$ for which there is a relaxed scenario
  $\scen' = (T',S,\sigma,\mu,\tT,\tS)$ with $\lambda' = \lambda_{\scen}$.  In
  this case, $(\gfitch(T,\lambda),\sigma) = (\gfitch(\scen'),\sigma)$ is an
  rs-Fitch graph. We shall briefly return to these issues in the discussion
  section~\ref{sect:concl}.
  
  \subsection{LDT Graphs and rs-Fitch Graphs}
  
  We proceed to investigate to what extent an LDT graph provides information
  about an rs-Fitch graph.  As we shall see in
  Thm.~\ref{thm:FitchRu-scenario} there is indeed a close connection between
  rs-Fitch graphs and LDT graphs. We start with a useful relation between the
  edges of rs-Fitch graphs and the reconciliation maps $\mu$ of their
  scenarios.
  
  \begin{lemma}
    Let $\gfitch(\scen)$ be an rs-Fitch graph for some  relaxed  
    scenario $\scen$. Then,
    $ab\notin E(\gfitch(\scen))$ implies that
    $\lca_S(\sigma(a),\sigma(b)) \preceq_S \mu(\lca_T(a,b)) $.
    \label{lem:independent-lca}
  \end{lemma}
  \begin{proof}
    Assume first that $ab\notin E(\gfitch(\scen))$ and denote by $P_{xy}$ the
    unique path in $T$ that connects the two vertices $x$ and $y$. Clearly,
    $u\coloneqq \lca_T(a,b)$ is contained in $P_{ab}$, and this path $P_{ab}$
    can be subdivided into the two paths $P_{u,a}$ and $P_{u,b}$ that have
    only vertex $u$ in common.  Since $ab\notin E(\gfitch(\scen))$, none of
    the edges $(v,w)$ along the path $P_{ab}$ in $T$ is a transfer edge, and
    thus, the images $\mu(v)$ and $\mu(w)$ are comparable in $S$. This
    implies that the images of any two vertices along the path $P_{u,a}$ as
    well as the images of any two vertices along $P_{u,b}$ are
    comparable.  In particular, therefore, $\mu(u)$ is comparable with both
    $\mu(a)=\sigma(a)\eqqcolon A$ and $\mu(b)=\sigma(b)\eqqcolon B$, where we
    may have $A=B$.  Together with the fact that $A$ and $B$ are leaves in $S$,
    this implies that $\mu(u)$ is an ancestor of $A$ and $B$. Since
    $\lca_S(A,B)$ is the ``last'' vertex that is an ancestor of both $A$ and
    $B$, we have $\lca_S(A,B) \preceq_S \mu(u)$.  
  \end{proof}
  
  The next result shows that a subset of transfer edges can be inferred
  immediately from LDT graphs:
  \begin{theorem}
    If $(G,\sigma)$ is an LDT graph, then $G\subseteq \gfitch(\scen)$ for all
    relaxed scenarios $\scen$ that explain $(G,\sigma)$. 
    \label{thm:infer-fitch}
  \end{theorem}
  \begin{proof}
    Let $\scen = (T,S,\sigma,\mu,\tT,\tS)$ be a relaxed scenario that
    explains $(G,\sigma)$, i.e., $G = \Gu(\scen)$.  By definition,
    $V(G) = V(\gfitch(\scen)) = L(T)$. Hence it remains to show that
    $E(G) \subseteq E(\gfitch(\scen))$. To this end, consider $ab \in E(G)$
    and assume, for contradiction, that $ab\notin E(\gfitch(\scen))$.  Let
    $A \coloneqq \sigma(a)$ and $B\coloneqq \sigma(b)$.  By Lemma
    \ref{lem:independent-lca}, $\lca_S(A,B) \preceq_S \mu(\lca_T(a,b))$.  But
    then, by Def.\ \ref{def:time-map} and \ref{def:tc-map},
    $\tS(\lca_S(A,B)) \leq \tS(\lca_T(a,b))$, implying $ab\notin E(G)$, a
    contradiction.  
  \end{proof}
  
  Since we only have that $xy$ is an edge in $\gfitch(\scen)$ if the path
  connecting $x$ and $y$ in the tree $T$ of $\scen$ contains a transfer edge,
  Thm.~\ref{thm:infer-fitch} immediately implies
  \begin{corollary}
    For every relaxed scenario $\scen=(T,S,\sigma,\mu,\tT,\tS)$ without
    transfer edges, it holds that $E(\Gu(\scen)) = \emptyset$.
    \label{cor:noHGT}
  \end{corollary}
  
  Thm.~\ref{thm:infer-fitch} provides the formal justification for indirect
  phylogenetic approaches to HGT inference that are based on the work of
  \citet{Lawrence:92}, \citet{Clarke:02}, and \citet{Novichkov:04} by showing
  that $xy\in E(\Gu(\scen))$ can be explained only by HGT, irrespective of
  how complex the true biological scenario might have been. However, it does
  not cover all HGT events. Fig.~\ref{fig:Fitch-not-RU} shows that there are
  relaxed scenarios $\scen$ for which $\Gu(\scen) \neq \gfitch(\scen)$ even
  though $\gfitch(\scen)$ is properly colored. Moreover, it is possible that
  an rs-Fitch graph $(G,\sigma)$ contains edges $xy\in E(G)$ with
  $\sigma(x)=\sigma(y)$. In particular, therefore, an rs-Fitch graph is not
  always an LDT graph.
  
  It is natural, therefore, to ask whether for every properly colored Fitch
  graph there is a relaxed scenario $\scen$ such that
  $\Gu(\scen) = \gfitch(\scen)$. An affirmative answer is provided by
  \begin{theorem}
    The following statements are equivalent.
    \begin{enumerate}
      \item $(G,\sigma)$ is a properly colored complete multipartite graph.
      \item There is a relaxed scenario $\scen=(T,S,\sigma,\mu,\tT,\tS)$ with
      coloring $\sigma$ such that $G=\Gu(\scen) = \gfitch(\scen)$.
      \item $(G,\sigma)$ is complete multipartite and an LDT graph. 
      \item $(G,\sigma)$ is properly colored and an rs-Fitch graph.
    \end{enumerate}
    In particular, for every properly colored complete multipartite graph
    $(G,\sigma)$ the triple set $\Sri(G,\sigma)$ is compatible.
    \label{thm:FitchRu-scenario}
  \end{theorem}
  \begin{proof}
    \par\noindent\emph{(1) implies (2)}.  We assume that $(G,\sigma)$ is a
    properly colored multipartite graph and set $L\coloneqq V(G)$ and
    $E\coloneqq E(G)$.  If $E=\emptyset$, then the relaxed scenario $\scen$
    constructed in the proof of Lemma~\ref{lem:Rempty} satisfies
    $G=\Gu(\scen) = \gfitch(\scen)$, i.e., the graphs are edgeless.  Hence,
    we assume that $E\neq \emptyset$ and explicitly construct a relaxed
    scenario $\scen = (T,S,\sigma,\mu,\tT,\tS)$ with coloring $\sigma$ such
    that $G=\Gu(\scen) = \gfitch(\scen)$.
    
    The graph $(G,\sigma)$ is properly colored and complete multipartite by
    assumption.  Let $I_1,\dots, I_k$ denote the independent sets of $G$.
    Since $E\neq \emptyset$, we have $k>1$.  Since all $x\in I_i$ are
    adjacent to all $y\in I_j$, $i\neq j$ and $(G,\sigma)$ is properly
    colored, it must hold that $\sigma(I_i)\cap \sigma(I_j)=\emptyset$.  For
    a fixed $i$ let $v_i^1,\dots v_i^{|I_i|}$ denote the elements in $I_i$.
    
    We first start with the construction of the species tree $S$.  First we
    add a planted root $0_S$ with child $\rho_S$.  Vertex $\rho_S$ has
    children $w_1,\dots, w_k$ where each $w_j$ corresponds to one $I_j$.
    Note, $\sigma\colon L\to M$ may not be surjective, in which case we would
    add one additional child $x$ to $\rho_S$ for each color
    $x\in M\setminus \sigma(L)$.
    
    If $|\sigma(I_j)| = 1$, then we identify the single color $x\in 
    \sigma(I_j)$ with
    $w_j$.  Otherwise, i.e., if $|\sigma(I_j)| > 1$, vertex $w_j$ 
    has as
    children the set $\child_S(w_j)=\sigma(I_j)$ which are leaves in $S$.
    See Fig.~\ref{fig:Fitch-RU} for an illustrative example.  Now we can
    choose the time map $\tS$ for $S$ such $\tS(0_S)=3$, 
    $\tS(\rho_S)=2$,
    $\tS(x)=0$ for all $x\in L(S)$ and $\tS(x)=1$ for all
    $x\in V^0(S)\setminus\{\rho_S\}$.
    
    We now construct $T$ as follows. The tree $T$ has planted root $0_T$ with
    child $\rho_T$. Vertex $\rho_T$ has $k$ children $u_1,\dots, u_k$ where
    each $u_j$ corresponds to one $I_j$.  Vertex $u_j$ is a leaf if $|I_j|=1$,
    and, otherwise, has exactly $|I_j|$ children that are 
    uniquely identified with 
    the elements in $I_j$.
    
    \begin{figure}[ht]
      \begin{center}
        \includegraphics[width=0.85\textwidth]{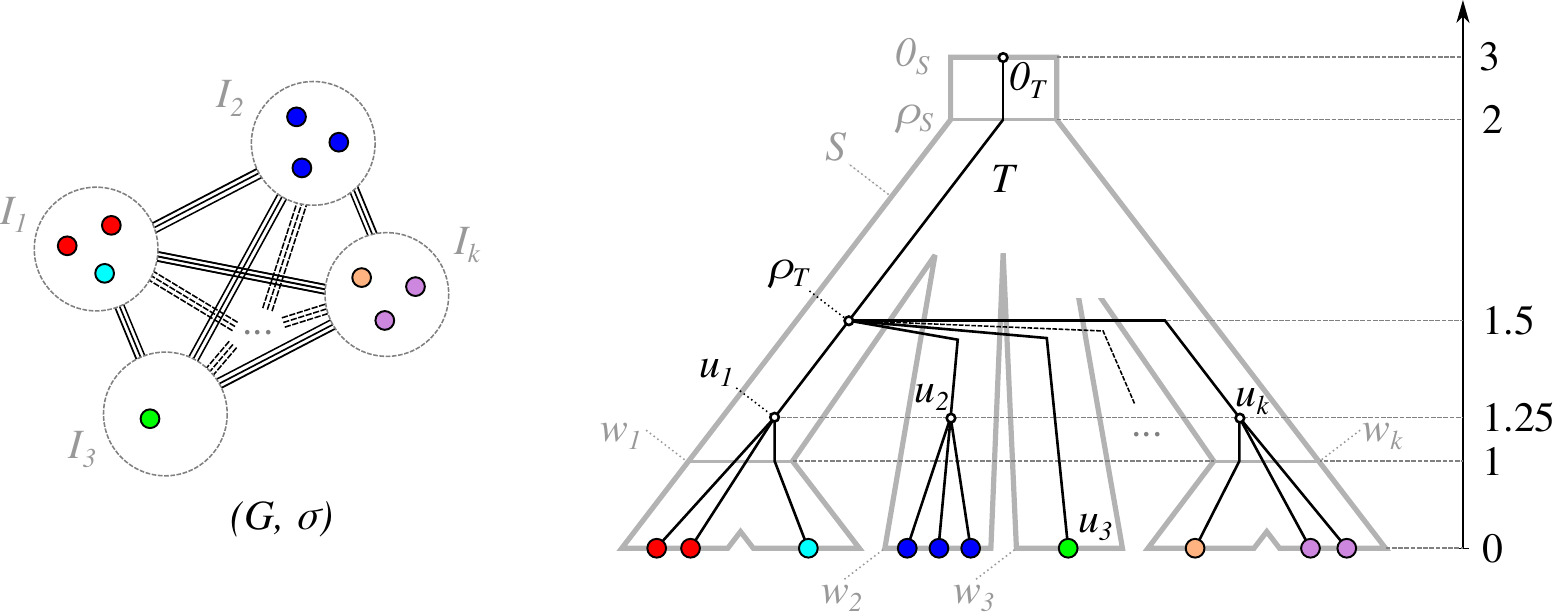}
      \end{center}
      \caption{Construction of the  relaxed  scenario $\scen$ in the proof of 
        Thm.~\ref{thm:FitchRu-scenario}. }
      \label{fig:Fitch-RU}
    \end{figure}
    
    We now define the time map $\tT$ and reconciliation map $\mu$ for
    $v\in V(T)$:
    \begin{equation*}
      \tT(v) \coloneqq
      \begin{cases}
        3=\tS(0_S) \\
        0 \\
        1.5 \\
        1.25
      \end{cases}
      \mu(v) \coloneqq
      \begin{cases}
        0_S &\text{if } v=0_T,\\
        \sigma(v) &\text{if } v\in L(T),\\
        (\rho_S,w_1) &\text{if } v = \rho_T, \textrm{ and}\\
        (\rho_S,w_i) &\text{if } v=u_i\not\in L(T), 1\leq i\leq k.
      \end{cases}
    \end{equation*}
    With the help of Fig.~\ref{fig:Fitch-RU} it is now easy to verify that
    (i) $\tT$ is a time map for $T$, and that (ii) the reconciliation map
    $\mu$ is time-consistent. In summary the constructed
    $\scen=(T,S,\sigma,\mu,\tT,\tS)$ is a relaxed scenario.
    
    We continue with showing that $E=E(\Gu(\scen))=E(\gfitch(\scen))$.  To
    this end, let $a,b\in L$ be two vertices. Note, $ab\in E$ if and only if
    $a\in I_i$ and $b\in I_j$ for distinct
    $i,j\in [k]\coloneqq \{1,2,\ldots,k\}$.
    
    First assume that $ab\in E$ and thus, $a\in I_i$ and $b\in I_j$ for
    distinct $i,j\in [k]$.  By construction,
    $a\preceq_{T}u_i\ne u_j\succeq_{T} b$ with
    $\lca_{T}(u_i,u_j)=\rho_{T}$. In particular, we have
    $\parent_T(u_i)=\parent_T(u_j)=\rho_{T}$ and the path from $a$ to $b$
    contains the two edges $(\rho_{T},u_i)$ and $(\rho_{T},u_j)$.  By
    construction, we have $\mu(\rho_T)=(\rho_{S},w_1)$, and for all
    $1\leq l\leq k$, $\mu(u_l)=\sigma(u_l)=w_l$ if $u_l$ is a leaf, and
    $\mu(u_l)=(\rho_S,w_l)$ otherwise.  These two arguments imply that
    $\mu(\rho_T)$ and $\mu(u_l)$ are comparable if and only if $u_l=u_1$.
    Now, since $u_i\ne u_j$, they cannot both be equal to $u_1$ and thus, at
    least one of the edges $(\rho_{T},u_i)$ and $(\rho_{T},u_j)$ is a
    transfer edge. Hence, $ab\in E(\gfitch(\scen))$.  By construction,
    $ab\in E$ implies $\lca_T(a,b)=\rho_T$.  Hence, we have
    $\mu(\lca_T(a,b)) = \mu(\rho_T)=(\rho_S,w_1)\prec_S\rho_S =
    \lca_S(\sigma(a),\sigma(b))$, and thus $ab\in E(\Gu(\scen))$.
    
    Now assume that $ab\notin E$, and thus, $a,b\in I_i$ for some $i\in[k]$.
    It clearly suffices to consider the case $a\ne b$, and thus,
    $a,b\in\child_T(u_i)$ and $u_i\notin L(T)$ holds by construction.  In
    particular, the path between $a$ and $b$ only consists of the edges
    $(u_i,a)$ and $(u_i,b)$.  Moreover, we have
    $\sigma(a),\sigma(b)\preceq_{S} w_i$ and $\mu(u_i)=(\rho_S,w_i)$.  Hence,
    none of the edges $(u_i,a)$ and $(u_i,b)$ is a transfer edge, and
    $ab\notin E(\gfitch(\scen))$.  We have
    $\mu(\lca_{T}(a,b))=(\rho_S,w_i)\succ_T w_i \succeq_{T}
    \lca_{S}(\sigma(a),\sigma(b))$, and thus
    $\tT(\lca_{T}(a,b))> \tS(\lca_{S}(\sigma(a),\sigma(b)))$.  Hence,
    $ab\notin E(\Gu(\scen))$.
    
    In summary, $ab\in E$ if and only if $ab\in E(\gfitch(\scen))$ if and
    only if $ab\in E(\Gu(\scen))$, and consequently,
    $G=\Gu(\scen) = \gfitch(\scen)$.
    
    \smallskip
    \par\noindent\emph{(2) implies (1).} Thus, suppose that there is a relaxed
    scenario $\scen=(T,S,\sigma,\mu,\tT,\tS)$ with coloring $\sigma$ such
    that $G=\Gu(\scen) = \gfitch(\scen)$. Prop.~\ref{prop:properCol} implies
    that $(G,\sigma)=(\Gu(\scen),\sigma)$ is properly colored. Moreover,
    $(G,\sigma)=(\gfitch(\scen),\sigma)$ is an rs-Fitch graph and thus, by
    Obs.~\ref{obs:Fitch}, $G$ is complete multipartite.
    
    Statements (1) and (2) together with Prop.~\ref{prop:fitch} imply
    (3). Conversely, if (3) is satisfied then Prop.~\ref{prop:properCol}
    implies that $(G,\sigma)$ is properly colored.  This and the fact that
    $G$ is complete multipartite implies (1).  Therefore, Statements (1), (2)
    and (3) are equivalent.
    
    Furthermore, (4) implies (1) by Obs.~\ref{obs:Fitch}.  Conversely,
    $(G,\sigma)$ in Statement (2) is an rs-Fitch graph and an LDT
    graph. Hence it is properly colored by Prop.~\ref{prop:properCol}.  Thus
    (2) implies (4).
    
    Statement (3), in particular, implies that every properly colored
    complete multipartite $(G,\sigma)$ is an LDT graph and, thus, there is a
    relaxed scenario $\scen$ such that $G=\Gu(\scen)$. Now, we can apply
    Lemma~\ref{lem:Ru-SpeciesTriple} to conclude that $\Sri(G,\sigma)$ is
    compatible, which completes the proof.  
  \end{proof}
  
  \begin{corollary}
    A colored graph $(G,\sigma)$ is an LDT graph and an rs-Fitch graph if and
    only if $(G,\sigma)$ is a properly colored complete multipartite graph
    (and thus, a properly colored Fitch graph for some 0/1-edge-labeled
    tree).
    \label{cor:scen-sat-fitch}
  \end{corollary}
  \begin{proof}
    If $(G,\sigma)$ is an rs-Fitch graph then, by Obs.~\ref{obs:Fitch}, $G$
    is a complete multipartite graph. Moreover, since $(G,\sigma)$ is an LDT
    graph, $(G,\sigma)$ is properly colored (cf.\
    Prop.~\ref{prop:properCol}).  Conversely, if $(G,\sigma)$ is a properly
    colored complete multipartite graph it is, by 
    Thm.~\ref{thm:FitchRu-scenario}(2), an rs-Fitch graph and an LDT graph.  Now
    the equivalence between Statements~(1) and (3) in
    Thm.~\ref{thm:FitchRu-scenario} shows that $(G,\sigma)$ is an LDT graph.
  \end{proof}
  
  \begin{corollary}
    Let $(G,\sigma)$ be a vertex-colored graph. If $\Tri(G) = \emptyset$ and
    $\Sri(G,\sigma)$ is incompatible, then $G$ is a complete multipartite
    graph (and thus, a Fitch graph for some 0/1-edge-labeled tree), but
    $\sigma$ is not a proper vertex coloring of $G$.
    \label{cor:Fitch-compatible}
  \end{corollary}
  \begin{proof}
    By definition, if $\Tri(G)=\emptyset$, then $G$ cannot contain an induced
    $K_2+K_1$. By Prop.~\ref{prop:fitch}, $G$ is a Fitch graph.
    Contraposition of the last statement in Thm.~\ref{thm:FitchRu-scenario}
    and $G$ being a Fitch graph for some $(T,\lambda)$ implies that $\sigma$
    is not a proper vertex coloring of $G$.  
  \end{proof}
  
  As outlined in the main part of this paper, LDT graphs are sufficient to
  describe replacing HGT. They fail, however, to describe additive HGT
  in full detail.
  
  \subsection{rs-Fitch Graphs with General Colorings}
  
  In scenarios with additive HGT, the rs-Fitch graph is no longer properly
  colored and no-longer coincides with the LDT graph. Since not every
  vertex-colored complete multipartite graphs $(G,\sigma)$ is an rs-Fitch
  graph (cf. Thm.~\ref{thm:char-rsFitch}), we ask whether an LDT graph
  $(G,\sigma)$ that is not itself already an rs-Fitch graph imposes
  constraints on the rs-Fitch graphs $(\gfitch(\scen),\sigma)$ that derive
  from  relaxed  scenarios $\scen$ that explain $(G,\sigma)$. As a first step 
  towards this goal, we aim to characterize rs-Fitch graphs, i.e., to 
  understand 
  the conditions imposed by the existence of an underlying scenario $\scen$ on
  the compatibility of the collection of independent sets $\mathcal{I}$ of
  $G$ and the coloring $\sigma$. As we shall see, these conditions can be
  explained in terms of an auxiliary graph that we introduce in a very
  general setting:
  \begin{definition}
    Let $L$ be a set, $\sigma\colon L\to M$ a map and
    $\mathcal{I}=\{I_1,\dots, I_k\}$ a set of subsets of $L$.  Then the graph
    $\auxfitch(\sigma,\mathcal{I})$ has vertex set $M$ and edges $xy$ if and
    only if $x\ne y$ and $x,y\in \sigma(I')$ for some $I'\in\mathcal{I}$. We
    define an edge labeling $\ell: E(\auxfitch) \to 2^{\mathcal{I}}$ such
    that
    $\ell(e) \coloneqq \{I\in\mathcal{I}\mid \exists x,y\in I \text{ s.t.\ }
    \sigma(x)\sigma(y)=e\}$.
    \label{def:auxfitch}
  \end{definition}
  By construction $\auxfitch(\sigma,\mathcal{I'})$ is a subgraph of
  $\auxfitch(\sigma,\mathcal{I})$ whenever
  $\mathcal{I'}\subseteq\mathcal{I}$. The labeling of an edge $e$ records the
  sets $I\in\mathcal{I}$ that imply the presence of the edge.
  
  \begin{theorem}\label{thm:char-rsFitch}
    A graph $(G,\sigma)$ is an rs-Fitch graph if and only if (i) it is
    complete multipartite with independent sets
    $\mathcal{I}=\{I_1,\dots, I_k\}$, and (ii) if $k>1$, there is an
    independent set $I'\in \mathcal{I}$ such that
    $\auxfitch(\sigma,\mathcal{I}\setminus\{I'\})$ is disconnected.
  \end{theorem}
  \begin{proof}
    Let $G=(L,E)$ be a graph with coloring $\sigma\colon L\to M$.  Suppose
    first that $G$ satisfies~(i) and~(ii).  To show that $(G,\sigma)$ is an
    rs-Fitch graph, we will construct a relaxed scenario
    $\scen=(T,S,\sigma,\mu,\tT,\tS)$ such that $G = \gfitch(\scen)$.  If
    $k=1$, or equivalently $E=\emptyset$, then the relaxed scenario $\scen$
    constructed in the proof of Lemma~\ref{lem:Rempty} satisfies
    $G=\gfitch(\scen)$, i.e., both graphs are edgeless.  Now assume that
    $k>1$ and thus, $E\neq \emptyset$.  Hence, we can choose an independent
    set $I'\in \mathcal{I}$ such that
    $\auxfitch'\coloneqq \auxfitch(\sigma,\mathcal{I}\setminus\{I'\})$ is
    disconnected.  Note that $\mathcal{I}\setminus\{I'\}$ is non-empty since
    $k>1$.  Moreover, since $\auxfitch'$ is a disconnected graph on the color
    set $M$, there is a connected component $C$ of $\auxfitch'$ such that
    $(M\setminus C) \cap \sigma(I')\ne\emptyset$. Hence
    $M_1\coloneqq M\setminus C$ and $M_2\coloneqq C$ form a bipartition of
    $M$ such that neither $M_1$ nor $M_2$ are empty sets.
    
    We continue by showing that every $I\in \mathcal{I}\setminus \{I'\}$
    satisfies either $\sigma(I)\subseteq M_1$ or $\sigma(I)\subseteq M_2$.
    To see this, assume, for contradiction, that there are colors
    $A\in \sigma(I)\cap M_1$ and $B\in \sigma(I)\cap M_2$ for some
    $I\in \mathcal{I}\setminus \{I'\}$. Thus, $B\in C$ and, by
    definition,  $AB\in E(\auxfitch')$. Therefore, $A$ and $B$ must
    lie in the connected component $C$; a contradiction.  Therefore, we can
    partition $\mathcal{I}\setminus \{I'\}$ into
    $\mathcal{I}_1\coloneqq \{I\in\mathcal{I}\setminus \{I'\} \mid
    \sigma(I)\subseteq M_1\}$ and
    $\mathcal{I}_2\coloneqq \{I\in\mathcal{I}\setminus \{I'\} \mid
    \sigma(I)\subseteq M_2\}$.  Note that one of the sets $\mathcal{I}_1$ and
    $\mathcal{I}_2$, but not both of them, may be empty. This may be the
    case, for instance, if $\sigma$ is not surjective.
    
    Now, we construct a relaxed scenario $\scen = (T,S,\sigma,\mu,\tT,\tS)$
    with coloring $\sigma$ such that $G=\gfitch(\scen)$.  We first define the
    species tree $S$ as the planted tree where $\rho_{S}$ (i.e.\ the single
    child of $0_S$) hast two children $w_1$ and $w_2$.  If $|M_1|=1$, we
    identify $w_1$ with the single element in $M_1$, and otherwise, we set
    $\child_S(w_1)=L(S(w_1))\coloneqq M_1$.  We proceed analogously for $w_2$
    and $M_2$.  Thus, $S$ is phylogenetic by construction.  We choose the
    time map $\tS$ by putting $\tS(0_S)=2$, $\tS(\rho_S)=1$,
    $\tS(w_1)=\tS(w_2)=0.5$ and $\tS(x)=0$ for all $x\in L(S)$. This
    completes the construction of $S$ and $\tS$.
    
    We proceed with the construction of the gene tree $T$, its time map $\tT$
    and the reconciliation map $\mu$.  This tree $T$ has leaf set $L$,
    planted root $0_T$, and root $\rho_T$.  We set $\mu(0_T)=0_S$ and
    $\tT(0_T)=\tS(0_S)=2$, and moreover $\mu(x)=\sigma(x)$ and $\tT(x)=0$ for
    all $x\in L$.
    
    For each $I_j\in \mathcal{I}\setminus\{I'\}$, we add a vertex $u_j$. We
    will later specify how these vertices are connected (via paths) to
    $\rho_T$. If $|I_j|=1$, $u_j$ becomes a leaf of $T$ that is
    identified with the unique element in $I_j$.  Otherwise, we add
    exactly $|I_j|$ children to $u_j$, each of which is identified with
    one of the elements in $I_j$. If $u_j$ is a leaf, we already defined
    $\mu(u_j)=\sigma(u_j)$ and $\tT(u_j)=0$.
    
    Otherwise, we set $\tT(u_j)=0.6$ and $\mu(u_j)=(\rho_S,w_1)$ if
    $I_j\in\mathcal{I}_1$ and $\mu(u_j)=(\rho_S,w_2)$ if
    $I_j\in\mathcal{I}_2$.  Recall that $M_1\cap
    \sigma(I')\ne\emptyset$. However, both
    $M_2\cap \sigma(I')\ne\emptyset$ and $M_2\cap \sigma(I')=\emptyset$
    are possible. The latter case appears e.g.\ whenever
    $\auxfitch(\sigma,\mathcal{I})$ was already disconnected.  To connect
    the vertices $u_j$ to $\rho_T$, we distinguish the three mutually
    exclusive cases:
    
    \begin{figure}[h]
      \begin{center}
        \includegraphics[width=0.85\textwidth]{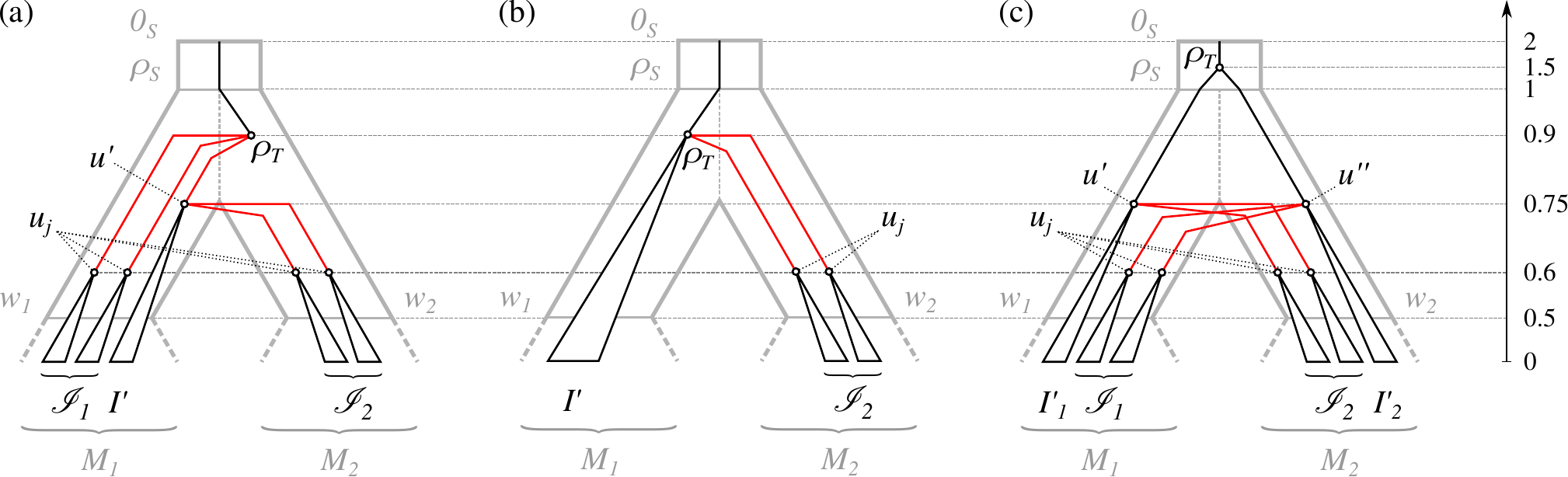} 
      \end{center}
      \caption{Illustration of the  relaxed  scenario constructed in the
        \emph{if}-direction of the proof of Thm.~\ref{thm:char-rsFitch}. For
        Cases~(a) and~(c), only the situation in which a vertex $u'$ and
        $u''$, resp., is necessary is shown. Otherwise, the single element in
        $I'$, $I'_1$ or $I'_2$ would be a child of the root $\rho_T$.
        Moreover, the vertices $u_j$ are drawn under the assumption that
        $|I_j|>1$. Otherwise, there are identified with the single leaf in
        $I_j$.}
      \label{fig:rs-fitch-charac}
    \end{figure}
    
    \par\noindent\emph{Case~(a): $M_2\cap \sigma(I')=\emptyset$ and
      $\mathcal{I}_1\ne\emptyset$.}\\ We set $\mu(\rho_T)=(\rho_S,w_2)$ and
    $\tT(\rho_T)=0.9$.  We attach all $u_j$ that correspond to elements
    $I_j\in \mathcal{I}_1$ as children of $\rho_T$.  If $|I'|> 1$ or
    $\mathcal{I}_2\ne\emptyset$, we create a vertex $u'$ to which all
    elements in $I'$ and all $u_j$ such that $I_j\in \mathcal{I}_2$ are
    attached as children, attach $u'$ as a child of $\rho_T$, and set
    $\mu(u')=(\rho_S,w_1)$ and $\tT(u')=0.75$.  Otherwise, we simply attach
    the single element $x'$ in $I'$ as a child of $\rho_T$.  Clearly, the so
    constructed tree $T$ is phylogenetic.  Note that the edges
    $(\rho_T, u_j)$ with $I_j\in \mathcal{I}_1$ as well as the edges
    $(u',u_j)$ with $I_j\in \mathcal{I}_2$ are transfer edges.  Together with
    $(\rho_T,u')$ or $(\rho_T,x)$, respectively, these are the only transfer
    edges.
    
    \par\noindent\emph{Case~(b): $M_2\cap \sigma(I')=\emptyset$ and
      $\mathcal{I}_1=\emptyset$.}\\ By the arguments above, the latter
    implies $\mathcal{I}_2\ne\emptyset$.  Hence, we can set
    $\mu(\rho_{T})=(\rho_S,w_1)$ and $\tT(\rho_T)=0.9$ and attach all
    elements of $I'$ as well as the vertices $u_j$ corresponding to the
    independent sets $I_j\in\mathcal{I}_2=\mathcal{I}\setminus \{I'\}$ as
    children of $\rho_T$.  Since $|I'|\ge 1$ and $\mathcal{I}_2\ge 1$, the
    tree $T$ obtained in this manner is again phylogenetic.  Moreover, note
    that the transfer edges are exactly the edges $(\rho_T,u_j)$.
    
    \par\noindent\emph{Case~(c): $M_2\cap \sigma(I')\ne\emptyset$.}\\
    In this case, the sets $I'_1\coloneqq\{x\in I'\mid \sigma(x)\in M_1\}$
    and $I'_2\coloneqq\{x\in I'\mid \sigma(x)\in M_2\}$ must be non-empty.
    We set $\mu(\rho_T)=(0_T,\rho_T)$ and $\tT(\rho_T)=1.5$.  If $|I'_1|> 1$
    or $\mathcal{I}_2\ne\emptyset$, we create a vertex $u'$ to which all
    elements in $I'_1$ and all $u_j$ such that $I_j\in \mathcal{I}_2$ are
    attached as children, and set $\mu(u')=(\rho_S,w_1)$ and $\tT(u')=0.75$.
    Otherwise, we simply attach the single element in $I'_1$ as a child of
    $\rho_T$.  For the ``other side'', we proceed analogously: If $|I'_2|> 1$
    or $\mathcal{I}_1\ne\emptyset$, we create a vertex $u''$ to which all
    elements in $I'_2$ and all $u_j$ such that $I_j\in \mathcal{I}_1$ are
    attached as children, and set $\mu(u')=(\rho_S,w_2)$ and $\tT(u'')=0.75$.
    Otherwise, we simply attach the single element in $I'_2$ as a child of
    $\rho_T$.  By construction, the so constructed tree is again
    phylogenetic.  Moreover, the transfer edges are exactly the edges
    $(u',u_j)$ and $(u'',u_j)$.
    
    Using Fig.~\ref{fig:rs-fitch-charac}, one can easily verify that, in all
    three Cases~(a)-(c), the reconciliation map $\mu$ is time-consistent
    with $\tT$ and $\tS$.  Thus, $\scen$ is a relaxed scenario.  Moreover,
    Fig.~\ref{fig:rs-fitch-charac} together with the fact that
    $\sigma(I)\subseteq M_1$ holds for all $I\in \mathcal{I}_1$, and
    $\sigma(I)\subseteq M_2$ holds for all $I\in \mathcal{I}_2$, shows that
    $G=\gfitch(\scen)$ in all three cases.  Hence, $(G,\sigma)$ is an
    rs-Fitch graph.
    
    For the \emph{only-if}-direction, assume that $(G=(V,E),\sigma)$ is an
    rs-Fitch graph.  Hence, there exists a  relaxed  scenario
    $\scen=(T,S,\sigma,\mu,\tT,\tS)$ such that $G = \gfitch(\scen)$.  By
    Obs.~\ref{obs:Fitch} and Prop.~\ref{prop:fitch}, $(G,\sigma)$ is a
    complete multipartite graph that is determined by its set of independent
    sets $\mathcal{I}=\{I_1,\dots,I_k\}$.  Hence, Condition (i) is satisfied.
    
    Now assume, for contradiction, that Condition (ii) is violated.  Thus
    $k\ge 2$ and there is no independent set $I'\in \mathcal{C}$ such that
    $\auxfitch(\sigma,\mathcal{I}\setminus\{I'\})$ is disconnected.  If
    $|M|=1$, then the species tree $S$ only consists of the planted root
    $0_S$ and the root $\rho_S$, which in this case is identified with the
    single element in $M$.  Clearly, all vertices and edges are comparable in
    such a tree $S$, and hence, there is no transfer edges in $\scen$,
    implying $E = \emptyset$ and thus $|\mathcal{I}| = 1$; a contradiction to
    $k\ge 2$.
    
    Thus we have $|M|\ge 2$ and the root $\rho_S$ of the species tree $S$ has
    at least two children.  Since
    $\auxfitch(\sigma,\mathcal{I}\setminus\{I'\})$ is connected for every
    $I'\in \mathcal{C}$, the graph $\auxfitch(\sigma,\mathcal{I})$ is also
    connected. Since each color appears at most once as a leaf of $S$,
    $\sigma(L(S(v_1))) \cap \sigma(L(S(v_2)))=\emptyset$ holds for any two
    distinct children $v_1,v_2\in\child_S (\rho_S)$. These three assertions,
    together with the definition of the auxiliary graph
    $\auxfitch(\sigma,\mathcal{I})$, imply that there are two distinct colors
    $A, B\in M$ such that $AB$ is an edge in $\auxfitch(\sigma,\mathcal{I})$,
    $A\preceq_S v_1$ and $B\prec_{S} v_2$ for distinct children
    $v_1,v_2\in\child_S (\rho_S)$. By definition of
    $\auxfitch(\sigma,\mathcal{I})$ there is an independent set
    $I'\in\mathcal{I}$ containing a vertex $a\in I'$ with $\sigma(a)=A$ and a
    vertex $b\in I'$ with $\sigma(b)=B$.  Since $a$ and $b$ lie in the same
    independent set, we have $ab\notin E$.  By
    Lemma~\ref{lem:independent-lca},
    $\mu(\lca_T(a,b)) \succeq_S \lca_S(A,B)=\rho_S$.  Since, by assumption,
    $\auxfitch(\sigma,\mathcal{I}\setminus\{I'\})$ is also connected, we find
    two distinct colors $C$ and $D$ (not necessarily distinct from $A$ and
    $B$) such that $CD$ is an edge in $\auxfitch(\sigma,\mathcal{I})$,
    $C\preceq_S v_3$ and $D\prec_{S} v_4$ for distinct children
    $v_3,v_4\in\child_S (\rho_S)$ (but not necessarily distinct from $v_1$
    and $v_2$), and in particular, an independent set
    $I''\in\mathcal{I}\setminus \{I'\}$ containing a vertex $c\in I''$ with
    $\sigma(c)=C$ and a vertex $d\in I''$ with $\sigma(d)=D$.
    By construction, $I'\ne I''$, and thus, all edges between $I'$ and $I''$
    exist in $G$, in particular the edges $ac,ad,bc,bd$.  Since $c,d\in I''$,
    we have $cd\notin E$ and thus, by Lemma~\ref{lem:independent-lca},
    $\mu(\lca_T(c,d)) \succeq_S \lca_S(C,D)=\rho_S$.
    
    We now consider the unique path $P$ in $T$ that connects $\lca_T(a,b)$
    and $\lca_T(c,d)$.  Since $\mu$ is time-consistent and
    $\mu(\lca_T(a,b)), \mu(\lca_T(c,d)) \succeq_S \rho_S$, we conclude that,
    for every edge $uv$ along this path $P$, we have
    $\mu(u), \mu(v)\succeq_S \rho_S$ and thus
    $\mu(u), \mu(v)\in \{\rho_S, (0_S,\rho_S)\}$.  But then, $\mu(u)$ and
    $\mu(v)$ are comparable in $S$.  Therefore, $P$ does not contain any
    transfer edge.  Since $ab\notin E$, the path connecting $a$ and
    $\lca_{T}(a,b)$ does not contain any transfer edges. Likewise,
    $cd\notin E$ implies that the path connecting $c$ and $\lca_{T}(c,d)$
    does not contain any transfer edges.  Thus, the path connecting $a$ and
    $c$ also does not contain any transfer edge, which implies that
    $ac\notin E(\gfitch(\scen))=E$; a contradiction since $a$ and $c$ belong
    to two distinct independent sets.
    
    Hence, we conclude that for $k>1$ there exists an independent set
    $I'\in \mathcal{C}$ such that
    $\auxfitch(\sigma,\mathcal{I}\setminus\{I'\})$ is disconnected.  
  \end{proof}
  
  \begin{corollary}
    rs-Fitch graphs can be recognized in polynomial time.
    \label{cor:auxfitch1}
  \end{corollary}
  \begin{proof}
    Every rs-Fitch graph $(G,\sigma)$ must be complete multipartite, which
    can be verified in polynomial time. In this case, the set of independent
    sets $\mathcal{I}=\{I_1,\dots, I_k\}$ of $G$ can also be determined and
    the graph $\auxfitch(\sigma,\mathcal{I})$ can be constructed in polynomial
    time.  Finally, we need to find an independent set $I'\in \mathcal{I}$,
    such that $\auxfitch(\sigma,\mathcal{I}\setminus\{I'\})$ is disconnected.
    Clearly, checking whether $\auxfitch(\sigma,\mathcal{I}\setminus\{I'\})$
    is disconnected can be done in polynomial time and since there are at
    most $|V(G)|$ independent sets in $\mathcal{I}$, finding an independent
    set $I'$ such that $\auxfitch(\sigma,\mathcal{I}\setminus\{I'\})$ is
    disconnected (if one exists) can be done in polynomial time as well.
  \end{proof}
  
  \begin{corollary}
    Let $(G,\sigma)$ be a complete multipartite graph with coloring
    $\sigma\colon V(G) \to M$ and set of independent sets $\mathcal{I}$.
    Then, $(G,\sigma)$ is an rs-Fitch graph if and only if
    $\auxfitch(\sigma,\mathcal{I})$ is disconnected or there is a cut
    $Q\subseteq E(\auxfitch(\sigma,\mathcal{I}))$ such that all edges
    $e\in Q$ have the same label $\ell(e)=\{I\}$ for some $I\in\mathcal{I}$.
    \label{cor:auxfitch2}
  \end{corollary}
  \begin{proof}
    If $\auxfitch(\sigma,\mathcal{I})$ is disconnected, then
    $\auxfitch(\sigma,\mathcal{I}\setminus \{I\})$ remains disconnected for
    all $I\in \mathcal{I}$ and, by Thm.~\ref{thm:char-rsFitch},
    $(G,\sigma)$ is an rs-Fitch graph.
    
    If there is a cut $Q\subseteq E(\auxfitch(\sigma,\mathcal{I}))$ such that
    all edges $e\in Q$ have the same label $\ell(e)=\{I\}$ for some
    $I\in\mathcal{I}$, then, by definition,
    $E(\auxfitch(\sigma,\mathcal{I}\setminus \{I\}))\subseteq E'\coloneqq
    E(\auxfitch(\sigma,\mathcal{I}))\setminus Q$. Since $Q$ is a cut
    in $\auxfitch(\sigma,\mathcal{I})$, the resulting graph
    $\auxfitch'= (M,E')$ is disconnected. By the latter arguments,
    $\auxfitch(\sigma,\mathcal{I}\setminus \{I\})$ is a subgraph of
    $\auxfitch'$, and thus, disconnected as well. By
    Thm.~\ref{thm:char-rsFitch}, $(G,\sigma)$ is an rs-Fitch graph.
    
    Conversely, if $(G,\sigma)$ is an rs-Fitch graph, then 
    Thm.~\ref{thm:char-rsFitch} implies that
    $\auxfitch(\sigma,\mathcal{I}\setminus \{I\})$ is disconnected for some
    $I\in \mathcal{I}$.  If $\auxfitch(\sigma,\mathcal{I})$ was already
    disconnected, then there is nothing to show. Hence assume that
    $\auxfitch(\sigma,\mathcal{I}) = (M,E)$ is connected and let
    $\auxfitch(\sigma,\mathcal{I}\setminus \{I\}) = (M,E')$.  Moreover, let
    $F\subseteq E$ be the subset of edges $e\in E$ with $I\in \ell(e)$. Note,
    $F$ contains all edges of $E$ that have potentially been removed from $E$
    to obtain $E'$. However, all edges $e=xy$ in $F$ with $|\ell(e)|>1$ must
    remain in $\auxfitch(\sigma,\mathcal{I}\setminus \{I\})$, since there is
    another independent set $I'\in \ell(e)\setminus \{I\}$ such that
    $x,y\in \sigma(I')$. Hence, only those edges $e$ in $F$ for which
    $|\ell(e)|=1$ are removed from $E$. Hence, there is a cut
    $Q\subseteq F\subseteq E$ such that all edges $e\in Q$ have the same
    label $\ell(e)=\{I\}$ for some $I\in\mathcal{I}$.  
  \end{proof}
  
  \begin{corollary}
    If $(G,\sigma)$ with coloring $\sigma\colon V(G) \to M$ is an rs-Fitch
    graph, then there are no two disjoint independent sets $I$ and $I'$ of $G$
    with $\sigma(I)=\sigma(I')= M$.
    \label{cor:auxfitch3}
  \end{corollary}
  \begin{proof}
    Let $\mathcal{I}$ be the set of independent sets of $G$.  If
    $|\mathcal{I}|=1$, there is nothing to show and thus, we assume that
    $|\mathcal{I}|>1$.  Assume, for contradiction, that there are two
    distinct independent sets $I, I' \in \mathcal{I}$ such that
    $\sigma(I)=\sigma(I')= M$.  For every $I''\in\mathcal{I}$, the set
    $\mathcal{I}\setminus \{I''\}$ clearly contains at least one of the two
    sets $I$ and $I'$, both of which contain all colors in $M$. Therefore,
    $\auxfitch(\sigma, \mathcal{I}\setminus \{I''\})$ is the complete graph
    by construction and, thus, connected for every
    $I''\in\mathcal{I}$. This together with Thm.~\ref{thm:char-rsFitch}
    implies that $(G,\sigma)$ is not an rs-Fitch graph; a contradiction.
  \end{proof}
  
  \begin{corollary}\label{cor:surj-rsF}
    Every complete multipartite graph $(G,\sigma)$ with a vertex coloring
    $\sigma\colon V(G) \to M$ that is not surjective is an rs-Fitch graph.
  \end{corollary}
  \begin{proof}
    If $\sigma\colon V(G) \to M$ is not surjective, then
    $\auxfitch(\sigma,\mathcal{I})$ is disconnected, where $\mathcal{I}$
    denotes the set of independent sets of $G$.  Hence, if $k>1$, then
    $\auxfitch(\sigma,\mathcal{I}\setminus \{I\})$ remains disconnected for
    all $I\in \mathcal{I}$.  By Thm.~\ref{thm:char-rsFitch}, $(G,\sigma)$ is
    an rs-Fitch graph.  
  \end{proof}
  
  Cor.~\ref{cor:surj-rsF} may seem surprising since it implies that the
  property of being an rs-Fitch graph can depend on species (colors $M$) for
  which we have no genes $L$ in the data. The reason is that an additional
  lineage in the species tree provides a place to ``park'' interior vertices
  in the gene tree from which HGT-edges can emanate that could not always be
  accommodated within lineages that have survivors -- where they may force
  additional HGT edges.
  
  \begin{corollary}
    Every Fitch graph $(G,\sigma)$ that contains an independent set $I$ and a
    vertex $x\in I$ with $\sigma(x)\notin\sigma(I')$ for all other
    independent sets $I'\neq I$, is an rs-Fitch graph.
    \label{cor:re-fitsh-isolatedcolor}
  \end{corollary}
  \begin{proof}
    Let $\mathcal{I}$ denote the set of independent sets of $G$.  If there is
    an independent set $I\in \mathcal{I}$ that contains a vertex $x\in I$
    with $\sigma(x)\notin \sigma(I')$ for all other independent sets
    $I'\neq I$, then the vertex $\sigma(x)$ in
    $\auxfitch(\sigma,\mathcal{I}\setminus \{I\})$ is an isolated vertex and
    thus, $\auxfitch(\sigma,\mathcal{I}\setminus \{I\})$ is disconnected.  By
    Thm.~\ref{thm:char-rsFitch}, $(G,\sigma)$ is an rs-Fitch graph.  
  \end{proof}
  
  \par\noindent
  As for LDT graphs, the property of being an rs-Fitch graph is hereditary.
  \begin{corollary}
    If $(G=(L,E),\sigma)$ is an rs-Fitch graph, then the colored vertex induced 
    subgraph
    $(G[W],\sigma_{|W})$ is an rs-Fitch graph for all non-empty subsets
    $W\subseteq L$.
    \label{cor:rsFitch-hereditary}
  \end{corollary}
  \begin{proof}
    It suffices to show the statement for $W = L\setminus\{x\}$ for an
    arbitrary vertex $x\in L$.  If $G=(L,E)$ is edgeless, then
    $G[W]$ is edgeless and thus, by Thm.~\ref{thm:char-rsFitch},
    an rs-Fitch graph.
    
    Thus, assume that $E\neq \emptyset$ and thus, for the set $\mathcal{I}$
    of independent sets of $G$ it holds that $|\mathcal{I}|>1$.  Since $G$
    does not contain an induced $K_2+K_1$, it is easy to see that $G[W]$
    cannot contain an induced $K_2+K_1$ and thus, $G[W]$ is a complete
    multipartite graph.  Hence, Thm.~\ref{thm:char-rsFitch}(i) is satisfied.
    Moreover, if for the set $\mathcal{I}'$ of independent sets of $G[W]$
    it holds that $|\mathcal{I}'|=1$ then, Thm.~\ref{thm:char-rsFitch}
    already shows that $(G[W],\sigma_{|W})$ is an rs-Fitch graph.
    
    Thus, assume that $|\mathcal{I}'|>1$.  Now compare the labeling $\ell$ of
    the edges in $\auxfitch = \auxfitch(\sigma, \mathcal{I})$ and the
    labeling $\ell'$ of the edges in
    $\auxfitch' = \auxfitch(\sigma_{|W}, \mathcal{I}')$.  Note, $\auxfitch$
    and $\auxfitch'$ have still the same vertex set $M$.  Let
    $I\in\mathcal{I}$ with $x\in I$. For all vertices $y\in I$ with
    $\sigma(x)\neq \sigma(y)$, we have an edge $e =\sigma(x)\sigma(y)$ in
    $\auxfitch$ and $I\in \ell(e)$. Consequently, for all edges $e$ of
    $\auxfitch$ that are present in $\auxfitch'$ we have
    $\ell'(e)\subseteq \ell(e)$. In particular, $\auxfitch'$ cannot have
    edges that are not present in $\auxfitch$, since we reduced for one
    independent set the size by one.  Therefore,
    $\auxfitch'$ is a subgraph of $\auxfitch$.
    
    By Thm.~\ref{thm:char-rsFitch}, there is an independent set
    $I'\in \mathcal{I}$, not necessarily distinct from $I$, such that
    $\auxfitch(\sigma,\mathcal{I}\setminus\{I'\})$ is disconnected.  If
    $I' = \{x\}$, then $\mathcal{I}' = \mathcal{I}\setminus \{I'\}$ and
    $\auxfitch' =\auxfitch$ must be disconnected as well.  Otherwise,
    $\auxfitch'\subseteq \auxfitch$ and similar arguments as above show that
    $\auxfitch(\sigma,\mathcal{I'}\setminus\{I'\}) \subseteq
    \auxfitch(\sigma,\mathcal{I}\setminus\{I'\})$. Therefore, in both of the
    latter cases, $\auxfitch(\sigma,\mathcal{I'}\setminus\{I'\})$ is
    disconnected and Thm.~\ref{thm:char-rsFitch} implies that
    $(G[W],\sigma_{|W})$ is an rs-Fitch graph.  
  \end{proof}
  
  As outlined in the main part of this paper,
  Cor.~\ref{cor:rsFitch-hereditary} is usually not satisfied if we restrict
  the codomain of $\sigma$ to the observable part of colors, even if $\sigma$
  is surjective.
  
  \subsection{Least Resolved Trees for Fitch graphs}
  \label{ssec:LRTFitch}
  
  It is important to note that the characterization of rs-Fitch graphs in
  Thm.~\ref{thm:char-rsFitch} does not provide us with a characterization of
  rs-Fitch graphs that share a common  relaxed  scenario with a given LDT
  graph. As a potential avenue to address this problem we investigate the
  structure of least-resolved trees for Fitch graphs as possible source of
  additional constraints.
  
  \emph{All trees considered in this subsection \ref{ssec:LRTFitch} are
    rooted and phylogenetic but not planted unless stated differently.} This
  is no loss of generality, since we are interested in Fitch-least-resolved
  trees, which are never planted because the edge incident with the
  planted root can be contracted without affecting the paths between the
  leaves.
  
  \begin{definition}\label{def:FLRT}
    The edge-labeled tree $(T,\lambda)$ is \emph{Fitch-least-resolved}
    w.r.t.\ $\gfitch(T,\lambda)$, if for all trees $T'\neq T$ that are
    displayed by $T$ and every labeling $\lambda'$ of $T'$ it holds that
    $\gfitch(T,\lambda)\neq \gfitch(T',\lambda')$.
  \end{definition}
  
  \begin{definition} \label{def:contract} Let $(T,\lambda)$ be an
    edge-labeled  tree and let $e=(x,y)\in E(T)$ be an inner
    edge.  The tree $(T_{/e}, \lambda_{/e})$ with $L(T_{/e})=L(T)$, is
    obtained by contraction of the edge $e$ in $T$ and by keeping the edge
    labels of all non-contracted edges.
  \end{definition}
  Note, if $e$ is an inner edge of a phylogenetic tree $T$, then the tree
  $T_{/e}$ is again phylogenetic.
  
  \begin{definition}\label{def:rel-label}
    An edge $e$ in $(T,\lambda)$ is \emph{relevantly-labeled in
      $(T,\lambda)$} if, for the tree $(T,\lambda')$ with
    $\lambda'(f)=\lambda(f)$ for all $f\in E(T)\setminus\{e\}$ and
    $\lambda'(e)\neq \lambda(e)$, it holds that
    $\gfitch(T,\lambda)\neq \gfitch(T,\lambda')$.
  \end{definition}
  
  \begin{lemma}
    An outer 0-edge $e=(v,x)$ in $(T,\lambda)$ is \emph{relevantly-labeled in
      $(T,\lambda)$} if and only if $zx\notin E(\gfitch(T,\lambda))$ for some
    $z\in L(T)\setminus \{x\}$.
    \label{lem:rel-label}
  \end{lemma}
  \begin{proof}
    Assume that $e=(v,x)$ is a relevantly-labeled outer 0-edge.  Hence, for
    $(T,\lambda')$ with $\lambda'(f)=\lambda(f)$ for all
    $f\in E(T)\setminus\{e\}$ and $\lambda'(e)=1$, it holds that
    $\gfitch(T,\lambda)\neq \gfitch(T,\lambda')$.  Since we only changed the
    label of the outer edge $(v,x)$, it still holds that
    $yy'\in E(\gfitch(T,\lambda'))$ if and only if
    $yy'\in E(\gfitch(T,\lambda))$ for all distinct
    $y,y'\in L(T)\setminus \{x\}$.  Moreover, since $\lambda'(e)=1$ and
    $e=(v,x)$ is an outer edge, we have $xz\in E(\gfitch(T,\lambda'))$ for
    all $z\in L(T)\setminus \{x\}$.  Thus,
    $\gfitch(T,\lambda)\neq \gfitch(T,\lambda')$ implies that
    $xz\notin E(\gfitch(T,\lambda))$ for at least one
    $z\in L(T)\setminus \{x\}$.
    
    Now, suppose that $zx\notin E(\gfitch(T,\lambda))$ for some
    $z\in L(T)\setminus \{x\}$.  Clearly, this implies that the outer edges
    $e=(v,x)$ and $f=(w,z)$ must be 0-edges and changing one of them to a
    1-edge would imply that $xz$ becomes an edge in the Fitch graph.  Hence,
    $e$ is relevantly-labeled in $(T,\lambda)$.  
  \end{proof}
  
  \begin{lemma}
    For every tree $(T,\lambda)$ and every inner 0-edge $e$ of $T$, it holds 
    $\gfitch(T,\lambda)=\gfitch(T_{/e},\lambda_{/e})$.
    \label{lem:contract-0-edge}
  \end{lemma}
  \begin{proof}
    Suppose that $(T,\lambda)$ contains an inner 0-edge $e=(u,v)$.  The
    contraction of this edge does not change the number of 1-edges along the
    paths connecting any two leaves.  It affects the least common ancestor of
    $x$ and $y$, if $\lca_T(x, y) = u$ or $\lca_T (x, y) = v$.  In either
    case, however, the number of 1-edges between $\lca_T (x, y)$ and the
    leaves $x$ and $y$ remains unchanged.  Hence, we have
    $\gfitch(T,\lambda) = \gfitch(T_{/e},\lambda_{/e})$.  
  \end{proof}
  
  \begin{lemma}
    If $(T,\lambda)$ is a Fitch-least-resolved tree w.r.t.\
    $\gfitch(T,\lambda)$, then it does neither contain inner 0-edges nor
    inner 1-edges that are not relevantly-labeled.
    \label{lem:innerEdgesinLRT}
  \end{lemma}
  \begin{proof}
    Suppose first, by contraposition, that $(T,\lambda)$ contains an inner
    0-edge $e=(u,v)$.  By Lemma~\ref{lem:contract-0-edge},
    $\gfitch(T,\lambda) = \gfitch(T_{/e},\lambda_{/e})$, and thus,
    $(T,\lambda)$ is not Fitch-least-resolved.
    
    Assume now, by contraposition, that $(T,\lambda)$ contains an inner
    1-edge $e$ that is not relevantly-labeled.  Hence, we can put
    $\lambda'(e)=0$ and $\lambda(f)=\lambda(f')$ for all
    $f\in E(T)\setminus \{e\}$ and obtain
    $\gfitch(T,\lambda) = \gfitch(T,\lambda')$.  Since $(T,\lambda')$
    contains an inner 0-edge, it cannot be Fitch-least-resolved. Therefore
    and by definition, $(T,\lambda)$ cannot be Fitch-least-resolved as well.
  \end{proof}
  
  The converse of Lemma~\ref{lem:innerEdgesinLRT} is, however, not always
  satisfied.  To see this, consider the Fitch graph $G \simeq K_3$ with
  vertices $x,y$ and $z$.  Now, consider the tree $(T,\lambda)$ where $T$ is
  the triple $xy|z$, the two outer edges incident to $y$ and $z$ are 0-edges
  while the remaining two edges in $T$ are 1-edges. It is easy to verify that
  $G=\gfitch(T,\lambda)$.  In particular, the inner edge $e$ is
  relevantly-labeled, since if $\lambda'(e) = 0$ we would have
  $yz\notin E(\gfitch(T,\lambda'))$.  However, $(T,\lambda)$ is not
  Fitch-least-resolved w.r.t.\ $G$, since the star tree $T'$ on the three
  leaves $x,y,z$ is displayed by $T$, and the labeling $\lambda'$ with
  $\lambda'(e)=1$ for all $e\in E(T')$ provides a tree $(T',\lambda')$ with
  $G=\gfitch(T',\lambda')$.
  
  \begin{lemma}
    A tree $(T,\lambda)$ is a Fitch-least-resolved tree w.r.t.\
    $\gfitch(T,\lambda)$ if and only if
    $\gfitch(T,\lambda) \neq \gfitch(T_{/e},\lambda')$ holds for all
    labelings $\lambda'$ of $T_{/e}$ and all inner edges $e$ in $T$.
    \label{lem:no-cont}
  \end{lemma}
  \begin{proof}
    Let $(T,\lambda)$ be an edge-labeled tree. Suppose first that
    $(T,\lambda)$ is Fitch-least-resolved w.r.t.\ $\gfitch(T,\lambda)$.  For
    every inner edge $e$ in $T$, the tree $T_{/e}\ne T$ is displayed by
    $T$. By definition of Fitch-least-resolved trees, we have
    $\gfitch(T,\lambda)\neq \gfitch(T_{/e},\lambda')$ for every labeling
    $\lambda'$ of $T_{/e}$.
    
    For the converse, assume, for contraposition, that $(T,\lambda)$ is not
    Fitch-least-resolved w.r.t.\ $\gfitch(T,\lambda)$.  Hence, there is a
    tree $(T',\lambda')$ such that $T'\ne T$ is displayed by $T$ and
    $\gfitch(T,\lambda) = \gfitch(T',\lambda')$.  Clearly, $T$ and $T'$ must
    have the same leaf set.  Therefore and since $T'<T$, the tree $T'$ can be
    obtained from $T$ by a sequence of contractions of inner edges
    $e_1,\dots,e_{\ell}$ (in this order) where $\ell\ge 1$.  If $\ell=1$,
    then we have $T'=T_{/e_1}$ and, by assumption,
    $\gfitch(T,\lambda) = \gfitch(T_{/e_1},\lambda')$.  Thus, we are done.
    Now assume $\ell\ge 2$.  We consider the tree $(T_{/e_1},\lambda'')$
    where $\lambda''(f)=\lambda'(f)$ if $f \in E(T')$ and $\lambda''(f)=0$
    otherwise.  Hence, $(T',\lambda')$ can be obtained from
    $(T_{/e_1},\lambda'')$ by stepwise contraction of the 0-edges
    $e_2,\dots,e_{\ell}$, and by keeping the labeling of $\lambda''$ for the
    remaining edges in each step.  Hence, we can repeatedly apply
    Lemma~\ref{lem:contract-0-edge} to conclude that
    $\gfitch(T_{/e_1},\lambda'')=\gfitch(T',\lambda')$.  Together with
    $\gfitch(T,\lambda) = \gfitch(T',\lambda')$, we obtain
    $\gfitch(T,\lambda) = \gfitch(T_{/e_1},\lambda'')$, which completes the
    proof.  
  \end{proof}
  
  As a consequence of Lemma~\ref{lem:no-cont}, it suffices to show that
  $\gfitch(T,\lambda) = \gfitch(T_{/e},\lambda')$ for some inner edge
  $e\in E(T)$ and some labeling $\lambda'$ for $T_{/e}$ to show that
  $(T,\lambda)$ is not Fitch-least-resolved tree w.r.t.\
  $\gfitch(T,\lambda)$. The next result characterizes Fitch-least-resolved
  trees and is very similar to the results for ``directed'' Fitch graphs of
  0/1-edge-labeled trees (cf.\ Lemma~11(1,3) in \cite{Geiss:18a}).  However,
  we note that we defined Fitch-least-resolved in terms of all possible
  labelings $\lambda'$ for trees $T'$ displayed by $T$, whereas
  \citet{Geiss:18a} call $(T,\lambda)$ least-resolved whenever
  $(T_{/e},\lambda_{/e})$ results in a (directed) Fitch graph that differs
  from the one provided by $(T,\lambda)$ for every $e\in E(T)$.
  
  \begin{theorem}
    Let $G$ be a Fitch graph, and $(T,\lambda)$ be a tree such that
    $G=\gfitch(T,\lambda)$.  If all independent sets of $G$ are of size one
    (except possibly for one independent set), then $(T,\lambda)$ is
    Fitch-least-resolved for $G$ if and only if it is a star tree.\\
    If $G$ has at least two independent sets of size at least two, then
    $(T,\lambda)$ is Fitch-least-resolved for $G$ if and only if
    \begin{itemize}
      \item[(a)] every inner edge of $(T,\lambda)$ is a 1-edge,
      \item[(b)] for every inner vertex $v\in V^0(T)$ there are (at least) two
      relevantly-labeled outer 0-edges $(v,x), (v,y)$ in $(T,\lambda)$
    \end{itemize}
    In particular, if distinct $x, y\in L(T)$ are in the same independent set
    of $G$, then they have the same parent in $T$ and $(\parent(x), x)$,
    $(\parent(x), y)$ are relevantly-labeled outer 0-edges.
    \label{thm:LRT-rsFitch}
  \end{theorem}
  \begin{proof}
    Suppose that every independent set of $G$ is of size one (except possibly
    for one). Let $(T,\lambda)$ be the star tree where
    $\lambda((\rho_T,v)) =1$ if and only if $v$ is the single element in an
    independent set of size one. It is now a simple exercise to verify that
    $G=\gfitch(T,\lambda)$. Since $(T,\lambda)$ is a star tree, it is clearly
    Fitch-least-resolved.  The converse follows immediately from this
    construction together with fact that the star tree is displayed by all
    trees with leaf set $V(G)$.  In the following we assume that $G$ contains
    at least two independent sets of size at least two.
    
    First suppose that $(T,\lambda)$ is Fitch-least resolved w.r.t.\
    $\gfitch(T,\lambda)$.  By Lemma~\ref{lem:innerEdgesinLRT}, Condition~(a)
    is satisfied.  We continue with showing that Condition~(b) is
    satisfied. In particular, we show first that every inner vertex
    $v\in V^0(T)$ is incident to at least one relevantly-labeled outer
    0-edge.  To this end, assume, for contradiction, that $(T,\lambda)$
    contains an inner vertex $v\in V^0(T)$ for which this property is not
    satisfied.
    
    That is, $v$ is either (i) incident to 1-edges only (incl.\
    $\lambda((\parent_T(v),v))=1$ in case $v\neq \rho_T$ by Condition~(a)) or
    (ii) there is an outer 0-edge $(v,x)$ that is not relevantly-labeled.  In
    Case (i), we put $\lambda'=\lambda$.  In Case (ii), we obtain a new
    labeling $\lambda'$ by changing the label of every outer 0-edge $(v,x)$
    with $x\in \child_T(v) \cap L(T)$ to ``1'' while keeping the labels of
    all other edges.  This does not affect the Fitch graph, since every such
    0-edge is not relevantly-labeled, and thus, $zx\in E(\gfitch(T,\lambda))$
    for all $z\in L(T)\setminus \{x\}$ by Lemma~\ref{lem:rel-label}.  Hence,
    for both Cases (i) and (ii), for the labeling $\lambda'$ \emph{all} outer
    edges $(v,x)$ with $x\in \child(v)\cap L(T)$ are labeled as 1-edges, $v$
    is incident to 1-edges only (by Condition~(a)) and
    $\gfitch(T,\lambda) = \gfitch(T,\lambda')$.  We thus have
    $xy\in E(\gfitch(T,\lambda')) =E(\gfitch(T,\lambda))$ for all
    $x\in L(T(v))$ and $y\in L(T)\setminus L(T(v))$.  Now, if $v\neq \rho_T$
    let $e=(u\coloneqq\parent_T(v),v)$.  Otherwise, if $v=\rho_T$ then let
    $e=(v,u)$ for some inner vertex $u\in \child_T(v)$.  Note, such an inner
    edge $(\rho_T,u)$ exists since $G$ contains at least two independent sets
    of size at least two and $T$ is not a star tree as shown above.  Now
    consider the tree $(T_{/e},\lambda'_{/e})$, and denote by $w$ the vertex
    obtained by contraction of the inner edge $e$. By construction, every path
    in $T_{/e}$ connecting any $x\in L(T(v))$ and
    $y\in L(T)\setminus L(T(v))$ must contain some 1-edge $(w,w')$ with
    $w'\in\child_{T_{/e}}(w)=\child_{T}(v)$ implying
    $xy\in E(\gfitch(T_{/e},\lambda'_{/e}))$.  Moreover, the edge contraction
    does not affect whether or not the path between any vertices within
    $L(T(v))$ or within $L(T)\setminus L(T(v))$ contains a 1-edge.  Hence,
    $\gfitch(T,\lambda) = \gfitch(T,\lambda') =
    \gfitch(T_{/e},\lambda'_{/e})$, and $(T,\lambda)$ is not
    Fitch-least-resolved; a contradiction.  In summary, every inner vertex
    $v$ must be incident to at least one relevantly-labeled outer 0-edge
    $(v,x)$. By Lemma~\ref{lem:rel-label}, $(v,x)$ is a relevantly-labeled
    outer 0-edge if and only if there is a vertex $z\in L(T)\setminus \{x\}$
    such that $zx\notin E(\gfitch(T,\lambda))$. By Condition (a), all inner
    edges in $(T,\lambda)$ are 1-edges, and thus, there is only one place
    where the leaf $z$ can be located in $T$, namely as a leaf adjacent to
    $v$. In particular, the outer edge $(v,z)$ is a relevantly-labeled
    0-edge, since $zx\notin E(\gfitch(T,\lambda))$. Therefore, Condition (b)
    is satisfied for every inner vertex $v$ of $T$.
    
    The latter arguments also show that all distinct vertices $x,y\in L(T)$
    that are contained in the same independent set must have the same parent.
    Clearly, $(\parent(x), x)$, $(\parent(x), y)$ must be outer 0-edges,
    since otherwise $xy\in E(\gfitch(T,\lambda))$. Hence, the final statement
    of the theorem is satisfied.
    
    Now let $(T,\lambda)$ be such that Conditions~(a) and~(b) are satisfied.
    First observe that none of the outer edges can be contracted without
    changing $L(T)$. Now let $e = (u,v)$ be an inner edge.  By Condition (a),
    $e$ is a 1-edge. Moreover, by Condition (b), vertex $u$ and $v$ are both
    incident to at least two relevantly-labeled outer 0-edges. Hence, there
    are outer 0-edges $(u,x),(u,x'),(v,y),(v,y')$ with pairwise distinct
    leaves $x,x',y,y'$ in $T$.  Since $(u,v)$ is a 1-edge, we have
    $xy,xy',x'y,x'y' \in E(\gfitch(T,\lambda))$.  Moreover, we have
    $xx',yy'\notin E(\gfitch(T,\lambda))$.  Now consider the tree
    $(T_{/e}, \lambda')$ with an arbitrary labeling $\lambda'$ and denote by
    $w$ the vertex obtained by contraction of the inner edge $(u,v)$.  In
    this tree, $x,x',y,y'$ all have the same parent $w$.  If
    $\lambda'((w,x))=1$ \emph{or} $\lambda'((w,y))=1$, we have
    $xx'\in\gfitch(T_{/e}, \lambda')$ or
    $yy'\in E(\gfitch(T_{/e}, \lambda'))$, respectively. If
    $\lambda'((w,x))=0$ \emph{and} $\lambda'((w,y))=0$, we have
    $xy\notin E(\gfitch(T_{/e}, \lambda'))$.  Hence, it holds
    $\gfitch(T_{/e}, \lambda')\ne \gfitch(T, \lambda)$ in both cases.  Since
    the inner edge $e$ and $\lambda'$ were chosen arbitrarily, we can apply
    Lemma~\ref{lem:no-cont} to conclude that $(T,\lambda)$ is
    Fitch-least-resolved.  
  \end{proof}
  
  As a consequence of Thm.~\ref{thm:LRT-rsFitch}, Fitch-least-resolved trees
  can be constructed in polynomial time. To be more precise, if a Fitch graph
  $G$ contains only independent sets of size one (except possibly for one),
  we can construct a star tree $T$ with edge labeling $\lambda$ as specified
  in the proof of Thm.~\ref{thm:LRT-rsFitch} to obtain the 0/1-edge labeled
  tree $(T,\lambda)$ that is Fitch-least-resolved w.r.t.\ $G$.  This
  construction can be done in $O(|V(G)|)$ time.
  
  Now, assume that $G$ has at least two independent sets of size at least
  two.  Let $\mathcal{I}$ be the set of independent sets of $G$ and
  $I_1,\dots,I_k\in \mathcal{I}$, $k\geq 2$ be all independent sets of size
  at least two.  We now construct a tree $(T,\lambda)$ with root $\rho_T$ as
  follows: First we add $k$ vertices $v_1 = \rho_T$ and $v_2,\dots,v_{k}$,
  and add inner edges $e_i=(v_i,v_{i+1})$ with label $\lambda(e_i)=1$,
  $1\leq i\leq k-1$.  Each vertex $v_i$ gets as children the leaves in $I_i$,
  $1\leq i\leq k$ and all these additional outer edges obtain label ``0''.
  Finally, all elements in the remaining independent sets
  $\mathcal{I}\setminus \{I_1,\dots,I_k\}$ are of size one and are connected
  as leaves via outer 1-edges to the root $v_1=\rho_T$.  It is an easy
  exercise to verify that $T$ is a phylogenetic tree and that
  $\gfitch(T,\lambda)=G$.  In particular, Thm.~\ref{thm:LRT-rsFitch} implies
  that $(T,\lambda)$ is Fitch-least-resolved w.r.t.\ $G$.  This construction
  can be done in $O(|V(G)|)$ time. We summarize this discussion as
  \begin{proposition}\label{prop:FLRT-polytime}
    For a given Fitch graph $G$, a Fitch-least-resolved tree can be
    constructed in $O(|V(G)|)$ time.
  \end{proposition}
  
  Fitch-least-resolved trees, however, are only of very limited use for the
  construction of relaxed scenarios $\scen=(T,S,\sigma,\mu,\tT,\tS)$ from an
  underlying Fitch graph.  First note that we would need to consider
  \emph{planted versions} of Fitch-least-resolved trees, i.e.,
  Fitch-least-resolved trees to which a planted root is added, since
  otherwise, such trees cannot be part of an explaining scenario, which is
  defined in terms of planted trees.  Even though $(G,\sigma)$ is an
  rs-Fitch graph, Example~\ref{ex:FLRT-noScen} shows that it is possible that
  there is no relaxed scenario $\scen=(T,S,\sigma,\mu,\tT,\tS)$ with
  HGT-labeling $\lambda_{\scen}$ such that
  $(T,\lambda) = (T,\lambda_{\scen})$ for the planted version
  $(T,\lambda)$ of \emph{any} of its Fitch-least-resolved trees.
  
  \begin{figure}[t]
    \begin{center}
      \includegraphics[width=0.85\textwidth]{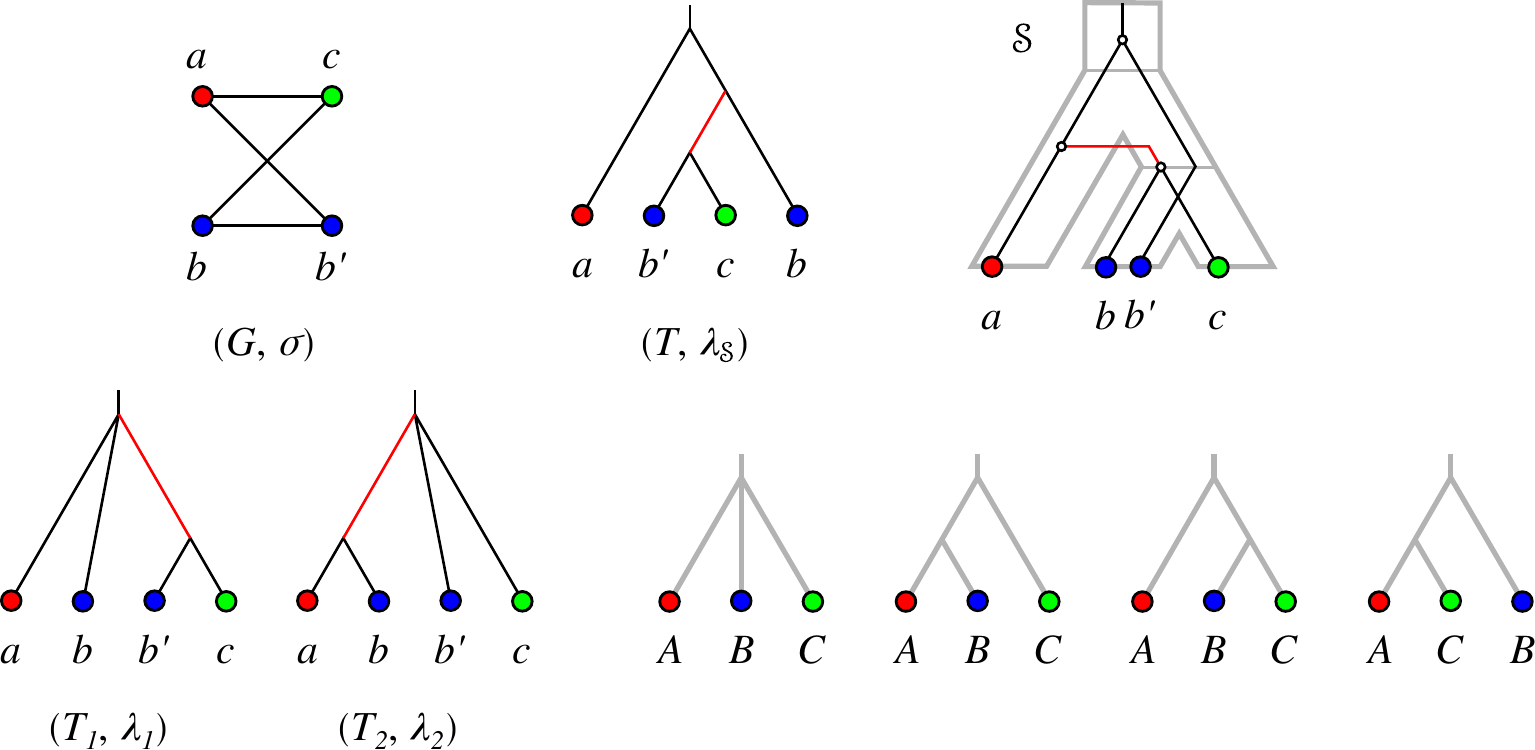}
    \end{center}
    \caption{An rs-Fitch graph $(G,\sigma)$ and a possible  relaxed 
      scenario $\scen=(T,S,\sigma,\mu,\tT,\tS)$ with $G =
      \gfitch(T,\lambda_{\scen})$. For the planted versions 
      $(T_1,\lambda_1)$ and $(T_2,\lambda_2)$ of the Fitch-least-resolved trees
      of $(G,\sigma)$ there is no
      relaxed  scenario $\scen$ such that $(T_i,\lambda_i) = 
      (T_i,\lambda_{\scen})$,
      $i\in \{1,2\}$. Red edges indicate 1-labeled (i.e., transfer)
      edges. See Example~\ref{ex:FLRT-noScen} for further details.}
    \label{fig:FLRT-noScen}
  \end{figure}
  
  \begin{xmpl}\label{ex:FLRT-noScen}
    Consider the rs-Fitch graph $(G,\sigma)$ with $V(G)=\{a,b,b',c\}$,
    $E(G)=\{ab',ac,bb',bc\}$ and surjective coloring $\sigma$ such that
    $\sigma(a)=A$, $\sigma(b)=\sigma(b')=B$, $\sigma(c)=C$ and $A,B,C$ are
    pairwise distinct. The rs-Fitch graph $(G,\sigma)$, a Fitch tree
    $(T,\lambda)$ and  relaxed  scenario $\scen$ with
    $(T,\lambda) = (T,\lambda_{\scen})$ as well as the planted versions 
    $(T_1,\lambda_1)$ and $(T_2,\lambda_2)$ of its two
    Fitch-least-resolved trees are
    shown in Fig.~\ref{fig:FLRT-noScen}.
    
    Fitch-least-resolved trees for $(G,\sigma)$ must contain an inner 1-edge,
    since $G$ has two independent sets of size two and by
    Thm.~\ref{thm:LRT-rsFitch}.  Thus, it is easy to verify that
    there are no other
    Fitch-least-resolved trees for $(G,\sigma)$.
    
    By Lemma~\ref{lem:independent-lca}, we obtain
    $\lca_S(A,B) \preceq_S \mu(\lca_{T_i}(a,b))$ and
    $\lca_S(B,C) \preceq_S \mu(\lca_{T_i}(b',c))$, $i\in\{1,2\}$, for both
    (planted versions of the)
    Fitch-least-resolved trees.  However, for all of the possible species
    trees on three leaves $A,B,C$, this implies that the images
    $\mu(\lca_{T_i}(a,b))$ and $\mu(\lca_{T_i}(b',c))$ are the single inner
    edge or the edge $(0_T,\rho_T)$ in $S$. Therefore, $\mu(\lca_{T_i}(a,b))$
    and $\mu(\lca_{T_i}(b',c))$ are always comparable in $S$. Hence, for all
    possible relaxed scenarios $\scen$, we have $\lambda_{\scen}(e)=0$ for
    the single inner edge $e$, whereas $\lambda_i(e)=1$ in $T_i$,
    $i\in \{1,2\}$.  This implies that there is no relaxed scenario $\scen$
    with $(T_i,\lambda_i) = (T_i,\lambda_{\scen})$, $i\in \{1,2\}$.
  \end{xmpl}		
  
  \newcommand{\DX}{\displaystyle}
  
  \section{Editing Problems}
  \label{app:edit}

  \subsection{Editing Colored Graphs to LDT Graphs and Fitch Graphs}
  
  We consider the following two edge modification problems for completion, 
  deletion, and editing.
  
  \begin{problem}[\PROBLEM{LDT-Graph-Modification (LDT-M)}]\ \\
    \begin{tabular}{ll}
      \emph{Input:}    & A colored graph $(G =(V,E),\sigma)$
      and an integer $k$.\\
      \emph{Question:} & Is there a subset $F\subseteq E$ such that $|F|\leq
      k$ and $(G'=(V,E\star F),\sigma)$  \\ &is an LDT graph  
      where $\star\in \{\setminus, \cup, \Delta\}$?
    \end{tabular}
  \end{problem}
  
  \begin{problem}[\PROBLEM{rs-Fitch Graph-Completion/Editing (rsF-D/E)}]\ \\
    \begin{tabular}{ll}
      \emph{Input:}    & A colored graph $(G =(V,E),\sigma)$
      and an integer $k$.\\
      \emph{Question:} & Is there a subset $F\subseteq E$ such that $|F|\leq
      k$ and $(G'=(V,E\star F),\sigma)$  \\ &is an rs-Fitch 
      graph  where $\star\in \{\setminus, \cup, \Delta\}$?
    \end{tabular}
  \end{problem}
  NP-completeness of \PROBLEM{LDT-M} be shown by reduction from 
  \begin{problem}[\PROBLEM{Maximum Rooted Triple Compatibility
      (MaxRTC)}]\ \\
    \begin{tabular}{ll}
      \emph{Input:}    & A set of (rooted) triples $\mathscr{R}$
      and an integer $k$.\\
      \emph{Question:} & Is there a compatible subset $\mathscr{R}^*\subseteq 
      \mathscr{R}$ such that $|\mathscr{R}^*|\geq |\mathscr{R}|-k$?
    \end{tabular}
  \end{problem}
  
  \begin{theorem}{\cite[Thm.~1]{Jansson:01}}
    \PROBLEM{MaxRTC} is NP-complete.
  \end{theorem}
  
  \begin{theorem}
    \PROBLEM{LDT-M} is NP-complete. 
    \label{thm:LDT-M-NP}
  \end{theorem}
  \begin{proof} 
    Since LDT graphs can be recognized in polynomial time (cf.\
    Cor.~\ref{cor:LDTpoly}), a given solution can be verified in polynomial
    time. Thus, \PROBLEM{LDT-M} is contained in NP.
    
    We now show NP-hardness by reduction from \PROBLEM{MaxRTC}.  Let
    $(\mathscr{R},k)$ be an instance of this problem, i.e., $\mathscr{R}$ is
    a set of triples and $k$ is a non-negative integer.  We construct a
    colored graph $(G_\mathscr{R}=(L,E),\sigma)$ as follows: For each triple
    $r_i = xy|z\in \mathscr{R}$, we add three vertices $x_i,y_i,z_i$, two
    edges $x_iz_i$ and $y_iz_i$, and put $\sigma(x_i) = x$, $\sigma(y_i) = y$
    and $\sigma(z_i) = z$.  Hence, $(G_\mathscr{R},\sigma)$ is properly
    colored and the disjoint union of paths on three vertices $P_3$.  In
    particular, therefore, $(G_\mathscr{R},\sigma)$ does not contain an
    induced $P_4$, and is therefore a properly colored cograph (cf.\
    Prop.~\ref{prop:cograph}).  By definition and construction, we have
    $\mathscr{R} = \Sri(G_\mathscr{R},\sigma)$.
    
    First assume that \PROBLEM{MaxRTC} with input $(\mathscr{R}, k)$ has a
    yes-answer. In this case let $\mathscr{R}^*\subseteq \mathscr{R}$ be a
    compatible subset such that $|\mathscr{R}^*| \geq |\mathscr{R}| - k$.
    For each of the triples $r_i= xy|z\in \mathscr{R}\setminus\mathscr{R}^*$,
    we add the edge $x_iy_i$ to $G_\mathscr{R}$ or remove the edge $x_iz_i$
    from $G_\mathscr{R}$ for \PROBLEM{LDT-E/C} and \PROBLEM{LDT-D},
    respectively, to obtain the graph $G^*$.  In both cases, we eliminate the
    corresponding triple $xy|z$ from $\Sri(G^*,\sigma)$.  By construction,
    therefore, we observe that $\Sri(G^*,\sigma) = \mathscr{R}^*$ is
    compatible.  Moreover, since we have never added edges between distinct
    $P_3$s, all connected components of $G^*$ are of size at most
    three. Therefore, $G^*$ does not contain an induced $P_4$, and thus
    remains a cograph. By Thm.~\ref{thm:characterization}, the latter
    arguments imply that $(G^*,\sigma)$ is an LDT graph. Since $(G^*,\sigma)$
    was obtained from $(G_\mathscr{R},\sigma)$ by using
    $|\mathscr{R}\setminus\mathscr{R}^*| \leq k$ edge modifications, we
    conclude that \PROBLEM{LDT-M} with input $(G_\mathscr{R},\sigma, k)$ has
    a yes-answer.
    
    For the converse, suppose that \PROBLEM{LDT-M} with input
    $(G_\mathscr{R},\sigma, k)$ has a yes-answer with a solution
    $(G^* = (L,E\star F),\sigma)$, i.e., $(G^*,\sigma)$ is an LDT graph and
    $|F|\le k$.  By Thm.~\ref{thm:characterization}, $\Sri(G^*,\sigma)$ is
    compatible.  Let $\mathscr{R}^*$ be the subset of
    $\mathscr{R} = \Sri(G_\mathscr{R},\sigma)$ containing all triples of
    $\mathscr{R}$ for which the corresponding induced $P_3$ in
    $G_\mathscr{R}$ remains unmodified and thus, is still an induced $P_3$ in
    $G^*$.  By construction, we have
    $\mathscr{R}^*\subseteq \Sri(G^*,\sigma)$.  Hence, $\mathscr{R}^*$ is
    compatible.  Moreover, since $|F|\le k$, at most $k$ of the
    vertex-disjoint $P_3$s have been modified. Therefore, we conclude that
    $|\mathscr{R}^*|\ge |\mathscr{R}|-k$.
    
    In summary, \PROBLEM{LDT-M} is NP-hard.  
  \end{proof}
  
  \begin{theorem}
    \PROBLEM{rsF-C} and \PROBLEM{rsF-E} are NP-complete.
    \label{thm:rsF-M-NP}
  \end{theorem}
  \begin{proof}
    Since rs-Fitch graphs can be recognized in polynomial time, 
    a given solution can be verified as being a yes- or no-answer
    in polynomial time. Thus, \PROBLEM{rsF-C/E}$\in NP$.
    
    Consider an arbitrary graph $G$ and an integer $k$. We construct an
    instance $(G,\sigma,k)$ of \PROBLEM{rsF-C/E} by coloring all vertices
    distinctly.  Then condition (ii) in Thm.~\ref{thm:char-rsFitch} is always
    satisfied.  To see this, we note that for $k>1$ there are no edges
    between colors in the auxiliary graph $\auxfitch(\sigma,\mathcal{I})$
    such that their corresponding unique vertices are in distinct independent
    sets $I, I'\in \mathcal{I}$.  The problem therefore reduces to
    completion/editing of $(G,\sigma)$ to a complete multipartite graph,
    which is equivalent to a complementary deletion/editing of the complement
    of $(G,k)$ to a disjoint union of cliques, i.e., a cluster graph. Both
    \PROBLEM{Cluster Deletion} and \PROBLEM{Cluster Editing} are NP-hard
    \cite{Shamir:04}.  
  \end{proof}
  
  Although \PROBLEM{Cluster Completion} is polynomial (it is solved by
  computing the transitive closure), \PROBLEM{rsF-D} remains open: Consider
  a colored complete multipartite graph $(G,\sigma)$ that is not an
  rs-Fitch graph. Then solving \PROBLEM{Cluster Completion} on the
  complement returns $(G,\sigma)$, which by construction is not a solution
  to \PROBLEM{rsF-D}.
  
  \subsection{Editing LDT Graphs to Fitch Graphs}

  \begin{lemma}
    There is a linear-time algorithm to solve Problem \ref{problem:Fcomp}
    for every cograph $G$.
    \label{lem:editing}
  \end{lemma}
  \begin{proof}
    Instead of inserting in the cograph $G$ the minimum number of edges
    necessary to reach a complete multipartite graph, we consider the
    equivalent problem of \emph{deleting} a minimal set $Q$ of edges from its
    complement $\overline{G}$, which is also a cograph, to obtain the
    complement of a complete multipartite graph, i.e., the disjoint union of
    complete graphs.  This problem is known as the \textsc{Cluster Deletion}
    problem \cite{Shamir:04}, which is known to have an polynomial-time
    solution for cographs \cite{Gao:13}: A greedy maximum clique partition of
    $G$ is obtained by recursively removing a maximum clique $K$ from $G$,
    see also \cite{Dessmark:07}. For cographs, the greedy maximum clique
    partitions are the solutions of the \textsc{Cluster Deletion} problem
    \cite[Thm.~1]{Gao:13}. The \textsc{Maximum Clique} problem on cographs
    can be solved in linear time using the co-tree of $G$ \cite{Corneil:81},
    which can also be obtained in linear time \cite{Corneil:81}.  
  \end{proof}
  
  An efficient algorithm to solve the \textsc{Cluster Deletion} problem for
  cographs can be devised by making use of the recursive construction of a
  cograph along its discriminating cotree $(T,t)$. For all $u\in V(T)$, we 
  have
  \begin{equation*}
    G[u] = \begin{cases}
      \DX\bigcupdot_{v\in\child(u)} G[v] & \text{ if } t(u)=0 \\
      \DX\bigjoin_{v\in\child(u)} G[v]   & \text{ if } t(u)=1 \\
      \DX (\{u\},\emptyset)             & \text{ if } u \text{ is a leaf } 
    \end{cases}
  \end{equation*}
  Denote by $\mathscr{P}(u)$ the optimal clique partition of the cograph
  implied by the subtree $T(u)$ of the discriminating cotree $(T,t)$. We
  think of $\mathscr{P}(u) := [Q_1(u),Q_2(u),\dots]$ as an ordered list, such
  that $|Q_i(u)|\ge |Q_j(u)|$ if $i<j$. It will be convenient to assume that
  the list contains an arbitrary number of empty sets acting as an identity
  element for the join and disjoint union operation. With this convention,
  the optimal clique partitions $\mathscr{P}(u)$ satisfy the recursion:
  \begin{equation*}
    \mathscr{P}(u) = \begin{cases}
      \DX\bigcup_{v\in\child(u)} \mathscr{P}(v) & \text{ if } t(u)=0 \\
      \DX\left[ \bigcup_{v\in\child(u)} Q_i(v) \quad \Big| \; i=1,2,\dots 
      \right]
      & \text{ if } t(u)=1 \\
      \DX[\{u\},\emptyset,\dots] & \text{ if } u \text{ is a leaf } 
    \end{cases}
  \end{equation*}
  In the first case, where $t(u)=0$, we assume that the union 
  operation to obtain $\mathscr{P}(u) = [Q_1(u),Q_2(u),\dots]$ maintains the 
  property $|Q_i(u)|\ge |Q_j(u)|$ if $i<j$. In an implementation, this can 
  e.g.\ 
  be achieved using $k$-way merging where $k=|\child(u)|$.
  
  To see that the recursion is correct, it suffices to recall that the greedy
  clique partition is optimal for cographs as input \cite{Gao:13} and to
  observe the following simple properties of cliques in cographs
  \cite{Corneil:81}: (i) a largest clique in a disjoint union of graphs is
  also a largest clique in any of its components. The optimal clique
  partition of a disjoint union of graphs is, therefore, the union of the
  optimal clique partitions of the constituent connected components. (ii)
  For a join of two or more graphs $G_i$, each maximum size clique $Q$ is the
  join of a maximum size clique of each constituent. The next largest clique
  disjoint from $Q=\bigjoin_i Q_i$ is, thus, the join of a largest cliques
  disjoint from $Q_i$ in each constituent graph $G_i$. Thus a greedy clique
  partition of $G$ is obtained by size ordering the clique partitions of
  $G_i$ and joining the $k$-largest cliques from each.
  
  The recursive construction of $\mathscr{P}(\rho_T)$ operates directly on
  the discriminating cotree $(T,t)$ of the cograph $G$.  For each node $u$,
  the effort is proportional to $|L(T(u))| \log(\deg(u))$ for the
  $\deg(u)$-wise merge sort step if $t(u)=0$ and proportional to $|L(T(u))|$ for
  the merging of the $k$-th largest clusters for $t(u)=1$. Using
  $\sum_u \deg(u)|L(T(u))|\le |L(T)|\sum_u \deg(u)\le |L(T)|2|E(T)|$ together
  with $|E(T)|=|V(T)|-1$ and $|V(T)|\leq 2 |L(T)|-1$ (cf.\ \cite[Lemma
  1]{Hellmuth:15a}), we obtain
  $\sum_u \deg(u)|L(T(u))| \in \mathcal{O}(|L(T)|^2) =
  \mathcal{O}(|V(G)|^2)$, that is, a quadratic upper bound on the running
  time.
  
\end{appendix}

\bibliographystyle{plainnat}
\bibliography{preprint2-Rbelow}

\end{document}